\documentclass[11pt]{article} 

\usepackage{fullpage}
\usepackage{graphicx}
\usepackage{upref}
\usepackage{enumerate}
\usepackage{qcircuit}

\usepackage{latexsym}
\usepackage{pdfpages}
\usepackage[bookmarks]{hyperref}
\usepackage{color,graphics}
\usepackage{comment} 
\usepackage{caption}
\usepackage{subcaption}
\usepackage{wrapfig}
\usepackage{braket}
\usepackage{amssymb}
\usepackage{amsmath}
\usepackage{amsthm}
\usepackage{mathrsfs}
\usepackage{mathtools}

\usepackage{epsfig,graphicx}
\usepackage{algorithmic}
\usepackage[ruled,linesnumbered]{algorithm2e}
\usepackage{amsfonts}
\usepackage{tikz} 
\usepackage{varioref}
\usepackage[ansinew]{inputenc}
\usepackage{qcircuit}
\usepackage{authblk}
\usepackage{multirow}
\usepackage{lineno}

\theoremstyle{plain}
\newtheorem{theorem}{Theorem}[section]
\theoremstyle{plain}
\newtheorem{lemma}{Lemma}[section]
\theoremstyle{plain}

\theoremstyle{plain}

\theoremstyle{plain}

\theoremstyle{plain}

\theoremstyle{definition}

\theoremstyle{definition}

\theoremstyle{remark}

\theoremstyle{definition}

\AtBeginDocument{}

\newcommand{\E}{\mathbb{E}} 	
\DeclareMathOperator{\real}{\mathbb{R}}

\newcommand{\intg}{\mathbb{Z}}

\newcommand{\tr}{\text{Tr}}
\newcommand{\id}{\mathbb{I}}

\newcommand{\cliff}{\mathcal{C}}
\newcommand{\pauli}{\mathcal{P}}
\newcommand{\clifft}{\mathcal{J}}

\newcommand{\X}{\text{X}}
\newcommand{\Y}{\text{Y}}
\newcommand{\Z}{\text{Z}}
\newcommand{\had}{\text{H}}
\newcommand{\T}{\text{T}}
\newcommand{\CNOT}{\text{CNOT}}
\newcommand{\phase}{\text{S}}

\newcommand{\vect}[1]{\mathbf{#1}}
\newcommand{\conj}[1]{\overline{#1}}

\newcommand{\chan}{\mathcal{E}}
\newcommand{\chanu}{\mathcal{U}}
\newcommand{\liou}{\mathcal{L}}
\newcommand{\cost}{\mathcal{C}}
\newcommand{\var}{\text{Var}}

\begin{document}

\title{Synthesizing efficient circuits for Hamiltonian simulation}

\author[1,2]{Priyanka Mukhopadhyay \thanks{mukhopadhyay.priyanka@gmail.com, p3mukhop@uwaterloo.ca (Corresponding author)}}
\author[3,4]{Nathan Wiebe \thanks{nawiebe@cs.toronto.edu}}
\author[5]{Hong Tao Zhang \thanks{ht.zhang@mail.utoronto.ca}}

\affil[1]{Institute for Quantum Computing, University of Waterloo, Canada}
\affil[2]{Department of Combinatorics and Optimization, University of Waterloo, Canada}
\affil[3]{Department of Computer Science, University of Toronto, Canada}
\affil[4]{Pacific Northwest National Laboratory, USA}
\affil[5]{Department of Mathematics, University of Toronto, Canada}

\date{}

\maketitle


\begin{abstract}
We provide a new approach for compiling quantum simulation circuits that appear in Trotter, qDRIFT and multi-product formulas to Clifford and non-Clifford operations that can reduce the number of non-Clifford operations by a factor of up to $4$. In fact, the total number of gates reduce in many cases. We show that it is possible to implement an exponentiated sum of commuting Paulis with at most $m$ (controlled)-rotation gates, where $m$ is the number of distinct non-zero eigenvalues (ignoring sign). Thus we can collect mutually commuting Hamiltonian terms into groups that satisfy one of several symmetries identified in this work which allow an inexpensive simulation of the entire group of terms.  We further show that the cost can in some cases be reduced by partially allocating Hamiltonian terms to several groups and provide a polynomial time classical algorithm that can greedily allocate the terms to appropriate groupings.  We further specifically discuss these optimizations for the case of fermionic dynamics and provide extensive numerical simulations for qDRIFT of our grouping strategy to 6 and 4-qubit Heisenberg models, $LiH$, $H_2$ and observe a factor of 1.8-3.2 reduction in the number of non-Clifford gates.  This suggests Trotter-based simulation of chemistry in second quantization may be even more practical than previously believed.
\end{abstract}

\section{Introduction}
\label{sec:intro}

One main reason that led Feynman \cite{1982_F} and others to propose the idea of quantum computers was the fact that problems like simulating the dynamics of quantum systems are intractable on a classical computer. Starting from the seminal work of Lloyd \cite{1996_L}, much research \cite{2018_CMNetal} has been done to develop algorithms for simulating Hamiltonians, culminating in various techniques like product formulas \cite{1991_S, 1959_T}, quantum walks \cite{2009_BC}, linear combination of unitaries \cite{2012_CW}, truncated Taylor series \cite{2015_BCCetal}, and quantum signal processing \cite{2017_LC}. Special techniques have been developed for simulating particular physical systems \cite{2018_BWMetal, 2019_CROetal, 2012_JLP, 2010_LWGetal, 2020_MEAetal, 2015_PHWetal, 2014_WBCetal, 2022_KL}, which might find applications in developing new pharmaceuticals, catalysts and materials. Phase estimation can be combined with quantum simulation to find the ground state energy \cite{1999_AL} and excited state energies \cite{2014_PMSetal, 2018_WHB, 2018_OTT} of the Hamiltonian. This is called the electronic structure problem \cite{2020_MEAetal}, which is important in chemistry and material science. Research in quantum simulation has also inspired the development of quantum algorithms for various other problems \cite{2014_B, 2017_BS, 2003_CCDetal, 2008_FGG, 2009_HHL}.

One main challenge for digital quantum simulation is the implementation with efficient circuits that can produce reliable results. Without it, a theoretical exponential speedup may not lead to a useful algorithm if a typical practical application requires an amount of time and memory that is beyond the reach of even a quantum computer. There are a number of factors that can affect the efficiency of a quantum circuit i.e. its running time and error, for example, the number of qubits, depth, gate count, etc. So depending upon the applications or other hardware constraints, one can design algorithms that optimize or reduce the count/depth of one particular type of quantum gate or other resources. For example, there are algorithms that does T-count and T-depth-optimal synthesis \cite{2021_MM, 2021_GMM, 2021_GMM2} given a unitary or does re-synthesis of a given circuit with reduced T-count, T-depth \cite{2014_AMM, 2020_DKPW, 2020_HS} or CNOT-count \cite{2008_PMH, 2018_AAM, 2020_GLMM}. The non-Clifford T gate has known constructions in most of the error correction schemes and the cost of fault-tolerantly implementing it exceeds the cost of the Clifford group gates by as much as a factor of hundred or more \cite{2009_FSG, 2006_AGP, 2012_FMMetal}. Quantum error correction and fault tolerance is especially significant for large quantum circuits, else the accumulation of errors will make any output highly unreliable and hence useless.
The minimum number of T-gates required to implement certain unitaries is a quantifier of difficulty in many algorithms \cite{2016_BG, 2016_BSS} that try to classically simulate quantum computation. So, even though alternative fault-tolerance methods such as completely transversal Clifford+T scheme \cite{2013_PR} and anyonic quantum computing \cite{2003_K} are also being explored, minimization of the number of T gates in quantum circuits remain an important and widely studied goal. Multi-qubit gates like CNOT introduce more error than single qubit gates, so reducing CNOT gate is important and especially relevant for the noisy intermediate scale quantum (NISQ) computers.

\paragraph{Our contributions : }
(I) One main result in this paper is Lemma \ref{lem:rot}, which shows that it is possible to implement an exponentiated sum of commuting Paulis with at most $m$ (controlled)-rotation gates, where $m$ is the number of distinct non-zero eigenvalues (ignoring sign). For illustration we consider the Hamiltonian for the Heisenberg model and we show that it is possible to achieve about $50\%$ reduction in the rotation gate cost and for certain underlying graphs this reduction can be about $75\%$. However, the cost of Toffolis may increase. We have given explicit circuits for 4-qubit and 6-qubit chain (or cycle), where we attempt to reduce both the rotation and Toffoli gate cost. 

(II) In most previous works, circuits for individual exponentiated Paulis are synthesized and combined. We show that it is possible to reduce the gate count (not only non-Clifford gates) if we instead consider groups of commuting Paulis. To give some practical demonstration we consider the qDRIFT Hamiltonian simulation algorithm \cite{2019_C}. We call the error introduced due to the algorithm as 'simulation error'. We take  the 1-D 4 qubit and 6 qubit Heisenberg Hamiltonians (Figure \ref{fig:10}) and also 4-qubit Hamiltonians for $H_2$ and $LiH$ (with freezing in the STO-3G basis) (Figure \ref{fig:11}), and compare the case where a single Pauli term is selected with the case where a set of commuting Pauli terms is selected for implementation at each iteration of qDRIFT. We observe that the error accumulation is less for multiple terms and also the rotation gate cost is less in this error regime. The number of Toffoli pairs is roughly equal to the number of $R_z/cR_z$ used, in case of multiple terms. So overall, we have less T-count when implementing multiple commuting Paulis per iteration of qDRIFT. This adds to the motivation of building efficient circuits for such Hamiltonians. 

(III) In Section \ref{sec:qChem} we derive explicit quantum circuits for the two-body excitation terms appearing in the Coulomb Hamiltonian in quantum chemistry. We mainly use the Clifford+T universal fault-tolerant gate set to implement unitaries. We design efficient circuits for different grouping of commuting Pauli operators. It is evident (Table \ref{tab:comp}) that the rotation gate cost depends on the coefficients of the Pauli summands. For some combination of coefficients the circuits derived here are optimal, in the sense, that they have the minimum (i.e. 1) number of $cR_z$ gates. Though our focus is on reducing the non-Clifford gate count, but most of the quantum circuits derived here have an overall reduced gate count, including reduction in the 2-qubit gates like CNOT. In Table \ref{tab:totalGate} we have compared the number of gates required to implement one of the Hamiltonians considered in this paper with a previous construction. For the remaining Hamiltonians we did not find any compact previous construction to compare with. In short, our approach can be useful not only in the fault-tolerant regime, but also in the NISQ era.  

(IV) In Algorithm 1 we describe a greedy method of grouping into commuting Paulis, but the objective is to optimize the number of non-Clifford gates. There have been a host of work that tackles the question of how to group the commuting Paulis and to the best of our knowledge most (if not all) of them has the objective to reduce the number of measurements required to make an estimation \cite{2021_HMRetal}. The latter problem is especially important for variational quantum eigensolvers. 
The grouping that optimizes the non-Clifford gates may not optimize the number of measurements. In most cases, finding the optimal grouping is difficult. But we can always ask the question that given a grouping (for whatever objective), is it possible to compile efficient circuits. In this case, we can use our techniques (Lemma \ref{lem:rot}) to reduce the gate count. Thus our methods can also be used to design circuits for the measurement problem.  

In this paper we use the Jordan-Wigner (JW) transformation \cite{1993_JW} to map from the fermionic to the qubit space. And then we group into commuting Paulis. Other transformations like Bravyi-Kitaev and parity transformations \cite{2002_BK} can also be used and may be beneficial in circumstances where Clifford operations are costly or inherent quantum error correction is desirable.  We focus on Jordan Wigner for two reasons.  First, in this paper we focus on the synthesis of efficient quantum circuits for exponentiated commuting Paulis and the techniques hold no matter whichever mapping is considered. Second, previous work has not shown obvious advantages for Bravyi--Kitaev transformations within the domain of fault-tolerant quantum computing.

\paragraph{How we compare the cost of non-Clifford resources : } In all the constructions discussed in this paper, two approximately implementable gates are used - $R_z$ and controlled $R_z$ ($cR_z$), whose T-count varies inversely with precision or synthesis error. From the results given in \cite{2021_GMM2} and from the implementations performed here until the error $10^{-6}$, we believe that T-count of $cR_z$ can be less than that of $R_z$ for most modestly small rotation angles. However, for convenience, we assume these have equal cost and with some abuse of terms, we refer to the T-count of $R_z/cR_z$ as the '(non-Clifford) rotation gate cost'. 
The only exactly implementable non-Clifford unitary/gate considered in the constructions is Toffoli with T-count 7 \cite{2021_MM} or 4 \cite{2013_J}. For low error regime, the T-count of approximately implementable $R_z/cR_z$ will dominate, while in high error regime the T-count of Toffoli may matter, if we use a lot of them. To reduce the T-count of compute-uncompute Toffoli pairs, we can use the temporary logical AND gadget, proposed by Gidney \cite{2018_G}. In fact, in our circuits, we use $R_z$ gates controlled on $n$ qubits ($n>1$), each of which can be decomposed into (compute-uncompute) pairs of NOT gates controlled on $n$ qubits and a $cR_Z$ gate. Each such multi-controlled NOT can be implemented with $n-1$ Toffoli. Alternatively, it can be implemented with $4n-4$ T-gates \cite{2018_G}. If we combine compute-uncompute pairs then the overall T-count of the circuit can reduce further, by using logical AND gadget. We must keep in mind that the implementations in \cite{2013_J, 2018_G} use classical resources and measurements, and it is not straightforward to argue that it will give advantage, inspite of using less number of T-gates. 
Alternatively, we can use the construction in \cite{2017_HLZetal} that implements an $n$-controlled NOT gate using $4n-4$ T, $4n-3$ CNOT and $n-1$ ancillae qubits. 
In our paper we have expressed the non-Clifford T-gate cost in terms of the rotation gate cost and the number of Toffoli pairs used.


\paragraph{Related work : }
In \cite{2017_RWSWT} the authors studied the non-Clifford resource cost required to simulate the chemical process of biological nitrogen fixation by nitrogenase. In \cite{2020_BT} the authors developed algorithms to synthesize circuits for the Clifford operators that diagonalize a group of commuting Paulis. The goal was to reduce the two-qubit CNOT gate count because of its low fidelity and limited qubit connectivity of near-term quantum computer architectures. Similar diagonalization algorithm has been used in \cite{2022_KF} for efficient simulation of Hamiltonian dynamics. Much work has been done for construction of quantum circuits for the evolution of molecular systems \cite{2011_WBA, 2014_WBCetal, 2018_BGSetal, 2018_KMWetal, 2019_GEBetal, 2020_GZBetal, 2020_LLAB, 2020_YAB} and Heisenberg model \cite{2021_GPAG}. 

\section{Results}
\label{sec:results}

\subsection{Notation}

 In many places we write $G_{(q)}$ to denote that the gate or operator $G$ acts on qubit $q$. For multi-qubit gates we write $CNOT_{(c,t)}$ to denote a CNOT with control at qubit $c$ and target at qubit $t$. For convenience, we have removed the parenthesis in the subscript whenever there is less ambiguity.
We write $[K]=\{1,2,\ldots,K\}$. We denote the $n\times n$ identity matrix by $\id_n$ or $\id$ if dimension is clear from the context. We denote the set of $n$-qubit unitaries by $\mathcal{U}_n$. The size of an $n$-qubit unitary is $N\times N$ where $N=2^n$. 
We have given detail description about the n-qubit Pauli operators ($\pauli_n$), Clifford group ($\cliff_n$) and the group ($\clifft_n$) generated by the Clifford and $\T$ gates in Supplementary Note 1 (Appendix \ref{app:prelim}).

\subsection{Optimizing Trotter-Decompositions}
\label{sec:Trotter}

The time evolution of a quantum system, described by a Hamiltonian $H$ is $e^{-iHt}$. Most often the Hamiltonian $H$ can be decomposed as the sum $H=\sum_{j=1}^m\alpha_jH_j$, where each $H_j$ is Hermitian.
There can be more than one decomposition of $H$ and we select the one such that for each $H_j$ the unitary $e^{-i\tau H_j}$ is efficiently implementable on a quantum computer, for any $\tau$. The goal of the Hamiltonian simulation problem is to find an approximation of $e^{-itH}$ into a sequence of $e^{-i\tau H_j}$, up to some desired precision. For example, using the Lie-Trotter formula \cite{1959_T} we have that
\begin{eqnarray}
 e^{-iHt}&=&\lim_{k\rightarrow\infty}\left(\prod_je^{-i(t/k)\alpha_jH_j}\right)^k. \nonumber
\end{eqnarray}
In the non-asymptotic regime, the Trotter scheme provides a first-order approximation, with the norm of the difference between the exact and approximate time evolution scaling as $O(t^2/k)$. More advanced higher order schemes \cite{1991_S, 2018_CMNetal} are also available. Alternatively, a randomized approach called qDRIFT can be used in place of a Trotter formula wherein the quantum state is evolved according to the probabilistic channel
\begin{equation}
    \rho \mapsto \sum_{j} \frac{\alpha_j}{|\alpha|_1}e^{-i |\alpha|_1 H_j t} \rho e^{i |\alpha|_1 H_j t}.
\end{equation}
Note the error here is also $O(t^2)$; however, in this case a single exponential is performed rather than $O(m)$ as would be needed for the comparable Trotter-formula.  The cost of such an approach scales as $O(|\alpha|_1^2 t^2/\epsilon)$ for error $\epsilon$ and does not directly depend on $m$.

The approximation errors arising in the use of product formulas are caused by non-commuting terms in the Hamiltonian. For example, see \cite{2021_CSTetal} for a detail exposition on Trotter errors. 
Given any set of mutually commuting operators $P_1,\ldots, P_m$ we have the following.
\begin{eqnarray}
 e^{-it\sum_{j=1}^mP_j}&=&\prod_{j=1}^me^{-itP_j} \label{eqn:commute}
\end{eqnarray}
Thus, the operators are partitioned into mutually commuting subsets. Time evolution for the sum of mutually commuting operators in each such subset is trivial, and the product formulas can be applied to the sum of Hamiltonians formed as the sum of each subset. This approach becomes especially applicable in scenarios where the Hamiltonian can be expressed as a sum of Pauli operators, for which the commutation relations can easily be evaluated. 

As a specific example, consider the case where $H= a Z\otimes Z\otimes Z$.  Since the Hamiltonian is diagonal, $e^{-iaZ\otimes Z\otimes Zt}$ has computational basis vector $\ket{b_1,b_2,b_3}$ and eigenvalues $e^{-i (-1)^{b_1\oplus b_2\oplus b_3} at}$.  Thus the eigenvalues are determined by the parity of the bit strings, which can be computed using CNOT gates.  From this reasoning the following quantum circuit will perform the simulation of this Pauli operator exactly.
\begin{equation}
\Qcircuit @C=1em @R=.7em {
& \ctrl{1}&\qw &\qw &\qw &\ctrl{1}&\qw\\
& \targ &\ctrl{1} &\qw &\ctrl{1} &\targ&\qw\\
& \qw &\targ &\gate{R_z(2at)}&\targ&\qw&\qw\\
}
\end{equation}
As every Pauli operator of weight $3$ can be diagonalized by Clifford conjugation, this circuit up to an elementary basis transformation, will simulate any weight $3$ Pauli Hamiltonian.  The exact same strategy of diagonalizing and simulating the Pauli operator in the eigenbasis shows that each exponential of a weight $\nu$ Pauli operator Hamiltonian requires $2(\nu-1)$ CNOT operators and one rotation gate.  This strategy is at the heart of most elementary networks for simulating chemistry and spin models~\cite{2011_WBA,2010_NC}.

The work of~\cite{2017_RWSWT} provided another way of thinking about these decompositions by showing an explicit method that can diagonalize sums of commuting operators that appear in chemistry simulations by transforming into a simultaneous eigenbasis of such terms.  In full generality, such transformations reduce the circuit depth but need not reduce the circuit size.  However, we will see here that for some Hamiltonians these transformations can reduce the circuit size as well.

As a motivating example, consider the Hamiltonian $H= XX + YY +ZZ$.  This Hamiltonian can be simulated, up to a global phase, by
\begin{equation}
\Qcircuit @C=1em @R=.7em {
& \ctrl{1} &\qw\\
& \ctrl{1} &\qw \\
\lstick{\ket{0}}& \gate{e^{i 4Z t}}&\qw
}
\end{equation}
This can be implemented using two Toffoli gates and a single qubit rotation.  In contrast, the standard approach from~\cite{2011_WBA,2010_NC} would use $3$ single qubit rotations and no Toffoli gates.  As rotation synthesis often is $10$ times more expensive than Toffoli gates~\cite{2017_RWSWT, 2021_MM, 2021_GMM2}, this will almost always be a favorable way of performing the simulation.  In contrast, if this symmetry is broken then the Hamiltonian term will be more expensive to simulate.  Thus it can be favorable to introduce such symmetries as needed artificially.  For example, consider
\begin{equation}
    H= XX + YY + (3/2) ZZ = (XX + YY + ZZ) + ZZ/2.
\end{equation}
Such a simulation can be performed using two rotation gates rather than the $3$ na\"ively needed and so it makes sense to compile the Hamiltonian terms this way to reduce the overall complexity. 


As another example, not all rotations are equally expensive and so we should also combine terms in such a way as to minimize the cost.  For example consider the time-evolution operator
\begin{equation}
    U(t) = e^{-i (\pi/4\sqrt{2} - \epsilon) Z -i(\pi/4\sqrt{2} + \epsilon) X} \approx e^{-i \pi/4 (X+Z)/\sqrt{2}} e^{ -i(X-Z)\epsilon}.
\end{equation}
While the first operation in this Trotterization is not a Clifford operation, it is a simulation of a Hadamard gate for time $\pi/4$.  As this corresponds to a special angle and since the Hadamard gate can be diagonalized using a constant size $H$ and $T$ circuit, the cost of implementing this first term is $O(1)$ and thus the dominant cost is the remaining rotation.  In contrast, if this property were not used then we would have two arbitrary rotations in the Trotterization which would be nearly twice the cost of this simplified approach.  These ideas can further be used in concert: remainder terms that arise from inexactly rounding a Hamiltonian evolution to a known cheap simulation can be absorbed into other terms or even other Trotter steps.

\begin{algorithm}[t!]
\SetAlgoLined
\KwData{$H= \sum_{j=1}^m \alpha_j P_j$ for distinct Pauli $P_j$ and $\alpha_j\ge 0$, cost function ${\rm COST}:\mathbb{R}^{m} \mapsto \mathbb{Z}^+$ such that ${\rm COST}(\vec{x})$ is the cost of simulating $H = \sum_j x_j P_j$, where $\vec{x}=(x_1,\ldots,x_m)$ and ${\rm COST}(\vec{x})=1$ if $\vec{x}$ contains at most $1$ non-zero entry.}
\KwOut{$\{H_j=h_j\sum_{i}P_i: j=1,\ldots,m'\}$ such that $H=\sum_{j=1}^{m'}H_j$.}

$\vec{\alpha'} \gets \vec{\alpha}$\;
\While{$\sum_j |\alpha_j'| >0$}{
$S\gets \{P_j:\alpha_j'\neq 0\}$\; 
$i\gets 0,\Gamma_{\max}\gets 0$\;
\While{$S\ne \emptyset$}{
$i\gets i+1$\;
$k_i\gets \text{argmax}(\{|\alpha'_i|| P_i \in S\})$\;
$\Gamma_{\max} \gets \max\left(\Gamma_{\max}, \text{max}_{\vec{\beta}}((\sum_j |\alpha'_j| -\sum_j |\alpha_j'-\beta_j|)/{\rm COST}(\vec{\beta})) \right)$ such that $\beta_j\ne 0 \Rightarrow \exists p\le i$ with $k_p = j$\;
$S\gets \{P_{\ell}\in S| [P_{\ell},P_{k_i}]=0 \text{ and } \alpha'_{\ell} \ne 0\}$\;
}
Store $\sum_{j=k_1}^{k_i}\Gamma_{\max}P_j$ as a Hamiltonian term\;
$\vec{\alpha'} \gets \vec{\alpha'} -\Gamma_{\max}\vec{1}$\;
}
\algorithmicreturn~Set of Hamiltonian terms computed. 
\caption{ Hamiltonian Compilation Using Greedy $1$-norm Minimization\label{alg:greedy}}
\end{algorithm}

We propose an algorithm in~Algorithm~\ref{alg:greedy} that exploits this intuition through a greedy decomposition of the Hamiltonian into sums of commuting terms.  These mutually commuting terms, or fragments, are chosen such that the ratio of the fraction of the Hamiltonian that is simulated by the term to the cost of the term is maximized.  This choice is motivated in part by the fact that the query complexity of a quantum simulation is lower bounded by $\Omega( |\alpha|_1 t)$~\cite{2007_BACetal} and thus designing circuits that simulate as large of a fraction of this one-norm as possible per quantum gate operation is a sensible optimization heuristic for our greedy algorithm.  Unlike traditional approaches to partitioning the Hamiltonian, our approach allows partial allocation of Hamiltonian terms to multiple commuting sets.  Further, the allocation can be negative in our approach.  This negative allocation is important because we will see that in some cases the introduction of more Hamiltonian weights on some terms can be more than offset by the reduced costs of simulating the fragment.  

The number of optimization steps required for our greedy algorithm is at most $O(m^2)$.  To see this, assume that the optimal strategy involves $\mu$ iterations of the outer loop for $\mu \in \Omega(m)$ and assume that the inner loop optimization requires $\nu$ iterations.  Since COST$\ge 1$ it holds that $\Gamma_{\max} \le \sum_j|\alpha'_j| -\sum_j |\alpha'_j-\beta_j|$.  Assume that $\sum_j|\alpha'_j| -\sum_j |\alpha'_j-\beta_j| < |\alpha'|_\infty$.  In this case, by assumption there exists a trivial solution that outperforms this where the largest term is simulated in isolation at cost $1$.  Therefore we must have that $\sum_j|\alpha'_j| -\sum_j |\alpha'_j-\beta_j| \ge |\alpha'|_\infty$.  Then from standard norm inequalities we have that  $|\alpha'|_\infty \ge |\alpha'|_1/m$.  Thus the one-norm of the vector is given by a first order difference equation of the form $|\alpha^{(j+1)}|_1 \le (1-1/m) |\alpha^{(j)}|_1$.  The general solution to this is $(1-1/m)^{j} |\alpha|_1$ which is $\epsilon$ for $j \in O(\log(1/\epsilon) / \log(1/(1-1/m))) \in O(m \log(1/\epsilon))$. This implies that $\mu \in O(m \log(1/\epsilon)$. Next $\nu$ is the maximum number of iterations for the inner loop.  Since each iteration continues until the total number of terms remaining is reduced by one we have that $\nu \in O(m)$. Thus the total number of iteration steps is $\mu \nu \in O(m^2 \log(1/\epsilon))$.  This shows that the algorithm scales polynomially with the number of terms if the optimization process is also efficient.

The cost of optimization can vary strongly depending on the continuity / convexity of the objective function and without making further assumptions we cannot assume that the optima over $\vec{\beta}$ can be found in polynomial time.  If we assume, however, that the optimizer works by considering one of a polynomial number of potential circuits for simulating the terms and then uses linear programming to find the optimal value of $\vec{\beta}$, we have that the optimization problem can be solved in polynomial time on a classical computer.  Such a choice corresponds exactly to the discussion in the next sections, where we propose the use of a discrete set of optimization strategies for simulating chemistry that can then be used within Algorithm~\ref{alg:greedy} to greedily find the best possible simulation circuit given these discrete set of optimizations for the value of $\vec{\beta}$ chosen.

\paragraph{Truncating Hamiltonian : }
We can terminate Algorithm~\ref{alg:greedy} before all terms are allocated i.e. we output $\{H_j=h_j\sum_iP_i : j=1,\ldots,m''\}$ such that $\sum_{j=1}^{m''}H_j=\tilde{H}\neq H$.  This leads to truncation errors in our simulation algorithm that will be present even if an algorithm such as qDRIFT is used for the simulation.  We show here that if we truncate some terms of the Hamiltonian, then the error incurred is at most twice the error incurred from the complete Hamiltonian simulation by qDRIFT, given that the distance of the truncated and given Hamiltonian is at most square root of the qDRIFT simulation error. We do this because in some cases we may be able to simulate the truncated Hamiltonian with less number of gates.

Suppose we write the given Hamiltonian as follows.
\begin{eqnarray}
 H&=&\sum_{j=1}^Mw_jH_j+\delta H=\tilde{H}+\delta H    \qquad [\|H\|\leq 1]
 \label{eqn:Htilde}
\end{eqnarray}
Here each $H_j$ is a Hermitian matrix for which an efficient simulation circuit exists. The protocol working with the truncated Hamiltonian $\tilde{H}$, samples each $H_j$ independently with probability $p_j=\frac{w_j}{\lambda}$ (where $\lambda=\sum_i|w_i|$), in each iteration.

The error per iteration of qDRIFT, i.e. $\epsilon_N$, is given by bounding the diamond distance between the channel $\chanu_N(\rho)$ corresponding to the Hamiltonian $H$ and the channel $\tilde{\chan}(\rho)$ implemented by the protocol.
\begin{lemma}
The error observed when there are $N$ time-steps taken using a qDRIFT channel, $\epsilon_N$, as quantified by the diamond distance as a function of the truncation error in the Hamiltonian $\delta$ is
$$
\epsilon_N\leq\|\tilde{\chan}(\rho)-\chanu_N(\rho)\|_{\diamond}\leq \epsilon_{qDRIFT}+2\delta\sqrt{\epsilon_{qDRIFT}}
$$
where $\epsilon_{qDRIFT}\lessapprox \frac{2\lambda^2t^2}{N^2}$ and $\lambda = \sum_i |w_i|$.
 \label{lem:truncErr}
\end{lemma}

The proof has been given in Supplementary Method 4 (Appendix \ref{app:error} : Lemma \ref{app:lem:truncErr}).
Thus the total error after all repetitions is as follows.
\begin{eqnarray}
 \epsilon &\leq& N\epsilon_N\lessapprox N\frac{2\lambda^2t^2}{N^2}+2\sqrt{2}\delta N\frac{\lambda t}{N}=\frac{2\lambda^2 t^2}{N}+2\sqrt{2}\delta\lambda t
\end{eqnarray}
This shows that if $\delta \in O(\sqrt{\epsilon_{qDRIFT}})$ then the asymptotic scaling is not impacted by the exclusion of the terms from the Hamiltonian.

\paragraph{Expected cost : } Let the cost of implementing the unitary $e^{itw_jL_j/N}$ be $c_j$. Cost can be defined in many ways, like total number of gates, number of non-Clifford gates like T or Toffoli gate, number of multi-qubit gates like CNOT, etc. In our paper we focus mainly on the number of non-Clifford gates. Let $\cost_N$ be the variable denoting the cost per repetition of our protocol. Then the expected cost and the variance per repetition is as follows.
\begin{eqnarray}
 \E[\cost_N]=\sum_{j=1}^Mp_jc_j=\frac{1}{\lambda}\sum_{j=1}^Mw_jc_j=\mu_N\quad\text{and}\quad \var[\cost_N]=\frac{1}{\lambda^2}\left(\lambda\sum_{j=1}^Mw_jc_j^2-\left(\sum_{j=1}^Mw_jc_j\right)^2\right)=\sigma_N^2
\end{eqnarray}
By Chebyshev's inequality (Supplementary Note 1 : Appendix \ref{app:prelim}) we have the following for some real number $k>0$.
\begin{eqnarray}
 \Pr\left[|\cost_N-\mu_N|\geq k\sigma_N\right]\leq \frac{1}{k^2}
\end{eqnarray}
The cost per repetition of our protocol is a bounded variable i.e. $a\leq \cost_N\leq b$, for some real numbers $a,b$. If $\cost$ is the variable denoting the cost of all repetitions of our protocol, then
\begin{eqnarray}
 \E[\cost]=N\mu_N
\end{eqnarray}
and since each repetition is independent, making the corresponding cost variables per repetition distributed identically and independently, so we apply Hoeffding's inequality (Supplementary Note 1 : Appendix \ref{app:prelim}) and obtain the following.
\begin{eqnarray}
 \Pr\left[\left|\cost-N\mu_N\right|\geq cN\mu_N\right]\leq 2\exp\left(-\frac{2c^2N^2\mu_N^2}{N(b-a)^2}\right)=2\exp\left(-\frac{2c^2N\mu_N^2}{(b-a)^2}\right)=\epsilon_c   \qquad [c>0]
\end{eqnarray}
Thus with probability at least $1-\epsilon_c$, the cost of all repetitions of the protocol is at most \\
$\frac{(c+1)N}{\lambda}\sum_{j=1}^Mw_jc_j$, where $c=\frac{b-a}{\mu_N\sqrt{2N}}\log\left(\frac{2}{\epsilon_c}\right)$.

\paragraph{Error in simulation while sampling multiple Paulis : }
We consider the qDRIFT protocol \cite{2019_C} for simulating Hamiltonians. If $H=\sum_jh_jH_j$, then in each iteration we sample $H_j$ with probability proportional to $h_j$ and then simulate it for a short time period. Now $H_j$ can be a single Pauli operator or a sum of commuting Paulis, as is achieved in Algorithm~\ref{alg:greedy}, to optimize the cost of simulation. Here we derive a bound on the difference in simulation error for these two cases. 

Let $H_j=\sum_{i_j=1}^{L_j} P_{i_j}$ - sum over commuting Paulis and $M$ be the total number of Pauli operators. So the Hamiltonian can be written as $H=\sum_{j=1}^Lh_jH_j=\sum_{j=1}^L\sum_{i_j=1}^{L_j}h_jP_{i_j}$. We assume the most general case where a single Pauli can be shared between multiple commuting groups i.e. $H_j$.

 In the first case, a group of commuting Paulis i.e one of the $H_j$ is selected independently with probability $q_j=\frac{h_j}{\sum_j h_j}$. In the second case, one single Pauli operator $P_k$ is sampled independently with probability $p_k'=\frac{\sum_{j'}h_{j'}}{\sum_i h_i L_i}$, where in the numerator the sum is over all the commuting Pauli groups in which $P_k$ appears. Let $\lambda=\sum_jh_j$ and $\lambda'=\sum_jh_jL_j$. We define the Liouvillian that generates unitaries under Hamiltonian $H_j$ and $P_{i_j}$ so that
\begin{eqnarray}
 \liou_j&=& i(H_j\rho-\rho H_j)\quad\text{and}\quad\liou_{i_j}=i(P_{i_j}\rho-\rho P_{i_j}).
\end{eqnarray}
Thus if $\liou=i(H\rho-\rho H)$, then
$
 \liou=\sum_{j=1}^Lh_j\liou_j=\sum_{j=1}^Lh_j\sum_{i_j=1}^{L_j}\liou_{i_j}.  
$
We define two channels $\chan_1=\sum_{j=1}^Lq_je^{\tau\liou_j}$ and $\chan_2=\sum_{j=1}^Lp_j\sum_{i_j=1}^{L_j}e^{\tau'\liou_{i_j}}$, where $p_j=\frac{h_j}{\lambda'}$, that evolves the superoperators $\liou_j$ and $\liou_{ij}$ for time interval $\tau=\frac{\lambda t}{N}$ and $\tau'=\frac{\lambda' t}{N}$ respectively. Here we note that for the second channel, for each Pauli $P_k$, we have expanded the sum $p_{k'}=\sum_{j'}\frac{h_{j'}}{\lambda'}$ to reflect the commuting groups in which it belongs. Thus $\sum_{k=1}^Mp_{k'}=\sum_{j=1}^L\sum_{i_j=1}^{L_j}p_j$.

Then we can prove the following.
\begin{lemma} The distance between the qDRIFT channel with single and grouped Hamiltonian terms for simulation time $t$ using $N$ time steps obeys
$$
    \|\chan_2-\chan_1\|_{\diamond}\leq\frac{4t^2\lambda'^2}{N^2}
$$
 \label{lem:multErr}
\end{lemma}
The proof has been given in Supplementary Method 4 (Appendix \ref{app:error} : Lemma \ref{app:lem:multErr}).

\subsection{Optimized circuits for quantum chemistry}
\label{sec:qChem}

In this section we review quantum algorithms for quantum chemistry and design efficient circuits that are useful for quantum chemistry simulation within the Trotter-Suzuki formalism. 
The electronic structure problem has emerged as a central application of quantum computers in recent years, with quantum algorithms providing potential exponential speedups relative to the best known classical algorithms~\cite{2011_WBA,2017_RWSWT}.  The electronic structure problem more specifically is, for a fixed set of positions of the nuclei, find the configuration of electrons that minimizes the total energy for a fixed number of electrons.  The properties of materials, molecules and atoms at low temperatures emerge from these energies. In the non-relativistic case, the dynamics of these electrons are governed by the Coulomb Hamiltonian.
\begin{eqnarray}
 H&=&-\sum_i\frac{\nabla_i^2}{2}-\sum_{i,j}\frac{\zeta_j}{|R_j-r_i|}+\sum_{i<j}\frac{1}{|r_i-r_j|}+\sum_{i<j}\frac{\zeta_i\zeta_j}{|R_i-R_j|}   \nonumber
\end{eqnarray}
where we have used atomic units, $r_i$ represent the positions of electrons, $R_i$ represent the positions of nuclei, and $\zeta_i$ are the charges of nuclei. 

Following the strategy outlined in \cite{2010_LWGetal}, we select the second quantization and discretize the Hamiltonian by representing it within some canonical basis such as a Gaussian basis or a planewave basis. Under the above assumptions, the electronic Hamiltonian can be represented in terms of creation and annihilation operators as follows \cite{2012_SO, 2014_HJO}.  Each spin orbital is assigned a (distinct) qubit where the state $\ket{1}$ corresponds to an occupied orbital and $\ket{0}$ an unoccupied orbital. 
Specifically, let $a_p^{\dagger}$ and $a_p$ be the fermionic raising and lowering operators acting on spin-orbital $p$ satisfying the anti-commutation relation $\{a_p^{\dagger},a_q\}=\delta_{pq}$ and $\{a_p,a_q\}=\{a_p^{\dagger},a_q^{\dagger}\}=0$,
\begin{eqnarray}
  H&=&\sum_{p,q}h_{pq}a_p^{\dagger}a_q+\frac{1}{2}\sum_{p,q,r,s}h_{pqrs}a_p^{\dagger}a_q^{\dagger}a_ra_s 
 \label{eqn:Hferm}
\end{eqnarray}
where the coefficients $h_{pq}, h_{pqrs}$ are determined by the discrete basis set chosen, and the sums run over the number of discretization elements or basis set for a single particle. From inspection, we can see that the number of terms in Equation \ref{eqn:Hferm} is $O(N^4)$, where $N$ is the size of the discrete representation. The molecular orbitals are one widely used basis set. These, in turn, can be expressed as linear combinations of atomic basis functions \cite{1996_F, 2007_SDEetal}. The coefficients of this expansion are obtained by solving the set of Hartree-Fock equations that arise from the variational minimization of the energy using a single determinant wave function. Thus in this representation the location of (indistinguishable) electrons are specified by the occupations of the discrete basis. 

The Jordan-Wigner \cite{1993_JW} or Bravyi-Kitaev \cite{2002_BK} transformations are commonly used to convert the fermionic creation and annihilation operators into Pauli operators. For example, within the Jordan-Wigner encoding, $a$ and $a^{\dagger}$ can be written in terms of qubit operators as follows.
\begin{eqnarray}
 a_p=Q_{(p)}\prod_{j=0}^{p-1}Z_{(j)}=\frac{1}{2}(X_{(p)}+iY_{(p)})\prod_{j=0}^{p-1}Z_{(j)}\quad\text{and}\quad a_p^{\dagger}=Q_{(p)}^{\dagger}\prod_{j=0}^{p-1}Z_{(j)}=\frac{1}{2}(X_{(p)}-iY_{(p)})\prod_{j=0}^{p-1}Z_{(j)}    \nonumber
\end{eqnarray}
Here $Q_{(p)}^{\dagger}=\frac{1}{2}(X_{(p)}-iY_{(p)})$ and $Q_{(p)}=\frac{1}{2}(X_{(p)}+iY_{(p)})$ are the qubit creation and annihilation operators respectively. $\prod_jZ_{(j)}$ acts as an exchange-phase factor, accounting for the anti-commutation relations of $a$ and $a^{\dagger}$. 

With these tools in place, the second-order Trotter-Suzuki approximation reads
\begin{equation}
    e^{-iHt} = \prod_{p,q} e^{-i t(h_{pq} a^\dagger_p a_q + h_{qp} a^\dagger_qa^\dagger_q a_p)/2} \prod_{p,q,r,s}e^{-i t(h_{pqrs} a^\dagger_pa^\dagger_q a_r a_s + h_{srqp} a^\dagger_sa^\dagger_r a_q a_p)/4} + O(t^2)
\end{equation}
Such a simulation can then be performed by substituting in the Pauli representation yielded by the Jordan-Wigner transformation.  Higher-order versions of this are also known~\cite{2007_BACetal} that can achieve error scaling $O(t^{2k+1})$; however, we do not focus on such cases here since the optimizations to the operator exponentials that we consider here will apply in all such cases.

\paragraph{Optimizing two-body operator exponentials : }
The two-body terms are the most common, and often the most significant, contribution to the complexity of a simulation of the Coulomb Hamiltonian in second quantization~\cite{2018_BWMetal}.
In this section we consider the general two-body double excitation terms to reduce this dominant cost for simulation of chemistry, which when expressed using the Jordan-Wigner transformation, can be written as product of $X,Y,Z$ operators as follows \cite{2011_WBA}. We have removed the parentheses in the subscripts, for convenience.
\begin{eqnarray}
 h_{pqrs}a_p^{\dagger}a_q^{\dagger}a_ra_s+h_{srqp}a_s^{\dagger}a_r^{\dagger}a_qa_p&=&\left(\bigotimes_{k=s+1}^{r-1}Z_{k}\right)\left(\bigotimes_{k=q+1}^{p-1}Z_{k}\right)\left(\frac{\Re\{h_{pqrs}\}}{8}H_r+\frac{\Im\{h_{pqrs}\}}{8}H_i\right) \label{eqn:hpqrs}   \\
 \text{where }H_r&=&X_sX_rX_qX_p-X_sX_rY_qY_p+X_sY_rX_qY_p+Y_sX_rX_qY_p \nonumber   \\
&& +Y_sX_rY_qX_p-Y_sY_rX_qX_p+X_sY_rY_qX_p+Y_sY_rY_qY_p \label{eqn:hr}  \\
\text{and }H_i&=&Y_sX_rX_qX_p+X_sY_rX_qX_p-X_sX_rY_qY_p-X_sY_rY_qY_p    \nonumber   \\
&& -Y_sX_rY_qY_p+Y_sY_rX_qX_p+Y_sY_rX_qY_p+Y_sY_rY_qX_p  \nonumber
\end{eqnarray}
Note that if a Gaussian orbital basis is chosen then the values of $h_{pqrs}$ are typically real, resulting in $H_i=0$.  We will assume in the remainder of the discussion that such terms are zero and focus our attention on only the real part of the Hamiltonian.

If we define $h_1=(h_{pqrs}\delta_{X_pX_s}\delta_{X_qX_r}-h_{qprs}\delta_{X_pX_r}\delta_{X_qX_s})$, 
 $h_2=(h_{psqr}\delta_{X_pX_r}\delta_{X_qX_s}-h_{spqr}\delta_{X_pX_q}\delta_{X_rX_s})$ and 
 $h_3=(h_{prsq}\delta_{X_pX_q}\delta_{X_rX_s}-h_{prqs}\delta_{X_pX_s}\delta_{X_qX_r}) 
$, for distinct $p,q,r,s$ then we have the following \cite{2011_WBA}.
\begin{eqnarray}
&& \frac{1}{2}\sum_{p,q,r,s}h_{pqrs}(a_p^{\dagger}a_q^{\dagger}a_ra_s+a_s^{\dagger}a_r^{\dagger}a_qa_p) \nonumber \\
&=&\frac{1}{8}\left(\bigotimes_{k=p+1}^{q-1}\bigotimes_{k=r+1}^{s-1}Z_k\right)
 \Bigg( (X_pX_qX_rX_s+Y_pY_qY_rY_s)(-h_1-h_2+h_3)   \nonumber\\
 &&+(X_pX_qY_rY_s+Y_pY_qX_rX_s)(h_1-h_2+h_3)  
 +(Y_pX_qY_rX_s+X_pY_qX_rY_s)(-h_1-h_2-h_3) \nonumber\\
 &&+(Y_pX_qX_rY_s+X_pY_qY_rX_s)(-h_1+h_2+h_3) \Bigg)  \label{eqn:Hgen}
\end{eqnarray}
Thus conventionally, the part of the Hamiltonian which can be expressed in the form of Equation \ref{eqn:hpqrs}, are broken down into groups of at most 8 commuting operators that act on the qubits in question. Each term is diagonalized by a Clifford circuit and the evolution is performed based on this, with some $R_z$ gates. In \cite{2017_RWSWT} the authors diagonalize all 8 terms in the simultaneous eigenbasis and parallelizes all 8 $R_z$ gates. This reduces the number of Clifford gates, depth, but comes at the cost of using extra 4 ancillae. Excluding the diagonalizing circuits on both sides they use 32 CNOTs and 8 $R_Z$. Each diagonalizing circuit uses 3 CNOT and 1 H gate.
Our goal in this section is to design more efficient quantum circuits for the double excitation terms. 
In Table \ref{tab:totalGate} we have compared the gate costs of the circuit in \cite{2017_RWSWT} with the circuits derived by us in each of the 3 cases considered by us.
Fermionic SWAP gates \cite{2009_VCL, 2018_KMWetal} can be used to make the orbitals neigboring and hence get rid of the tensor product of Z terms. So from here on, we ignore these terms.

Let $q_1, q_2, q_3, q_4$ be the qubits to which the fermions in the orbitals $p,q,r,s$ are mapped respectively. We follow the technique used in \cite{2017_RWSWT}. $W=CNOT_{(1,2)}CNOT_{(1,3)}CNOT_{(1,4)}H_{(1)}$ is the unitary diagonalizing the 8 terms in the simultaneous eigenbasis. We re-write the Hamiltonian $H$ with general coefficients $a_0,\ldots, a_7\subset\real$. Unless mentioned, the leftmost operator acts on qubit $q_1$, next ones on $q_2, q_3$ and the rightmost on qubit $q_4$.
\begin{eqnarray}
 H&=&a_0XXXX+a_1YYXX+a_2YXYX+a_3YXXY  \nonumber\\
    &&+a_4XYYX+a_5XYXY+a_6XXYY+a_7YYYY \label{eqn:H}
\end{eqnarray}
Then following the arguments in \cite{2017_RWSWT} we have the following.
\begin{eqnarray}
 e^{-iHt}&=&W\left(e^{-ia_0Z\id\id\id t}e^{ia_1ZZ\id\id t}e^{ia_2Z\id Z\id t}e^{ia_3Z\id\id Zt}e^{ia_4ZZZ\id t}e^{ia_5ZZ\id Zt}e^{ia_6Z\id ZZt}e^{-a_7ZZZZt} \right)W^{\dagger} \nonumber   \\
 &=&We^{i(-a_0Z\id\id\id+a_1ZZ\id\id +a_2Z\id Z\id +a_3Z\id\id Z+a_4ZZZ\id+a_5ZZ\id Z+a_6Z\id ZZ-a_7ZZZZ)t} W^{\dagger} \label{eqn:eH}
\end{eqnarray}
The terms in between $W$ and $W^{\dagger}$ add an overall phase $\phi$. We denote the state of the qubits $q_1,\ldots, q_4$ after application of $W$ by variables $x_1,\ldots, x_4$ respectively. It is sufficient to analyse the phase when the state is in the standard basis. Consider $e^{-ia_0Z\id\id\id t}$ - this term contributes a phase of $-a_0t$ if $x_1=0$ and $a_0t$ if $x_1=1$. Similarly $e^{ia_1ZZ\id\id t}$ contributes a phase of $a_1t$ if $x_1\oplus x_2=0$ and vice versa. Thus we can have the following expression for the overall phase.
\begin{eqnarray}
 \phi&=&\left(-(-1)^{x_1}a_0+(-1)^{x_1\oplus x_2}a_1+(-1)^{x_1\oplus x_3}a_2+(-1)^{x_1\oplus x_4}a_3     +(-1)^{x_1\oplus x_2\oplus x_3}a_4+(-1)^{x_1\oplus x_2\oplus x_4}a_5 \right. \nonumber   \\
 &&\left.+(-1)^{x_1\oplus x_3\oplus x_4}a_6-(-1)^{x_1\oplus x_2\oplus x_3\oplus x_4}a_7  \right) \label{eqn:phi}
\end{eqnarray}
For different values of $a_0,\ldots a_7$ we get different value of overall phase and different circuits. We consider the following three cases. It is easy to see that $\phi_{x_1=1}=-\phi_{x_1=0}$. So in all the cases below it is sufficient to calculate the phase while setting $x_1=0$.
\begin{figure}
 \centering
 \includegraphics[width=\textwidth]{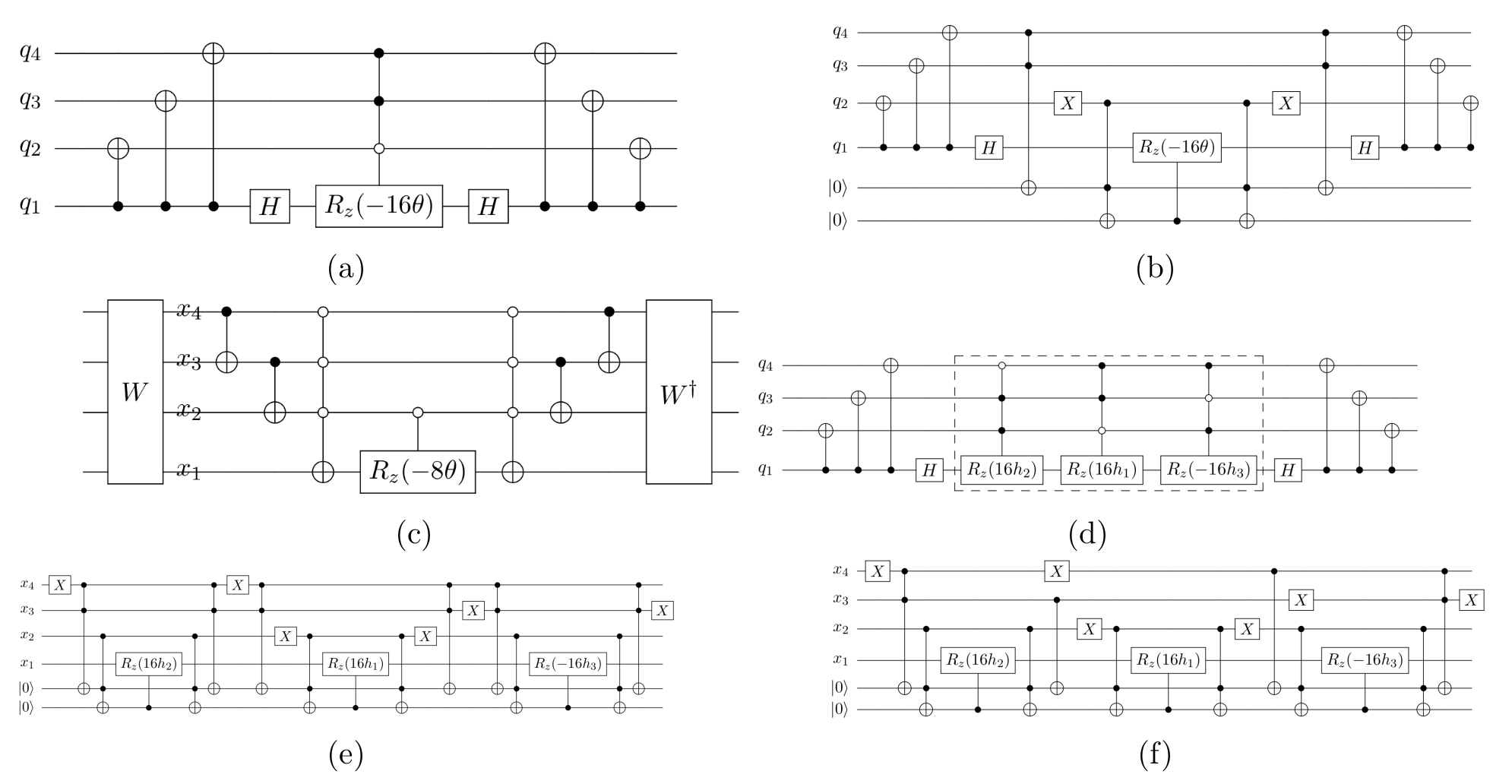}
 \caption{\textbf{Quantum circuit simulating $e^{-iHt}$ : } (a) Circuit when $a_1t=a_6t=-\theta$ and $a_0t=a_2t=\ldots=a_5t=a_7t=\theta$. (b) Circuit in (a) with the multi-controlled rotation implemented with Toffoli and controlled rotation. (c) Circuit when $a_0t=\ldots=a_7t=\theta$. (d) Circuit when the coefficients are as in Equation \ref{eqn:Hgen}. (e) Circuit in (d) with the multi-controlled rotations in the boxed section implemented with Toffolis and controlled rotations. (f) Circuit in (e) with some further reduction of the intermediate Toffoli gates.  }
\label{fig:01}
\end{figure}
\paragraph{Case I : } Let $a_1t=a_6t=-\theta$ and the remaining $a_0t=a_2t=\ldots=a_5t=a_7t=\theta$. Then we can verify that $\phi=8\theta$ if $x_1=1,x_2=0,x_3=x_4=1$, $\phi=-8\theta$ if $x_1=0,x_2=0,x_3=x_4=1$ and $\phi=0$ for the remaining values of $x_1,\ldots, x_4$. Then the quantum circuit simulating $e^{-iHt}$ is shown in Figures 1a and 1b.

\paragraph{Case II : } Let $a_0t=\ldots=a_7t=\theta$. If $x_2=1, x_3\oplus x_4=0$ then $\phi=0$ and if $x_2=0, x_3\oplus x_4=0$ then $\phi=(-1)^{x_3}4\theta$. When $x_3\oplus x_4=1$ then $\phi=-2\theta+(-1)^{x_2}2\theta$. This is equal to $-4\theta$ if $x_2=1$. The quantum circuit simulating $e^{-iHt}$ is shown in Figure \ref{fig:01}c.

\paragraph{Case III : } Let $a_0t=a_7t=-h_1-h_2+h_3$, $a_1t=a_6t=h_1-h_2+h_3$,
 $a_2t=a_5t=-h_1-h_2-h_3$ and $a_3t=a_4t=-h_1+h_2+h_3$, as shown in Equation \ref{eqn:Hgen}. It can be verified that $\phi_{x_2=x_3=1,x_4=0}=8h_2$, $\phi_{x_2=0,x_3=x_4=1}=8h_1$ and $\phi_{x_2=x_4=1,x_3=0}=-8h_3$. For every other values of $x_2,x_3,x_4$, $\phi=0$. The corresponding quantum circuit simulating $e^{-iHt}$ has been shown in Figures 1d, 1e, 1f.

 \begin{table}[t]
\centering 
\footnotesize
\begin{tabular}{|c|c|c|c|c|}
\hline 
& Case I & Case II & Case III & [50] \\
\hline\hline
CNOT & 6 & 10 & 8 & 38 \\
\hline 
H & 2 & 2 & 2 & 2 \\
\hline 
X & 2 & 6 & 6 & 0   \\
\hline 
$(c)R_z$ & 1 & 1 & 3 & 8 \\
\hline
$\#$Toff. pair & 2 & 2 & 4 & 0 \\
\hline \hline
Total & 15 & 23 & 27 & 42   \\
\hline
\end{tabular}
\caption{Comparison of gate counts required to simulate $e^{-iHt}$ (Equations 19, 20) using the circuit synthesized by us with the circuit in [50]. The gate counts have been given for the 3 different cases considered in this work.}
\label{tab:totalGate}
 \end{table}

 We already remarked that we can ignore the product of Z terms in Equations
\ref{eqn:hpqrs} and \ref{eqn:Hgen} by using fermionic SWAP gates. Now if we take two Hamiltonians of the form \ref{eqn:H} having some overlapping qubits, then we can get different Hamiltonians by rearranging the commuting Paulis. In the next few subsections we design circuits for the corresponding exponentials of these Hamiltonians. We must keep in mind that in the following subsections $P_0=\X$ and $P_1=\Y$, $\conj{i}=i+1\mod 2$. Table \ref{tab:comp} summarizes the number of non-Clifford gates used to implement the various circuits. All rotation gates with $n$ ($>1$) controls, can be decomposed into $cR_z$ (single control) and NOT with $n$ controls, each of which can be decomposed into $n-1$ Toffolis (as shown in Figure \ref{fig:01}b). We have discussed in 'Introduction' about special gadgets that can be used to further reduce the T-count of the circuits. In Figure \ref{fig:01}e, 1f we show how Toffolis can be reduced in segments of the circuits. Our circuits have less gates (even the Clifford gates), compared to \cite{2017_RWSWT} or the approaches where we synthesize circuit for each exponentiated Pauli and then concatenate them. In fact, we show the dependence of the circuit size or Clifford and non-Clifford gate cost on the coefficients of the commuting Paulis in the Hamiltonian expression. 

\begin{table}[t]
\centering 
\footnotesize
\begin{tabular}{|c|c|c|c|c|c|c|c|c|c|c|}
 \hline
 $\#q_{overlap}$ & Ham. & \multicolumn{3}{c|}{Case I} & \multicolumn{3}{c|}{Case II} & \multicolumn{3}{c|}{Case III} \\
 \cline{3-11}
 & & $\#cR_z$ & $\#R_z$ & $\#$Toff. pair & $\#cR_z$ & $\#R_z$ & $\#$Toff. pair & $\#cR_z$ & $\#R_z$ & $\#$Toff. pair \\
 \hline
 0 & $H$ & 1 & 0 & 2 & 1 & 0 & 2 & 3 & 0 & 4 \\
 \hline
 \multirow{2}{*}{1} & $H_{1y}$ & 2 & 0 & 2 & 0 & 2 & 8 & 6 & 0 & 24 \\
 \cline{2-11}
 & $H_{1x}$ & 2 & 0 & 2 & 0 & 2 & 8 & 6 & 0 & 6 \\
 \hline
 \multirow{2}{*}{2} & $H_{21}$ & - & - & - & 2 & 0 & 2 & 4 & 0 & 4 \\
 \cline{2-11}
 & $H_{20}$ & 2 & 0 & 2 & 2 & 0 & 2 & 4 & 0 & 4 \\
 \hline
 \multirow{2}{*}{3} & $H_{3y}$ & 2 & 0 & 2 & 1 & 0 & 2 & 6 & 0 & 6 \\
 \cline{2-11}
 & $H_{3x}$ & 2 & 0 & 2 & 0 & 2 & 2 & 6 & 0 & 6 \\
 \hline
\end{tabular}

\begin{tabular}{|c|c|c|c|c|}
 \hline
 $\#q_{overlap}$ & Ham. &  \\
 \hline 
 \multirow{2}{*}{0} & \multirow{2}{*}{$H$} & $a_0XXXX+a_1YYXX+a_2YXYX+a_3YXXY$ \\
 & & $+a_4XYYX+a_5XYXY+a_6XXYY+a_7YYYY $ \\
 \hline
 \multirow{4}{*}{1} & \multirow{2}{*}{$H_{1y}$} & $a_3YXXY\id\id\id+a_5XYXY\id\id\id+a_6XXYY\id\id\id+a_7YYYY\id\id\id $ \\
 & & $+b_1\id\id\id YYXX+b_2\id\id\id YXYX+b_3\id\id\id YXXY+b_7\id\id\id YYYY $ \\
 \cline{2-3}
 & \multirow{2}{*}{$H_{1x}$} & $a_0XXXX\id\id\id+a_1YYXX\id\id\id+a_2YXYX\id\id\id+a_4XYYX\id\id\id \nonumber $ \\
 & & $+b_0\id\id\id XXXX+b_4\id\id\id XYYX+b_5\id\id\id XYXY+b_6\id\id\id XXYY$ \\
 \hline
  \multirow{4}{*}{2} & \multirow{2}{*}{$H_{21}$} & $a_2YXYX\id\id+a_3YXXY\id\id+a_4XYYX\id\id+a_5XYXY\id\id $ \\
 & & $+b_2\id\id YXYX+b_3\id\id XYYX+b_4\id\id YXXY+b_5\id\id XYXY $ \\
 \cline{2-3}
 & \multirow{2}{*}{$H_{20}$} & $a_0XXXX\id\id+a_1YYXX\id\id+a_6XXYY\id\id+a_7YYYY\id\id $ \\
 & & $+b_0\id\id XXXX+b_6\id\id YYXX+b_1\id\id XXYY+b_7\id\id YYYY $ \\
 \hline
  \multirow{4}{*}{3} & \multirow{2}{*}{$H_{3y}$} & $a_1YYXX\id+a_2YXYX\id+a_3YXXY\id+a_7YYYY\id $ \\
 & & $+b_3\id YXXY+b_5\id XYXY+b_6\id XXYY+b_7\id YYYY$ \\
 \cline{2-3}
 & \multirow{2}{*}{$H_{3x}$} & $a_0XXXX\id+a_4XYYX\id+a_5XYXY\id+a_6XXYY\id $ \\
 & & $+b_0\id XXXX+b_1\id YYXX+b_2\id YXYX+b_4\id XYYX$ \\
 \hline
\end{tabular}

\begin{tabular}{|c|c|}
\hline 
 Case & Coefficients \\
 \hline 
 I & $a_1t=a_6t=-\theta_1$; $a_0t=a_2t=\ldots=a_5t=a_7t=\theta_1$; $b_1t=b_6t=-\theta_2$; $a_0t=a_2t=\ldots=a_5t=a_7t=\theta_2$ \\
 \hline
 II & $a_0t=\ldots=a_7t=\theta_1$; $b_0t=\ldots=b_7t=\theta_2$ \\
 \hline
 \multirow{2}{*}{III} &  $a_0t=a_7t=-h_1-h_2+h_3$; $a_1t=a_6t=h_1-h_2+h_3$; $a_2t=a_5t=-h_1-h_2-h_3$; $a_3t=a_4t=-h_1+h_2+h_3$\\
 &$b_0t=b_7t=-g_1-g_2+g_3$; $b_1t=b_6t=g_1-g_2+g_3$; $b_2t=b_5t=-g_1-g_2-g_3$; $b_3t=b_4t=-g_1+g_2+g_3$\\
 \hline
\end{tabular}
\caption{The first table shows the number of $R_z, cR_z$ and Toffoli (Toff.) pairs used to design the circuits implementing $e^{-iH't}$, where $H'$ are the Hamiltonians (Ham.) described in Section \ref{sec:qChem}. $\#q_{overlap}$ is the number of overlapping qubits.  Pairs of Toffoli gates are cited here, since [46] can be used in these cases to uncompute the action of the Toffoli gates in these cases. The second table summarizes the Hamiltonians and the third table summarizes the value of the coefficients for the different cases. When $\#q_{overlap}=0$ then $\theta_1=\theta$, since there are no $b$-coefficients.}
\label{tab:comp}
\end{table}

\paragraph{Overlap on 1 qubit : }
Previously, we provided an analysis of the circuits for cases where many of the Hamiltonian coefficients are chosen to follow regular patterns and see that the costs of the simulation can be reduced through the use of these techniques.  Here we provide a more aggressive strategy wherein we combine multiple commuting terms together and find particular combinations of angles such that the simulation circuits are efficient.  The results are summarized in Table \ref{tab:comp}.
We consider the case when there is overlap on 1 qubit. We can have the following sets of commuting Paulis. 
\begin{eqnarray}
 G_{1y}&=&\{P_iP_jP_kY\id\id\id, \id\id\id Y P_aP_bP_c:i+j+k\equiv 1\mod 2, a+b+c\equiv 1\mod 2 \}    \label{eqn:G1y} \\
 G_{1x}&=&\{P_iP_jP_kX\id\id\id, \id\id\id X P_aP_bP_c:i+j+k\equiv 0\mod 2, a+b+c\equiv 0\mod 2 \}    \label{eqn:G1x} 
\end{eqnarray}
Without loss of generality, we assume that the leftmost operator acts on qubit $q_1$, next one on $q_2$ and so on - rightmost one acts on qubit $q_7$. We denote a state vector as $\ket{Q_1vQ_2}$ where $Q_1=\ket{q_1q_2q_3}$, $Q_2=\ket{q_5q_6q_7}$ and $v,q_1,\ldots,q_7\in\{0,1\}$. We can have the following Hamiltonian terms, expressed as sums of commuting Paulis from the above two sets.
\begin{eqnarray}
  H_{1y}&=&a_3YXXY\id\id\id+a_5XYXY\id\id\id+a_6XXYY\id\id\id+a_7YYYY\id\id\id \nonumber \\
 &&+b_1\id\id\id YYXX+b_2\id\id\id YXYX+b_3\id\id\id YXXY+b_7\id\id\id YYYY 
 \label{eqn:h1y} \\
 H_{1x}&=&a_0XXXX\id\id\id+a_1YYXX\id\id\id+a_2YXYX\id\id\id+a_4XYYX\id\id\id \nonumber \\
 &&+b_0\id\id\id XXXX+b_4\id\id\id XYYX+b_5\id\id\id XYXY+b_6\id\id\id XXYY
 \label{eqn:h1x} 
\end{eqnarray}

\paragraph{Circuit for simulating $e^{-iH_{1y}t}$ : }

Let $W_{1y}$ be the unitary consisting of the following sequence of gates. The rightmost one is the first to be applied. With a slight abuse of notation we denote $CNOT_{(4,1)}CNOT_{(4,2)}CNOT_{(4,3)}$ by $CNOT_{(4,I)}$ and $CNOT_{(4,5)}CNOT_{(4,6)}CNOT_{(4,7)}$ by $CNOT_{(4,II)}$.
\begin{eqnarray}
 W_{1y}=CNOT_{(4,I)}H_{(4)}Z_{(4)}CNOT_{(4,I)}CNOT_{(4,II)}H_{(4)}CNOT_{(4,I)}    \nonumber
\end{eqnarray}
In the following theorem we show that this is a diagonalizing circuit for the set of Paulis in $G_{1y}$.
\begin{theorem}
For each $i,j,k,l,a,b,c\in\intg_2$, such that $P_iP_jP_kY\id\id\id, \id\id\id YP_aP_bP_c\in G_{1y}$ we have the following.
 \begin{eqnarray}
  \sqrt{-1}^{i+j+k+1}W_{1y}\left(Z_{(1)}^iZ_{(2)}^jZ_{(3)}^kZ_{(4)}\id\id\id\right)W_{1y}^{\dagger}=P_iP_jP_k Y\id\id\id \nonumber \\
  \text{ and } \sqrt{-1}^{a+b+c+1}W_{1y}\left(\id\id\id Z_{(4)}Z_{(5)}^iZ_{(6)}^jZ_{(7)}^k\right)W_{1y}^{\dagger}=\id\id\id Y P_aP_bP_c \nonumber 
 \end{eqnarray}
 \label{thm:diagGy}
\end{theorem}
We prove this theorem by showing that the operators on the LHS and RHS have equivalent actions on the eigenstates corresponding to an eigenbasis for the Paulis in $G_{1y}$. The proof of this theorem has been given in Theorem \ref{app:thm:diagGy} of Supplementary Method 1 (Appendix \ref{app:1overlap}). The eigenbasis has been shown in Lemma \ref{app:lem:ebasisGy} of Supplementary Method 1 (Appendix \ref{app:1overlap}). 

Thus we have the following.
\begin{eqnarray}
 e^{-iH_{1y}t} &=&e^{-i(-a_3W_{1y}(Z\id\id Z\id \id\id)W_{1y}^{\dagger}-a_5W_{1y}(\id Z\id Z\id\id\id)W_{1y}^{\dagger}-a_6W_{1y}(\id\id ZZ\id\id\id)W_{1y}^{\dagger}+a_7W_{1y}(Z ZZZ\id\id\id)W_{1y}^{\dagger})t} \nonumber   \\
 &&\cdot e^{-i(-b_1W_{1y}(\id\id\id ZZ \id\id)W_{1y}^{\dagger}-b_2W_{1y}(\id\id\id Z\id Z\id)W_{1y}^{\dagger}-b_3W_{1y}(\id\id\id Z\id\id Z)W_{1y}^{\dagger}+b_7W_{1y}(\id\id\id ZZ ZZ)W_{1y}^{\dagger})t} \nonumber \\
 &=&W_{1y}e^{ia_3 Z\id\id Z\id \id\id t}e^{ia_5\id Z\id Z\id\id\id t}e^{ia_6 \id\id ZZ\id\id\id t}e^{-ia_7 ZZZZ\id\id\id t} 
  e^{ib_1\id\id\id ZZ \id\id t}e^{ib_2\id\id\id Z\id Z\id t}e^{ib_3\id\id\id Z\id\id Z t}e^{-ib_7\id\id\id ZZ ZZ t} W_{1y}^{\dagger}\nonumber 
\end{eqnarray}
The state of the qubits $q_1,\ldots,q_7$ after the application of $W_{1y}$ is denoted by variables $x_1,\ldots,x_7$ respectively. We have the following expression for the overall phase incurred between $W_{1y}$ and $W_{1y}^{\dagger}$.
\begin{eqnarray}
 \phi&=&(-1)^{x_4\oplus x_1}a_3t+(-1)^{x_4\oplus x_2}a_5t+(-1)^{x_4\oplus x_3}a_6t-(-1)^{x_4\oplus x_1\oplus x_2\oplus x_3}a_7t \nonumber \\
 &&+(-1)^{x_4\oplus x_7}b_3t+(-1)^{x_4\oplus x_6}b_2t+(-1)^{x_4\oplus x_5}b_1t-(-1)^{x_4\oplus x_5\oplus x_6\oplus x_7}b_7t   \nonumber
\end{eqnarray}
It is easy to check that $\phi_{\conj{x_4}}=-\phi_{x_4}$. We consider the following three cases and it is sufficient to check the phase values when $x_4=0$.

\begin{figure}
  \centering
   \includegraphics[width=\textwidth]{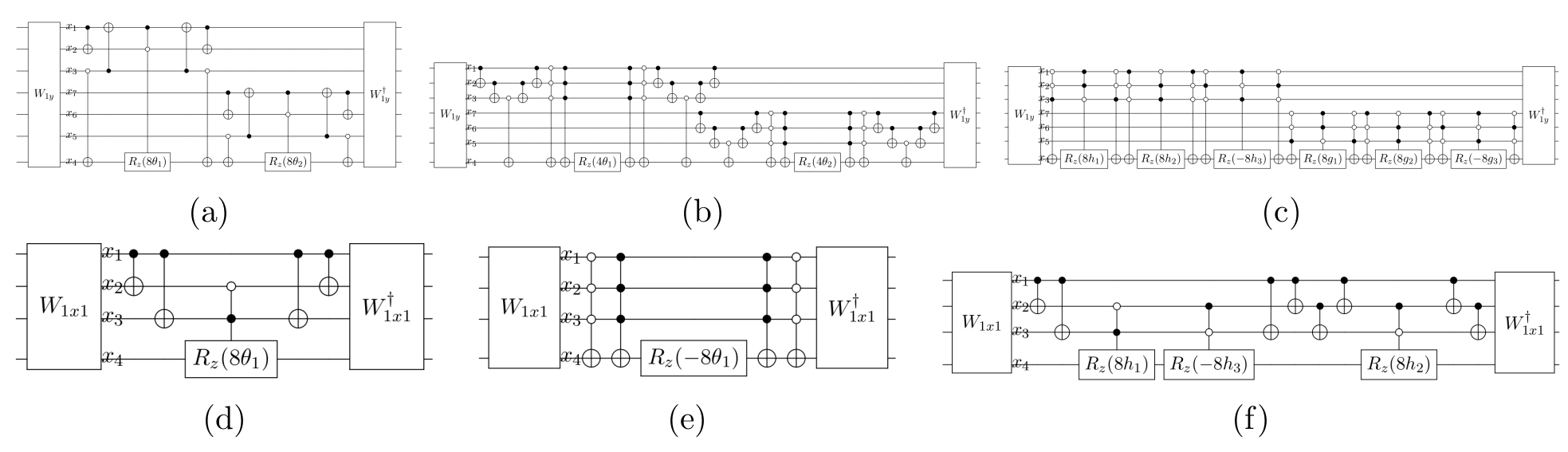}
 \caption{\textbf{Quantum circuit for $e^{-iH_{1y}t}$ and $e^{-iH_{1x1}t}$ : } (a)-(c): Circuit simulating $e^{-iH_{1y}t}$ (a) when $a_6t=-\theta_1, b_1t=-\theta_2, a_3t=a_5t=a_7t=\theta_1$ and $b_2t=b_3t=b_7t=\theta_2$; (b) when $a_0t=\ldots=a_7t=\theta_1$ and $b_0t=\ldots=b_7t=\theta_2$; (c) when the coefficients are as in Equation 18. (d)-(f): Circuit simulating $e^{-iH_{1x1}t}$ (d) when $a_1t=-\theta_1, b_6t=-\theta_2, a_0t=a_2t=a_3t=\theta_1$ and $b_0t=b_4t=b_5t=\theta_2$; (e) when $a_0t=\ldots=a_7t=\theta_1$ and $b_0t=\ldots=b_7t=\theta_2$; (f) when the coefficients are as in Equation 18. }
\label{fig:23}
\end{figure}
\paragraph{Case I : } Let $a_6 t=-\theta_1$, $b_1t=-\theta_2$, $a_3t=a_5t=a_7t=\theta_1$ and $b_2t=b_3t=b_7t=\theta_2$. Following the conventions and explanations given for Case I we have the following overall phase after the application of $W_{1y}$. We can write $\phi=f_1(\theta_1)+f_2(\theta)$, for two functions $f_1$ and $f_2$. The following can be verified.
\begin{enumerate}
 \item If $x_1=x_2=0$ and $x_3=1$ then $\phi=4\theta_1+f_2(\theta_2)$. Analogously, if $\phi=f_1(\theta_1)+4\theta_2$ when $x_7=x_6=0, x_5=1$.

\item If $x_1=x_2=1$ and $x_3=0$ then $\phi=-4\theta_1+f_2(\theta_2)$ and if $x_7=x_6=1, x_5=0$ then $\phi=f_1(\theta_1)-4\theta_2$.

\item For any other values of $x_1,x_2,x_3$, $\phi=f_2(\theta_2)$ and analogously, for any other values of $x_7,x_6,x_5$, $\phi_=f_1(\theta_1)$.
\end{enumerate}
A quantum circuit simulating $e^{-iH_{1y}t}$ is shown in Figure \ref{fig:23}a.

\paragraph{Case II : } Now we consider the case when $a_6 t=a_3t=a_5t=a_7t=\theta_1$ and $b_1t=b_2t=b_3t=b_7t=\theta_2$. Here also $\phi$ can be written as sum of two functions : $\phi=f_1(\theta_1)+f_2(\theta_2)$. We can make the following observations.
\begin{enumerate}
 \item If only one of $x_1, x_2, x_3$ is 1 then $\phi=2\theta_1+f_2(\theta_2)$ and analogously, if any one of $x_5, x_6, x_7$ is 1 then  $\phi=f_1(\theta_1)+2\theta_2$.
 
 \item If any two of $x_1,x_2,x_3$ is 1 then $\phi=-2\theta_1+f_2(\theta_2)$. Similarly, if any two of $x_5, x_6, x_7$ is 1 then $\phi=f_1(\theta_1)-2\theta_2$.
 
 \item If $x_1=x_2=x_3=0$ then $\phi=2\theta_1+f_2(\theta_2)$ and similarly, if $x_5=x_6=x_7=0$ then $\phi=f_1(\theta_1)+2\theta_2$.
 
 \item If $x_1=x_2=x_3=1$ then $\phi=-2\theta_1+f_2(\theta_2)$ and analogously, if $x_5=x_6=x_7=1$ then $\phi=f_1(\theta_1)-2\theta_2$.
\end{enumerate}
A circuit simulating $e^{-iH_{1y}t}$ in this case, has been shown in Figure \ref{fig:23}b.

\paragraph{Case III : } Let $a_3t=-h_1+h_2+h_3, a_5t=-h_1-h_2-h_3, a_6t=h_1-h_2+h_3, a_7t=-h_1-h_2+h_3$ and $b_3t=-g_1+g_2+g_3, b_2t=-g_1-g_2-g_3, b_1t=g_1-g_2+g_3, b_7t=-g_1-g_2+g_3$ (Equation \ref{eqn:Hgen}). Let $\vect{h}=(h_1,h_2,h_3)$ and $\vect{g}=(g_1,g_2,g_3)$. We can write $\phi=f_1(\vect{h})+f_2(\vect{g})$. We can make the following observations.
\begin{enumerate}
 \item If $x_1=x_2=x_3$ then $\phi=f_2(\vect(g))$ and analogously, if $x_5=x_6=x_7$ then $\phi=f_1(\vect{h})$.

 \item Suppose $x_i=x_j$ and $x_k\neq x_i$, where $i,j,k\in\{1,2,3\}$ and $i\neq j\neq k$. Then flipping the values changes the sign. For example, if $\phi_{x_1=x_2=0,x_3=1}=f_1(\vect{h})+f_2(\vect{g})$, then $\phi_{x_1=x_2=1,x_3=0}=-f_1(\vect{h})+f_2(\vect{g})$. Similar phenomenon occurs if $i,j,k\in\{7,6,5\}$, except this time sign of $f_2(\vect{g})$ flips.. So it is sufficient to consider the case when two variables are 1.
\begin{eqnarray}
 \phi_{x_1=x_2=1,x_3=0}=4h_1+f_2(\vect{g}),&\qquad& \phi_{x_7=x_6=1,x_5=0}=f_1(\vect{h})+4g_1      \nonumber\\
 \phi_{x_2=x_3=1,x_1=0}=4h_2+f_2(\vect{g}),&\qquad& \phi_{x_6=x_5=1,x_7=0}=f_1(\vect{h})+4g_2       \nonumber\\
 \phi_{x_3=x_1=1,x_2=0}=-4h_3+f_2(\vect{g}),&\qquad& \phi_{x_5=x_7=1,x_6=1}=f_1(\vect{h})-4g_3   \nonumber
\end{eqnarray}
\end{enumerate}
A circuit simulating $e^{-iH_{1y}t}$ in this case, has been shown in Figure \ref{fig:23}c.

\paragraph{Circuit for simulating $e^{-iH_{1x}t}$ : }
An eigenbasis for the Paulis in $G_{1x}$ has been given in Lemma \ref{app:lem:ebasisGx} of Supplementary Method 1 (Appendix \ref{app:1overlap}). But we are unable to find out (by hand) a unitary (analogous to $W_{1y}$) that diagonalizes the set of commuting Paulis in $G_{1x}$, as we did in the previous subsection for $G_{1y}$. 
So we divide the commuting Paulis into two groups of 4-qubit Paulis, i.e. we consider the following two sets.
\begin{eqnarray}
 G_{1x1}&=&\{P_iP_jP_kX\id\id\id : i+j+k\equiv 0\mod 2.\} \nonumber \\
 G_{1x2}&=&\{\id\id\id XP_aP_bP_c : a+b+c\equiv 0\mod 2.\}    \nonumber
\end{eqnarray}
and the following two Hamiltonians
\begin{eqnarray}
 H_{1x1}&=&a_0XXXX\id\id\id+a_1YYXX\id\id\id+a_2YXYX\id\id\id+a_4XYYX\id\id\id \nonumber \\
 H_{1x2}&=&b_0\id\id\id XXXX+b_4\id\id\id XYYX+b_5\id\id\id XYXY+b_6\id\id\id XXYY  \nonumber
\end{eqnarray}
We can use the diagonalizing circuit of \cite{2017_RWSWT} and have the following.
\begin{eqnarray}
e^{-iH_{1x1}t}&=&W_{1x1}e^{-ia_0 \id\id\id Z \id\id\id t}e^{ia_1ZZ\id Z \id\id\id t}e^{ia_2 Z\id ZZ \id\id\id t}e^{ia_4 \id ZZZ \id\id\id t}W_{1x1}^{\dagger} \nonumber\\
 e^{-iH_{1x2}t}&=&W_{1x2}e^{-ib_0 \id\id\id Z\id\id\id t}e^{ib_4\id\id\id ZZZ\id t}e^{ib_5 \id\id\id ZZ\id Z t}e^{ib_6 \id\id\id Z\id ZZ t}W_{1x2}^{\dagger}\nonumber 
\end{eqnarray}
where $W_{1x1}=CNOT_{(4,1)}CNOT_{(4,2)}CNOT_{(4,3)}H_{(4)}$ and $W_{1x2}=CNOT_{(4,5)}CNOT_{(4,6)}CNOT_{(4,7)}H_{(4)}$, where the rightmost gate is the first one to be applied. We denote the state of the qubits $q_1,\ldots,q_4$ after the application of $W_{1x1}$ by the variables $x_1,\ldots,x_4$ respectively. Also, the variables $x_4',\ldots,x_7'$ denote the state of the qubits $q_4,\ldots,q_7$, respectively, after the application of $W_{1x2}$. We have the following expression for the overall phase incurred between $W_{1x1}$, $W_{1x1}^{\dagger}$ and between $W_{1x2}$, $W_{1x2}^{\dagger}$.
\begin{eqnarray}
 \phi_1&=&-(-1)^{x_4}a_0t+(-1)^{x_4\oplus x_1\oplus x_2}a_1t+(-1)^{x_4\oplus x_1\oplus x_3}a_2t+(-1)^{x_4\oplus x_2\oplus x_3}a_4t     \nonumber \\
  \phi_2&=&-(-1)^{x_4'}b_0t+(-1)^{x_4'\oplus x_5'\oplus x_6'}b_4t+(-1)^{x_4'\oplus x_5'\oplus x_7'}b_5t+(-1)^{x_4'\oplus x_6'\oplus x_7'}b_6t     \nonumber
\end{eqnarray}
We consider the following three cases, in each of which $\phi_{1,\conj{x_4}}=-\phi_{1,x_4}$ and $\phi_{2,\conj{x_4'}}=\phi_{2,x_4'}$.
\paragraph{Case I : } Let $a_1 t=-\theta_1$, $b_6t=-\theta_2$, $a_0t=a_2t=a_4t=\theta_1$ and $b_0t=b_4t=b_5t=\theta_2$. It is easy to verify that a non-zero phase $\phi_1=-4\theta_1$ exists if and only if $x_1=x_2\neq x_3$. Analogously, $\phi_2=-4\theta_2$ if $x_7'=x_6'\neq x_5'$, else it is 0. 

\paragraph{Case II : } Now we consider the case when $a_0 t=a_1t=a_2t=a_4t=\theta_1$, $b_0t=b_4t=b_5t=b_6t=\theta_2$. If $x_1=x_2=x_3$ then $\phi_1=2\theta_1$, else it is $-2\theta_1$. Similarly for $\phi_2$. 

\paragraph{Case III : } Next we consider the case where $a_0t=-h_1-h_2+h_3, a_1t=h_1-h_2+h_3, a_2t=-h_1-h_2-h_3, a_4t=-h_1+h_2+h_3, $ and $b_0t=-g_1-g_2+g_3, b_4t=-g_1+g_2+g_3, b_5t=-g_1-g_2-g_3, b_6t=g_1-g_2+g_3$. Here, non-zero phase exists if any two of the variable have same value. 
 \begin{eqnarray}
  \phi_1(x_1=x_2\neq x_3)=4h_1; &\quad& \phi_1(x_2=x_3\neq x_1)=4h_2; \quad \phi_1(x_1=x_3\neq x_2)=-4h_3; \nonumber \\
  \phi_2(x_5'=x_6'\neq x_7')=4g_2 ; &\quad& \phi_2(x_6'=x_7'\neq x_5')=4g_1 ; \quad \phi_2(x_5'=x_7'\neq x_6')=-4g_3 ;   \nonumber
 \end{eqnarray}
 
Circuits simulating $e^{-iH_{1x1}t}$ in Case I, II and III have been shown in Figure \ref{fig:23}d, \ref{fig:23}e and \ref{fig:23}f respectively. Circuits for $e^{-iH_{1x2}t}$ are similar. Circuit for $e^{-iH_{1x}t}$ in each case is obtained by concatenating the corresponding circuits. 

\paragraph{Overlap on 2 qubits : }
In general, the more options that we have for grouping mutually commuting terms the more effective our compilation strategy will be.  While the most natural case to examine is the case where all of the Hamiltonian terms act on disjoint sets of qubits, Hamiltonian terms can commute if they overlap on only two qubits as well. For example, we can have the following sets of commuting Pauli operations
\begin{eqnarray}
 G_{21}&=&\{P_kP_lP_iP_j\id\id, \id\id P_iP_jP_kP_l:i+j\equiv 1\mod 2, k,l=i,j \text{ or } \conj{i},\conj{j}\text{ respectively}\}    \label{eqn:G21}  \\
 G_{20}&=&\{P_kP_lP_iP_j\id\id, \id\id P_iP_jP_kP_l:i+j\equiv 0\mod 2, k,l=i,j \text{ or } \conj{i},\conj{j}\text{ respectively}\}    \label{eqn:G20}
\end{eqnarray}
Without loss of generality, we assume that the leftmost operator acts on qubit $q_1$, next one on $q_2$ and so on - rightmost one acts on qubit $q_6$. We denote a state vector as $\ket{Q_1Q_2Q_3}$ where $Q_1=\ket{q_1q_2}$, $Q_2=\ket{q_3q_4}$ and $Q_3=\ket{q_5q_6}$ are the first, second and third pairs of qubits respectively. We can have the following Hamiltonian terms, expressed as sums of commuting Paulis from the above two sets.
\begin{eqnarray}
  H_{21}&=&a_2YXYX\id\id+a_3YXXY\id\id+a_4XYYX\id\id+a_5XYXY\id\id \nonumber \\
 &&+b_2\id\id YXYX+b_3\id\id XYYX+b_4\id\id YXXY+b_5\id\id XYXY 
 \label{eqn:h21} \\
 H_{20}&=&a_0XXXX\id\id+a_1YYXX\id\id+a_6XXYY\id\id+a_7YYYY\id\id \nonumber \\
 &&+b_0\id\id XXXX+b_6\id\id YYXX+b_1\id\id XXYY+b_7\id\id YYYY 
 \label{eqn:h20} 
\end{eqnarray}

\paragraph{Circuit for simulating $e^{-iH_{21}t}$ : }
As before our simulation strategy involves diagonalizing the Hamiltonian using a Clifford circuit and then build
Let $W_{1}$ be the unitary consisting of the following sequence of gates. The rightmost one is the first to be applied.
\begin{eqnarray}
 W_{1}=CNOT_{(3,1)}CNOT_{(3,4)}H_{(3)}Z_{(3)}CNOT_{(3,1)}CNOT_{(3,5)}H_{(3)}CNOT_{(3,1)}    \nonumber
\end{eqnarray}
The following theorem shows that this is a diagonalizing circuit for the set of Paulis in $G_{21}$.
\begin{theorem}
For each $i,j,k,l\in\intg_2$, such that $P_kP_lP_iP_j\id\id, \id\id P_iP_jP_kP_l\in G_{21}$ we have the following.
 \begin{eqnarray}
  \sqrt{-1}^{i+j+k+l}W_1\left(Z_{(1)}^kZ_{(2)}^lZ_{(3)}Z_{(4)}^j\id\id\right)W_1^{\dagger}=P_kP_lP_iP_j\id\id \nonumber \\
  \text{ and } \sqrt{-1}^{i+j+k+l}W_1\left(\id\id Z_{(3)}Z_{(4)}^jZ_{(5)}^kZ_{(6)}^l\right)W_1^{\dagger}=\id\id P_iP_jP_kP_l \nonumber 
 \end{eqnarray}
 \label{thm:diag61}
\end{theorem}
The proof is similar to Theorem \ref{thm:diagGy} and has been given in Supplementary Method 2 (Appendix \ref{app:2overlap}). Theorem~\ref{thm:diag61} then gives us the following.
\begin{eqnarray}
 e^{-iH_{21}t} &=&e^{-i(-a_2W_1(Z\id Z\id \id\id)W_1^{\dagger}-a_3W_1(Z\id ZZ\id\id)W_1^{\dagger}-a_4W_1(\id Z Z\id\id\id)W_1^{\dagger}-a_5W_1(\id ZZZ\id\id)W_1^{\dagger})t} \nonumber   \\
 &&\cdot e^{-i(-b_2W_1(\id\id Z\id Z\id)W_1^{\dagger}-b_3W_1(\id\id ZZZ\id)W_1^{\dagger}-b_4W_1(\id\id Z\id \id Z)W_1^{\dagger}-b_5W_1(\id\id ZZ\id Z)W_1^{\dagger})t} \nonumber \\
 &=&W_1e^{ia_2Z\id Z\id \id\id t}e^{a_3Z\id ZZ\id\id t}e^{a_4 \id Z Z\id\id\id t}e^{a_5 \id ZZZ\id\id t} 
  e^{ib_2\id\id Z\id Z\id t}e^{b_3\id\id ZZZ\id t}e^{b_4\id\id Z\id \id Z t}e^{b_5\id\id ZZ\id Z t} W_1^{\dagger}\nonumber 
\end{eqnarray}
We denote the state of the qubits $q_1,\ldots,q_6$ after the application of $W_1$ by the variables $x_1,\ldots,x_6$ respectively. We have the following expression for the overall phase incurred between $W_1$ and $W_1^{\dagger}$.
\begin{eqnarray}
 \phi&=&(-1)^{x_3\oplus x_1}a_2t+(-1)^{x_3\oplus x_2}a_4t+(-1)^{x_3\oplus x_4\oplus x_1}a_3t+(-1)^{x_3\oplus x_4\oplus x_2}a_5t\nonumber \\
&&+(-1)^{x_3\oplus x_5}b_2t+(-1)^{x_3\oplus x_6}b_4t+(-1)^{x_3\oplus x_4\oplus x_5}b_3t+(-1)^{x_3\oplus x_4\oplus x_6}b_5t    \nonumber     
\end{eqnarray}
\begin{figure}
 \centering
 \includegraphics[width=\textwidth]{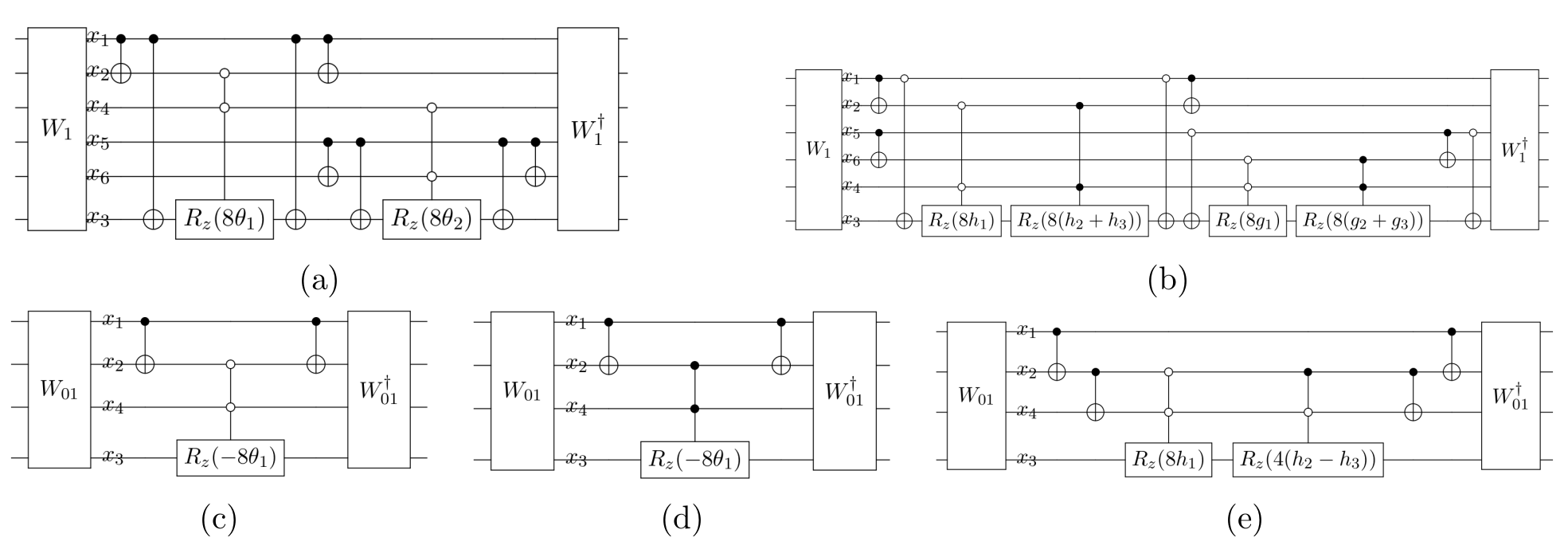}
 \caption{\textbf{Quantum circuit for $e^{-iH_{21}t}$ and $e^{-iH_{201}t}$ : } (a)-(b): Circuit simulating $e^{-iH_{21}t}$ (a) when $a_2t=a_3t=a_4t=a_5t=\theta_1$ and $b_2t=b_3t=b_4t=b_5t=\theta_2$; (b) when the coefficients are as in Equation 18. (c)-(e) Circuit simulating $e^{-iH_{201}t}$ (c) when $a_1t=a_6t=-\theta_1, b_1t=b_6t=-\theta_2, a_0t=a_7t=\theta_1$ and $b_0t=b_7t=\theta_2$; (d) when $a_0t=\ldots=a_7t=\theta_1$ and $b_0t=\ldots=b_7t=\theta_2$; (e) when the coefficients are as in Equation 18.  }
\label{fig:45}
\end{figure}
It is easy to check that $\phi_{\conj{x_3}}=-\phi_{x_3}$. We consider the following cases and it is sufficient to check the phase values when $x_3=0$.
\paragraph{Case I (II) : } We consider the case when $a_2 t=a_3t= a_4t = a_5t=\theta_1$, $b_2t=b_3t=b_4t=b_5t=\theta_2$. There are no $a_1, a_6, b_1, b_6$ in the expression of the Hamiltonian. So, for consistency with the previous and following subsection, we can consider this as either Case I or II.

We can write $\phi=f_1(\theta_1)+f_2(\theta_2)$. We can verify that when $q_1\oplus q_2=q_4=0$ then $\phi=(-1)^{q_1}4\theta_1+f_2(\theta_2)$ and analogously, when $q_5\oplus q_6=q_4=0$ then $\phi=f_1(\theta_1)+(-1)^{q_5}4\theta_2$. For all other values of $q_1,q_2, q_4$, $\phi=f_2(\theta_2)$ and for all other values of $q_5,q_6, q_4$, $\phi=f_1(\theta_1)$.
A quantum circuit for simulating $e^{-iH_{21}t}$ in this case has been shown in Figure \ref{fig:45}a.

\paragraph{Case III : } Let $a_2t=a_5t=-h_1-h_2-h_3$ and $a_3t=a_4t=-h_1+h_2+h_3$, $b_2t=b_5t=-g_1-g_2-g_3$ and $b_3t=b_4t=-g_1+g_2+g_3$. We can write $\phi=f_1(\vect{h})+f_2(\vect{g})$, where $\vect{h}=(h_1,h_2,h_3)$ and $\vect{g}=(g_1,g_2,g_3)$. When $q_1\oplus q_2=q_4=0$, then $\phi=-(-1)^{q_1}4h_1+f_2(\vect{g})$ and when $q_1\oplus q_2=q_4=1$, then $\phi=-(-1)^{q_1}4(h_1+h_3)+f_2(\vect{g})$. For every other values of $q_1,q_2,q_4$, $\phi=f_2(\vect{g})$. Similarly, when $q_5\oplus q_6=q_4=0$, then $\phi=f_1(\vect{h})-(-1)^{q_5}4g_1$ and when $q_5\oplus q_6=q_4=1$, then $\phi=f_1(\vect{h})-(-1)^{q_5}4(g_1+g_3)$. For every other values of $q_5,q_6,q_4$, $\phi=f_1(\vect{h})$. A quantum circuit simulating $e^{-iH_{21}t}$ in this case has been shown in Figure \ref{fig:45}b.

\paragraph{Circuit for simulating $e^{-iH_{20}t}$ : }
An eigenbasis for the Paulis in $G_{20}$ has been shown in Lemma \ref{app:lem:ebasisG0} of Supplementary Method 2 (Appendix \ref{app:2overlap}). But since we have been unable to derive a diagonalizing circuit, so we divide the commuting Paulis into two groups of 4-qubit Paulis as follows. 
\begin{eqnarray}
 G_{201}&=&\{P_kP_lP_iP_j\id\id : i+j,k+l\equiv 0\mod 2.\} \nonumber \\
 G_{202}&=&\{\id\id P_iP_jP_kP_l : i+j,k+l\equiv 0\mod 2.\}    \nonumber
\end{eqnarray}
We get the following two Hamiltonians.
\begin{eqnarray}
 H_{201}&=&a_0XXXX\id\id+a_1YYXX\id\id+a_6XXYY\id\id+a_7YYYY\id\id \nonumber \\
 H_{202}&=&b_0\id\id XXXX+b_6\id\id YYXX+b_1\id\id XXYY+b_7\id\id YYYY  \nonumber
\end{eqnarray}
Using the diagonalizing circuit of \cite{2017_RWSWT} and have the following.
\begin{eqnarray}
e^{-iH_{201}t}&=&W_{01}e^{-ia_0 \id\id Z\id \id\id t}e^{ia_1ZZZ\id \id\id t}e^{ia_6 \id\id ZZ \id\id t}e^{-ia_7 ZZZZ \id\id t}W_{01}^{\dagger} \nonumber\\
 e^{-iH_{202}t}&=&W_{02}e^{-ib_0 \id\id Z\id\id\id t}e^{ib_6\id\id ZZ\id\id t}e^{ib_1 \id\id Z\id ZZ t}e^{-ib_7 \id\id ZZZZ t}W_{02}^{\dagger}\nonumber 
\end{eqnarray}
where $W_{01}=CNOT_{(3,1)}CNOT_{(3,2)}CNOT_{(3,4)}H_{(3)}$ and $W_{02}=CNOT_{(3,4)}CNOT_{(3,5)}CNOT_{(3,6)}H_{(3)}$, where the rightmost gate is the first one to be applied. We denote the state of the qubits $q_1,\ldots,q_4$ and $q_3,\ldots,q_6$ after the application of $W_{01}$ and $W_{02}$ by the variables $x_1,\ldots,x_4$ and $x_3',\ldots,x_6'$ respectively. We have the following expression for the overall phase incurred between $W_{01}$, $W_{01}^{\dagger}$ and between $W_{02}$, $W_{02}^{\dagger}$.
\begin{eqnarray}
 \phi_1&=&-(-1)^{x_3}a_0t+(-1)^{x_3\oplus x_1\oplus x_2}a_1t+(-1)^{x_3\oplus x_4}a_6t-(-1)^{x_3\oplus x_4\oplus x_1\oplus x_2}a_7t     \nonumber \\
  \phi_2&=&-(-1)^{x_3'}b_0t+(-1)^{x_3'\oplus x_4'}b_6t+(-1)^{x_3'\oplus x_5'\oplus x_6'}b_1t-(-1)^{x_3'\oplus x_4'\oplus x_5'\oplus x_6'}b_7t     \nonumber
\end{eqnarray}
In all the cases considered below it is easy to verify that $\phi_{1,\conj{x_3}}=-\phi_{1,x_3}$ and $\phi_{2,\conj{x_3'}}=-\phi_{2,x_3'}$. So it is enough to consider $x_3=x_3'=0$.
\paragraph{Case I : } Assume $a_1 t=a_6t=-\theta_1$, $b_1t=b_6t=-\theta_2$, $a_0t=a_7t=\theta_1$ and $b_0t=b_7t=\theta_2$. If $x_1\oplus x_2=x_4=0$ then $\phi_1=-4\theta_1$, else it is 0. Similar conclusions follow for $\phi_2$ if we replace $x_1,x_2,x_3,x_4$ by $x_6',x_5',x_3',x_4'$ respectively. 

\paragraph{Case II : } Let $a_0 t=a_1t=a_6t=a_7t=\theta_1$, $b_0t=b_1t=b_6t=b_7t=\theta_2$. If $x_1\oplus x_2=x_4=1$ then $\phi_1=-4\theta_1$, else it is 0. Similar conclusions follow for $\phi_2$ if we replace $x_1,x_2,x_3,x_4$ by $x_6',x_5',x_3',x_4'$ respectively.

\paragraph{Case III : } Let $a_0t=a_7t=-h_1-h_2+h_3, a_1t=a_6t=h_1-h_2+h_3$ and $b_0t=b_7t=-g_1-g_2+g_3, b_1t=b_6t=g_1-g_2+g_3$. If $x_1\oplus x_2=x_4=0$ then $\phi_1=4h_1$ and if $x_1\oplus x_2=x_4=1$ then $\phi_2=2(h_2-h_3)$. Similar conclusions follow for $\phi_2$ if we replace $x_1, x_2, x_4$ by $x_6', x_5', x_4'$ respectively.

Circuits simulating $e^{-iH_{201}t}$ in Case I, II and III have been shown in Figure \ref{fig:45}c, \ref{fig:45}d and \ref{fig:45}e respectively. Circuits for $e^{-iH_{202}t}$ are similar. Circuit for $e^{-iH_{20}}$ in each of these cases is obtained by concatenating the corresponding circuits.

\paragraph{Overlap on 3 qubits : }
Now we consider the case when there is overlap on 3 qubits. We can have the following sets of commuting Paulis.
\begin{eqnarray}
 G_{3y}&=&\{YP_iP_jP_k\id, \id P_iP_jP_kY:i+j+k\equiv 1\mod 2\}    \label{eqn:G3y} \\
 G_{3x}&=&\{XP_iP_jP_k\id, \id P_iP_jP_kX:i+j+k\equiv 0\mod 2\}    \label{eqn:G3x} 
\end{eqnarray}
Without loss of generality, we assume that the leftmost operator acts on qubit $q_1$, next one on $q_2$ and so on - rightmost one acts on qubit $q_5$. We denote a state vector as $\ket{Q_1q_2q_3q_4Q_2}$ where $Q_1=\ket{q_1}$, $Q_2=\ket{q_5}$ and $q_1,\ldots,q_5\in\{0,1\}$. We can have the following Hamiltonian terms, expressed as sums of commuting Paulis from the above two sets.
\begin{eqnarray}
 H_{3y}&=&a_1YYXX\id+a_2YXYX\id+a_3YXXY\id+a_7YYYY\id \nonumber \\
 &&+b_3\id YXXY+b_5\id XYXY+b_6\id XXYY+b_7\id YYYY  \label{eqn:h3y} \\
 H_{3x}&=&a_0XXXX\id+a_4XYYX\id+a_5XYXY\id+a_6XXYY\id \nonumber \\
 &&+b_0\id XXXX+b_1\id YYXX+b_2\id YXYX+b_4\id XYYX \label{eqn:h3x} 
\end{eqnarray}

\paragraph{Circuit for simulating $e^{-iH_{3y}t}$ : }
Let $W_{3y}$ be the unitary consisting of the following sequence of gates. The rightmost one is the first to be applied. With a slight abuse of notation we denote $CNOT_{(c,t_1)}CNOT_{(c,t_2)}CNOT_{(c,t_3)}\ldots$ by $CNOT_{(c;t_1,t_2,t_3,\ldots)}$ (multi-target CNOT).
\begin{eqnarray}
 W_{3y}=CNOT_{(2;1,3,4)}H_{(2)}Z_{(2)}CNOT_{(2;1,5)}H_{(2)}CNOT_{(2,1)}    \nonumber
\end{eqnarray}
\begin{theorem}
For each $i,j,k\in\intg_2$, such that $YP_iP_jP_k\id, \id P_iP_jP_kY  \in G_{3y}$ we have the following.
 \begin{eqnarray}
  \sqrt{-1}^{i+j+k+1}W_{3y}\left(Z_{(1)} Z_{(2)}Z_{(3)}^jZ_{(4)}^k\id\right)W_{3y}^{\dagger}=YP_iP_jP_k\id \nonumber \\
  \text{ and } \sqrt{-1}^{i+j+k+1}W_{3y}\left(\id Z_{(2)}Z_{(3)}^jZ_{(4)}^k Z_{(5)}\right)W_{3y}^{\dagger}=\id P_iP_jP_kY \nonumber 
 \end{eqnarray}
 \label{thm:diagG5y}
\end{theorem}
The proof is similar to Theorem \ref{thm:diagGy} and has been shown in Supplementary Method 3 (Appendix \ref{app:3overlap}). Thus we have the following.
\begin{eqnarray}
 e^{-iH_{3y}t} &=&e^{-i(-a_1W_{3y}(Z Z\id \id\id)W_{3y}^{\dagger}-a_2W_{3y}(Z ZZ\id\id)W_{3y}^{\dagger}-a_3W_{3y}(Z Z\id Z\id)W_{3y}^{\dagger}+a_7W_{3y}(Z ZZZ\id)W_{3y}^{\dagger})t} \nonumber   \\
 &&\cdot e^{-i(-b_3W_{3y}(\id Z\id \id Z)W_{3y}^{\dagger}-b_5W_{3y}(\id ZZ\id Z)W_{3y}^{\dagger}-b_6W_{3y}(\id Z\id ZZ)W_{3y}^{\dagger}+b_7W_{3y}(\id Z ZZZ)W_{3y}^{\dagger})t} \nonumber \\
 &=&W_{3y}e^{ia_1ZZ\id \id\id t}e^{ia_2ZZZ\id\id t}e^{ia_3ZZ\id Z \id t}e^{-ia_7ZZZZ\id t} 
  e^{ib_3\id Z\id \id Z t}e^{ib_5\id ZZ\id Z t}e^{ib_6\id Z\id ZZ t}e^{-ib_7\id ZZZZ t} W_{3y}^{\dagger}\nonumber 
\end{eqnarray}
We denote the state of the qubits $q_1,\ldots,q_5$ after the application of $W_{3y}$ by the variables $x_1,\ldots,x_5$ respectively. We have the following expression for the overall phase incurred between $W_{3y}$ and $W_{3y}^{\dagger}$.
\begin{eqnarray}
 \phi&=&(-1)^{x_2\oplus x_1}a_1+(-1)^{x_2\oplus x_1\oplus x_3}a_2+(-1)^{x_2\oplus x_1\oplus x_4}a_3-(-1)^{x_2\oplus x_1\oplus x_3\oplus x_4}a_7  \nonumber\\
 &&+(-1)^{x_2\oplus x_5}b_3+(-1)^{x_2\oplus x_5\oplus x_3}b_5+(-1)^{x_2\oplus x_5\oplus x_4}b_6-(-1)^{x_2\oplus x_5\oplus x_3\oplus x_4}b_7  \nonumber
\end{eqnarray}

\begin{figure}
 \centering
 \includegraphics[width=\textwidth]{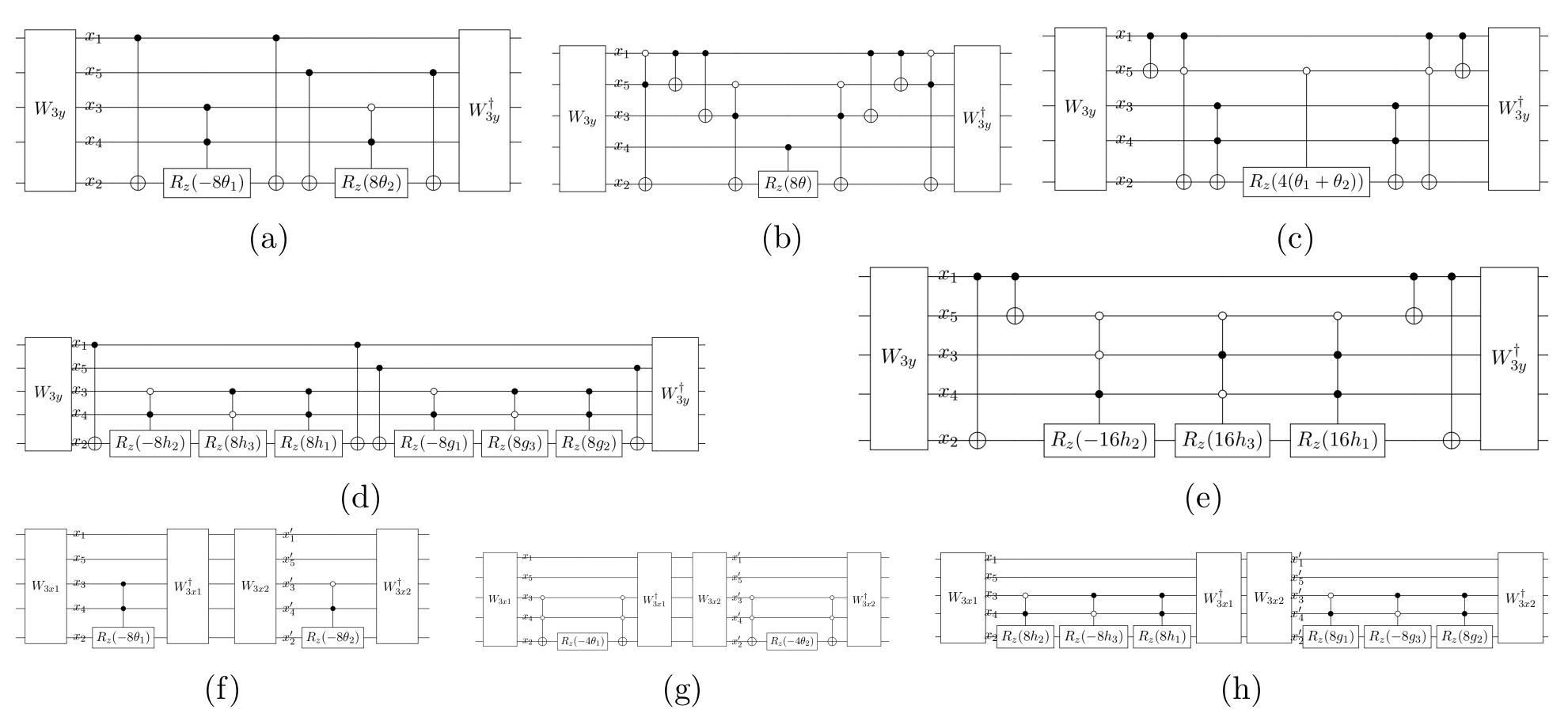}
 \caption{\textbf{Quantum circuit for $e^{-iH_{3y}t}$ and $e^{-iH_{3x1}t}$ : } (a)-(e): Circuit simulating $e^{-iH_{3y}t}$ (a) when $a_1t=-\theta_1, b_6t=-\theta_2$, $a_2t=a_3t=a_7t=\theta_1$ and $b_3t=b_5t=b_7t=\theta_2$; (b) with the coefficient values in (a), except here $\theta_1=\theta_2$; (c) when $a_1t=a_2t=a_3t=a_7t=\theta_1$ and $b_3t=b_5t=b_6t=b_7t=\theta_2$; (d) when the coefficients are as in Equation 18; (e) Circuit with the coefficients in (d), except in this case $h_2=g_1, h_3=g_3, h_1=g_2$. (f)-(h): Circuit simulating $e^{-iH_{3x1}t}$ (f) when $a_6t=-\theta_1, b_1t=-\theta_2, a_0t=a_5t=a_4t=\theta_1$ and $b_0t=b_2t=b_4t=\theta_2$; (g) when $a_0t=\ldots=a_7t=\theta_1$ and $b_0t=\ldots=b_7t=\theta_2$; (h) when the coefficients are as in Equation 18. }
\label{fig:67}
\end{figure}
It is easy to verify that $\phi_{\conj{x_2}}=-\phi_{x_2}$. We consider the following cases and it is sufficient to check the phase values when $x_2=0$.
\paragraph{Case I : } We consider the case when $a_1 t=-\theta_1$, $b_6t=-\theta_2$, $a_2t=a_3t=a_7t=\theta_1$ and $b_3t=b_5t=b_7t=\theta_2$. We can write $\phi=f_1(\theta_1)+f_2(\theta_2)$. It can be verified that $\phi_{\conj{x_1},x_5}=-f(\theta_1)+f(\theta_2)$ and $\phi_{x_1,\conj{x_5}}=f(\theta_1)-f(\theta_2)$. So we concentrate on $x_1=x_5=0$. If $x_3=x_4=1$ then $\phi=-4\theta_1$ and if $x_3=0, x_4=1$ then $\phi=4\theta_2$. A quantum circuit simulating $e^{-iH_{3y}t}$ has been shown in Figure \ref{fig:67}a. If $\theta_1=\theta_2$ then we can have a further reduction of controlled rotation gates, as shown in Figure \ref{fig:67}b.

\paragraph{Case II : } Next we consider the case when $a_1 t=a_2t=a_3t=a_7t=\theta_1$, $b_6t=b_3t=b_5t=b_7t=\theta_2$. In this case $\phi=0$ whenever $x_1\oplus x_5=1$. Else, as before $\phi_{\conj{x_1},x_5}=-f(\theta_1)+f(\theta_2)$ and $\phi_{x_1,\conj{x_5}}=f(\theta_1)-f(\theta_2)$. So it is enough to consider $x_1=x_5=0$. When $x_3=x_4=1$ then $\phi=-2(\theta_1+\theta_2)$, else $\phi=2(\theta_1+\theta_2)$. Thus we can have a quantum circuit simulating $e^{-iH_{3yt}}$, as shown in Figure \ref{fig:67}c.

\paragraph{Case III : } Now we consider the case when $a_1t=h_1-h_2+h_3, a_2t=-h_1-h_2-h_3, a_3t=-h_1+h_2+h_3, a_7t=-h_1-h_2+h_3$ and $b_3t=-g_1+g_2+g_3, b_5t=-g_1-g_2-g_3, b_6t=g_1-g_2+g_3, b_7t=-g_1-g_2+g_3$. If we denote $\vect{h}=(h_1,h_2,h_3)$ and $\vect{g}=(g_1,g_2,g_3)$, then we can write $\phi=f(\vect{h})+f(\vect{g})$. Here too, $\phi_{\conj{x_1},x_5}=-f(\vect{h})+f(\vect{g})$ and $\phi_{x_1,\conj{x_5}}=f(\vect{h})-f(\vect{g})$. So let us consider $x_1=x_5=0$. Then we have the following phase values.
 \begin{eqnarray}
  \phi_{x_3=x_4=0}=0, \quad  \phi_{x_3=0,x_4=1}=-4(h_2+g_1), \quad \phi_{x_3=1,x_4=0}=4(h_3+g_3),
  \quad 
  \phi_{x_3=x_4=1}=4(h_1+g_2) \nonumber
 \end{eqnarray}
A circuit simulating $e^{-iH_{3yt}}$ in this case has been shown in Figure \ref{fig:67}d. If $h_2=g_1, h_3=g_3, h_1=g_2$ then we can have a simpler circuit, as shown in Figure \ref{fig:67}e.

\paragraph{Circuit for simulating $e^{-iH_{3x}t}$ : }
The diagonalizing transformation for the Pauli operators in $G_{3x}$ is shown in Lemma \ref{app:lem:ebasisG5x} of Supplementary Method 3 (Appendix \ref{app:3overlap}). Since we have been unable to find a diagonalizing circuit, so we divide the commuting Paulis into two groups of 4-qubit Paulis,
\begin{eqnarray}
 G_{3x1}&=&\{XP_iP_jP_k\id : i+j+k\equiv 0\mod 2.\} \nonumber \\
 G_{3x2}&=&\{\id P_iP_jP_kX : i+j+k\equiv 0\mod 2.\}    \nonumber
\end{eqnarray}
and have the following two Hamiltonians.
\begin{eqnarray}
 H_{3x1}&=&a_0XXXX\id+a_4XYYX\id+a_5XYXY\id+a_6XXYY\id \nonumber \\
 H_{3x2}&=&b_0\id XXXX+b_1\id YYXX+b_2\id YXYX+b_4\id XYYX \nonumber
\end{eqnarray}
Using the diagonalizing circuit of \cite{2017_RWSWT} and have the following.
\begin{eqnarray}
e^{-iH_{3x1}t}&=&W_{3x1}e^{-ia_0 \id Z\id\id \id t}e^{ia_6\id ZZZ \id t}e^{ia_5 \id Z\id Z \id t}e^{ia_4 \id ZZ\id \id t}W_{3x1}^{\dagger} \nonumber\\
 e^{-iH_{3x2}t}&=&W_{3x2}e^{-ib_0 \id Z\id\id\id t}e^{ib_1 \id ZZ\id\id t}e^{ib_2 \id Z\id Z\id t}e^{ib_4 \id ZZZ\id t}W_{3x2}^{\dagger}\nonumber 
\end{eqnarray}
where $W_{3x1}=CNOT_{(2;1,3,4)}H_{(2)}$ and $W_{3x2}=CNOT_{(2;3,4,5)}H_{(2)}$, where the rightmost gate is the first one to be applied.
We denote the state of the qubits $q_1,\ldots,q_4$ and $q_2,\ldots,q_5$ after the application of $W_{3x1}$ and $W_{3x2}$ by the variables $x_1,\ldots,x_4$ and $x_2',\ldots,x_5'$ respectively. We have the following expression for the overall phase incurred between $W_{3x1}$, $W_{3x1}^{\dagger}$ and between $W_{3x2}$, $W_{3x2}^{\dagger}$.
\begin{eqnarray}
 \phi_1&=&-(-1)^{x_2}a_0t+(-1)^{x_2\oplus x_3\oplus x_4}a_6t+(-1)^{x_2\oplus x_4}a_5t+(-1)^{x_2\oplus x_3}a_4t     \nonumber \\
  \phi_2&=&-(-1)^{x_2'}b_0t+(-1)^{x_2'\oplus x_3'}b_1t+(-1)^{x_2'\oplus x_4'}b_2t+(-1)^{x_2'\oplus x_3'\oplus x_4'}b_4t     \nonumber
\end{eqnarray}
In all the cases considered below it is easy to verify that $\phi_{1,\conj{x_2}}=-\phi_{1,x_2}$ and $\phi_{2,\conj{x_2'}}=-\phi_{2,x_2'}$. So it is enough to consider $x_2=x_2'=0$.

\paragraph{Case I : } We consider the case when $a_6 t=-\theta_1$, $b_1t=-\theta_2$, $a_0t=a_5t=a_4t=\theta_1$ and $b_0t=b_4t=b_2t=\theta_2$. 
$\phi_1=-4\theta_1$ when $x_3=x_4=1$, else it is 0. And $\phi_2=-4\theta_2$ when $x_3'=0, x_4'=1$, else it is 0. 

\paragraph{Case II : } Let $a_0 t=a_6t=a_5t=a_4t=\theta_1$, $b_0t=b_4t=b_2t=b_1t=\theta_2$. $\phi_1=2\theta_1$ when $x_3=x_4=0$, else $\phi_1=-2\theta_1$. Similarly for $\phi_2$.

\paragraph{Case III : } Assume $a_0t=-h_1-h_2+h_3, a_6t=h_1-h_2+h_3, a_5t=-h_1-h_2-h_3, a_4t=-h_1+h_2+h_3, $ and $b_0t=-g_1-g_2+g_3, b_4t=-g_1+g_2+g_3, b_2t=-g_1-g_2-g_3, b_1t=g_1-g_2+g_3$. We have the following phases.
 \begin{eqnarray}
  \phi_1(x_3=0,x_4=1)=4h_2; &\quad& \phi_1(x_3=1,x_4=0)=-4h_3;\quad \phi_1(x_3=x_4=1)=4h_1;  \nonumber \\
  \phi_2(x_3'=0,x_4'=1)=4g_1, &\quad& \phi_2(x_3'=1,x_4'=0)=-4g_3;\quad \phi_2(x_3'=x_4'=1)=4g_2;    \nonumber
 \end{eqnarray}

Circuits simulating $e^{-iH_{3x1}t}$ in Case I, II and III have been shown in Figure \ref{fig:67}f, \ref{fig:67}g and \ref{fig:67}h respectively. Circuits for $e^{-iH_{3x2}t}$ are similar. Circuit for $e^{-iH_{3x}t}$ in each case is obtained by concatenating the corresponding circuits. 

\subsection{Circuit for arbitrary exponentiated Hamiltonians}
\label{sec:compileArbit}
Our previous discussion focuses on the case of fermionic simulation within a Jordan-Wigner representation using Hamiltonian terms that are fermionically swapped to be adjacent to each other.  While these simulation circuits are among the most important for applications in chemistry, it does not necessarily represent all cases of physical interest let alone chemistry.
Here we address this by discussing ways to synthesize circuits for arbitrary exponentiated Hamiltonians in $\mathbb{C}^{2^n \times 2^n}$, expressible as sum of Pauli operators, with an aim to reduce the number of non-Clifford resources. For reasons discussed previously, it is enough to consider a Hamiltonian $H$ expressed as sum of commuting Pauli operators.
\begin{eqnarray}
 H=\sum_{i}\alpha_iP_i\qquad P_i\in\pauli_n \nonumber
\end{eqnarray}
In most cases one synthesizes circuit for each $e^{-i\alpha_iP_it}$ using a number of CNOT and one $R_z$ gate. Thus the number of $R_z$ gates required is equal to the number of summands. Here we describe procedure to synthesize circuit for $e^{-iHt}$ i.e. considering multiple summands or Pauli operators.

We diagonalize $H$, for example, by using the algorithms in \cite{2020_BT}. In the previous section we have constructed explicit eigenbases for the diagonalization of some specific Hamiltonians. Then we get the following.
\begin{eqnarray}
 H=W\left(\sum_i\alpha_i'Q_i\right)W^{\dagger} \nonumber
\end{eqnarray}
Here $Q_i=\otimes_{j=1}^nQ_{ij}$, a tensor product of Z and $\id$ i.e. $Q_{ij}\in\{Z,\id\}$. $W$ is a diagonalizing Clifford circuit. Thus we get the following.
\begin{eqnarray}
 e^{-iHt}=We^{\sum_i\alpha_i'Q_i}W^{\dagger}
\end{eqnarray}

\begin{lemma}
 Let $\mathcal{H}=\sum_i\alpha_i'Q_i$, such that $Q_i=\otimes_{j=1}^nQ_{ij}$, where $Q_{ij}\in\{Z,\id\}$. With each $Q_i$ we associate an $n$-length vector $\vect{y}_i=(y_{i1},\ldots,y_{in})\in\{0,1\}^n$ such that $(\vect{y}_i)_j=y_{ij}=1$ if $Q_{ij}=Z$, else $y_{ij}=0$. Let $x_1,\ldots,x_n\in\{0,1\}$ and $\ket{0}=[1,0]^T$, $\ket{1}=[0,1]^T$. The eigenvectors of $\mathcal{H}$ are of the form $\ket{v}=\bigotimes_{j=1}^n\ket{x_j}$ and the corresponding eigenvalue is
 \begin{eqnarray}
  \phi_v=\sum_i\alpha_i' (-1)^{\oplus_{j=1}^n y_{ij}x_j}  \label{eqn:phi_v}
 \end{eqnarray}
 \label{lem:arbitSpectrum}
\end{lemma}
\begin{proof}
 The summands in $\mathcal{H}$ are mutually commuting and so they have a common eigenbasis. Let us first consider $Q_i$. Since $Q_{ij}\ket{x}=\ket{x}$ if $Q_{ij}=\id$ and  $Q_{ij}\ket{x}=(-1)^x\ket{x}$ if $Q_{ij}=Z$, where $x\in\{0,1\}$, so we have the following. 
 \begin{eqnarray}
  Q_i\ket{v}&=&\left(\bigotimes_{j=1}^nQ_{ij}\right)\left(\bigotimes_{j=1}^n\ket{x_j}\right)=\bigotimes_{j=1}^nQ_{ij}\ket{x_j}=(-1)^{\oplus_{j=1}^ny_{ij}x_j}  \nonumber
 \end{eqnarray}
This implies that
\begin{eqnarray}
 \mathcal{H}\ket{v}=\left(\sum_i\alpha_i'Q_i\right)\ket{v}=\sum_i\alpha_i'(-1)^{\oplus_{j=1}^ny_{ij}x_j}\ket{v}.  \nonumber
\end{eqnarray}
\end{proof}
We can also interpret $\phi$ as the
overall phase incurred between $W$ and $W^{\dagger}$.
For given values of $\alpha_i'$, $\phi$ is an $n$-variable Boolean function where $x_j$ are the Boolean variables. We can evaluate a truth table and get all the $2^n$ values of $\phi$ for different values of $(x_1,\ldots,x_n)\in\{0,1\}^n$. For each distinct non-zero absolute phase value $|\theta|$ (ignoring sign), we can have a sub-circuit $\mathcal{C}_{|\theta|}$ that has only one controlled rotation $cR_z(2\theta)$ gate. The complete circuit can be obtained by combining these different sub-circuits (one for each $|\theta|\neq 0$), in between the diagonalizing Clifford circuits $W, W^{\dagger}$. The ordering of the sub-circuits do not matter.

\begin{figure}
 \centering
 \includegraphics[width=0.4\textwidth]{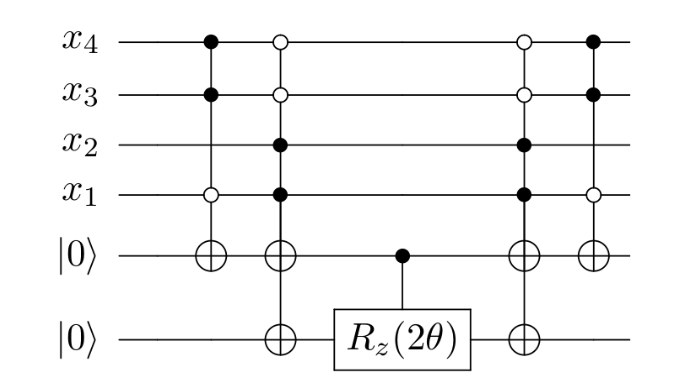}
  \caption{\textbf{Circuit $\mathcal{C}_{|\theta|}$ : } The circuit when $\mathcal{M}_{\theta}=\{(1,1,0,0),(1,1,1,0)\}$ and $\mathcal{M}_{-\theta}=\{(0,0,1,1)\}$.}
 \label{fig:08}
\end{figure}
Now we discuss how to synthesize sub-circuit $\mathcal{C}_{|\theta|}$, for one such distinct absolute value of $\phi$. Let $\mathcal{M}_{\theta}$ be the set of binary values for variables $x_1, x_2, \ldots, x_n$, such that $\phi$ computes to $\theta$ in Equation \ref{eqn:phi_v}.
\begin{eqnarray}
 \mathcal{M}_{\theta}=\{(x_1,\ldots,x_n)\in\{0,1\}^n : \phi_{x_1,\ldots,x_n}=\theta \} 
\end{eqnarray}
Analogously we can define $\mathcal{M}_{-\theta}$. 
 We can also associate $\mathcal{M}_{\theta}$ and $\mathcal{M}_{-\theta}$ with $S_{\theta}$ and $S_{-\theta}$, the sets of eigenvectors with eigenvalues $\theta$ and $-\theta$ respectively, as obtained from Lemma \ref{lem:arbitSpectrum}. We define the following operators, which acts on the input vector space and the space of two ancillae - $c$ and $r$, the latter being initialized to $0$.
\begin{eqnarray}
 V_{\theta}&=&\sum_{\ket{v}\in S_{\theta}}\ket{v,c\oplus 1, 0}\bra{v,c,0}+\sum_{\ket{w}\notin S_{\theta}}\ket{w,c,0}\bra{w,c,0} \nonumber \\
 V_{-\theta}&=&\sum_{\ket{v}\in S_{-\theta}}\ket{v,c\oplus 1, 1}\bra{v,c,0}+\sum_{\ket{w}\notin S_{-\theta}}\ket{w,c,0}\bra{w,c,0} \label{eqn:op}
\end{eqnarray}
The circuit $\mathcal{C}_{|\theta|}=V_{\theta}V_{-\theta}(cR_Z(2\theta))_{cr}V_{-\theta}^{\dagger}V_{\theta}^{\dagger}$. If the input vector is in $S_{\theta}$ or $S_{-\theta}$ then both these operators flip a control ancilla qubit ($c$). Additionally, if the vector is in $S_{-\theta}$ then the second ancilla $r$ is flipped. We apply a $cR_z(2\theta)$ gate on $r$, controlled on $c$. Thus if the input vector is in $S_{-\theta}$ then we actually apply $cR_z(-2\theta)$. The ancillae $c,r$ can be controlled by multi-controlled X gates, that can be further decomposed in terms of Toffoli and CNOT gates \cite{1995_BBCetal, 2022_SP}.  
For example, let $\mathcal{M}_{\theta}=\{(0,0,1,1),(0,1,1,1)\}$ and $\mathcal{M}_{-\theta}=\{(1,1,0,0)\}$. The two Boolean min-terms of $\mathcal{M}_{\theta}$ can be compressed to have a single term because when $x_1=0$, $x_3=x_4=1$ then $\phi=\theta$, irrespective of the value of $x_2$. We call it the 'don't care condition' for $x_2$. So, equivalently we can write $\mathcal{M}_{\theta}=\{(0,*,1,1)\}$, where $*$ denotes the don't-care condition. In general, algorithms like Karnaugh map \cite{1953_K}, ESPRESSO \cite{1982_BHHetal} can be used to get compact set of Boolean min-terms. A circuit $\mathcal{C}_{|\theta|}$ has been shown in Figure \ref{fig:08}. 

Hence, due to the invariance of the point spectrum of unitarily equivalent operators we have the following. 
\begin{lemma}
 Let $H=\sum_i\alpha_iP_i$, where $P_i$ are mutually commuting $n$-qubit Pauli operators. We can implement a circuit for $e^{-iHt}$ with at most $m$ (controlled)-rotations, where $m$ is the number of distinct non-zero eigenvalues
 (ignoring sign) of $H$.
 \label{lem:rot}
\end{lemma}


\paragraph{Illustration - Quantum Heisenberg and Quantum Ising model :}
We consider the problem of designing quantum circuits for simulating the quantum Heisenberg and Ising model with Hamiltonians $H_H$ and $H_I$ respectively. The Heisenberg Hamiltonian is widely used to study magnetic systems, where the magnetic spins are treated quantum mechanically \cite{1999_F, 2006_MI, 2008_S, 2021_GPAG, 2021_PS}.
Let $G=(E,V)$ be the underlying graph with the vertex and edge set being $V$ and $E$ respectively. 
\begin{eqnarray}
 H_H&=&\sum_{(i,j)\in E}\left(J_xX_{(i)}X_{(j)}+J_yY_{(i)}Y_{(j)}+J_zZ_{(i)}Z_{(j)}\right)+\sum_{i\in V}d_hZ_{(i)}   \label{eqn:Hh} \\
 H_I&=&\sum_{(i,j)\in E}J_zZ_{(i)}Z_{(j)}+\sum_{i\in V}d_h'Z_{(i)}   \label{eqn:Hi}
\end{eqnarray}
In the above $J_x,J_y,J_z$ are coupling parameters, denoting the exchange interaction between nearest neighbor spins along the X,Y,Z direction respectively. $d_h, d_h'$ is the time amplitude of the external magnetic field along the Z-direction. One set of commuting Paulis are $\{X_{(i)}X_{(j)}:(i,j)\in E\}$, $\{Y_{(i)}Y_{(j)}:(i,j)\in E\}$, $\{Z_{(i)}Z_{(j)}:(i,j)\in E\}$ 
and $\{Z_{(i)}:i\in V\}$.

Let us first consider the set $\{Z_{(i)}:i\in V\}$. Following the previous discussions and Lemma \ref{lem:arbitSpectrum}, the overall phase incurred or the eigenvalues are as follows.
\begin{eqnarray}
 \phi'=d_h\sum_{i\in V}(-1)^{x_i}\qquad x_i\in\{0,1\}
\end{eqnarray}
For $\vect{x}\in\{0,1\}^{|V|}$, one particular assignment of values to the Boolean variables, let $T_0=\{i\in V:x_i=0\}$ and $T_1=\{i\in V:x_i=1\}$. So $|T_0|+|T_1|=|V|$ and 
\begin{eqnarray}
 \phi_{\vect{x}}'=d_h\left(|T_0|-|T_1|\right)=d_h\left(|V|-2|T_1|\right).
\end{eqnarray}
So the number of distinct non-zero eigenvalues or absolute values of $\phi'$ can be $\lceil |V|/2\rceil$. Implementing each $e^{-id_hZ_{(i)}t}$ would require $|V|$ rotation gates. Thus, using Lemma \ref{lem:rot}, we have about $50\%$ reduction in the rotation gate cost. 

Now, let us consider the other commuting sets.
Since $H^{\dagger}XH=Z$ and $(HSX)^{\dagger}Y(HSX)=Z$, so each of the above sets can be diagonalized and we can focus on the problem of simulating a quantum circuit for the Hamiltonian : $H=J\sum_{(i,j)\in E}Z_{(i)}Z_{(j)}$, where $J$ is a constant. Our aim is to derive an upper bound on the number of controlled rotations required to simulate $e^{-iHt}$. Following the previous discussions, the overall phase incurred between the diagonalizing Cliffords $W, W^{\dagger}$ is as follows (Lemma \ref{lem:arbitSpectrum}).
\begin{eqnarray}
 \phi&=& J\sum_{(i,j)\in E}(-1)^{x_i\oplus x_j}
\end{eqnarray}
where $x_i, x_j\in \{0,1\}$ are variables denoting the state of the qubits after application of $W$. The quantum circuit has $|V|$ qubits, corresponding to each vertex of $G$. Let $\vect{x}=\in \{0,1\}^{|V|}$ denote one particular assignment of values to the variables $x_1,\ldots, x_{|V|}$. $S_0=\{(i,j)\in E : x_i=x_j=1\}$ and $S_1=\{(i,j)\in E : x_i \text{ or } x_j\text{ is }1\}$. If $S'=\{(i,j)\in E : x_i=x_j=0\}$, then $|S'|=|E|-|S_0|-|S_1|$. Let $\phi_{\vect{x}}$ be the value of the phase for this particular assignment. 
\begin{eqnarray}
 \phi_{\vect{x}}&=&J\left(|S'|+|S_0|-|S_1|\right)=J\left(|E|-|S_0|-|S_1|+|S_0|-|S_1|\right)=J\left(|E|-2|S_1|\right)
\end{eqnarray}
Let $V_{1\vect{x}}=\{i\in V:x_i=1\}$, $V_{0\vect{x}}=\{i\in V:x_i=0\}$ and $\mathcal{N}(k)$ be the set of neighbouring vertices of $k$ in G. Then
\begin{eqnarray}
 |S_1|=\sum_{k\in V_{1\vect{x}}}\left|\left\{\mathcal{N}(k)\setminus V_{1\vect{x}}\right\}\right|
\end{eqnarray}

Now for any assignment $|S_1|$ can vary from $1,\ldots, |E|$. So the number of distinct values of $\phi$ is at most $\lceil|E|/2\rceil$. And hence we need at most $\lceil|E|/2\rceil$ controlled-$R_z$ gates in the circuit simulating $e^{-iHt}$. Had we simulated each $e^{-iJZ_{(i)}Z_{(j)}t}$, we would have required $|E|$ $R_z$ gates. So we can achieve about $50\%$ reduction in the rotation gate cost under the assumption that controlled-$R_z$ costs the same to implement as a single $R_z$ gate.
\begin{figure}
 \centering
 \includegraphics[width=\textwidth]{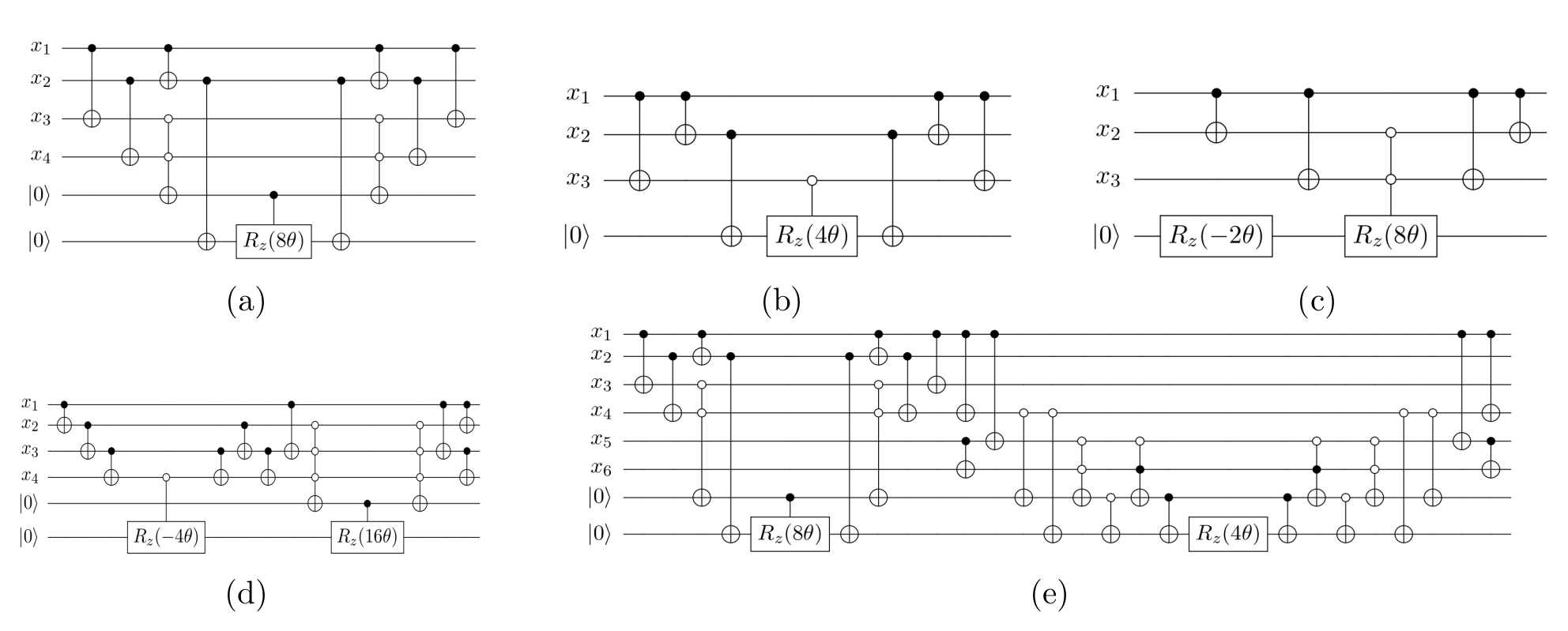}
 \caption{\textbf{Quantum circuit simulating $e^{it\theta\sum_{(i,j)\in E}Z_{(i)}Z_{(j)}}$ : } When the underlying graph $G=(V,E)$ is (a) 4-point circle, (b) 3-point line, (c) triangle, (d) 4-point complete graph, (e) 6-point circle.}
\label{fig:09}
\end{figure}
\paragraph{G is a cycle : } This is basically the traslationally invariant 1-D spin chain. Let $G_{V_{1\vect{x}}}$ be the subgraph induced by $V_{1\vect{x}}$, which is a union of paths. For each path $p$, let $S_{1p}=\bigcup_{k\in p}\{\mathcal{N}(k)\setminus V_{1\vect{x}}\}\subseteq S_1$ be the set of vertices in this path.
Each of the terminal vertices has one neighbour in $V\setminus V_{1\vect{x}}$. So $|S_{1p}|= 2$. Thus if $\mathcal{P}$ is the set of all such paths in $G_{V_{1\vect{x}}}$, then
\begin{eqnarray}
 \phi_{\vect{x}}&=&J\left(|E|-2\sum_{p\in\mathcal{P}}|S_{1p}|\right)=J\left(|E|-4|\mathcal{P}|\right)
\end{eqnarray}
Now $|\mathcal{P}|$ can vary from $1,\ldots,\lceil|E|/4\rceil$ to give $\lceil|E|/4\rceil$ distinct values of $|\phi|$. This implies a quantum circuit synthesizing $e^{-iHt}$ will require at most $\lceil|E|/4\rceil$ $cR_z$ gates. This is about $75\%$ reduction in the cost of rotation gates, compared to synthesizing each $e^{-iJZ_{(i)}Z_{(j)}t}$.

\paragraph{G is a complete graph : } In this case for each $k\in V_{1\vect{x}}$, we have $\mathcal{N}(k)\setminus V_{1\vect{x}}=V_{0\vect{x}}$. So we have,
\begin{eqnarray}
 \phi_{\vect{x}}=J\left(|E|-2|V_{1\vect{x}}||V_{0\vect{x}}|\right)=J\left(|E|-2\|\vect{x}\|_1(|V|-\|\vect{x}\|_1)\right)
\end{eqnarray}
So there can be $\lceil|V|/2\rceil$ distinct values of $|\phi_{\vect{x}}|$ as $\|x\|_1$ varies from $1,\ldots,\lceil|V|/2\rceil$. And hence we require at most $\lceil|V|/2\rceil$ $cR_z$ gates for simulating $e^{-iHt}$. If we simulate each $e^{-iJZ_{(i)}Z_{(j)}t}$ then we require $|E|=\frac{|V|(|V|-1)}{2}$ $R_z$ gates. This indicates about $100\left(1-\frac{2}{|V|-1}\right)\%$ reduction in the cost of rotation gates.

In Figure \ref{fig:09} we have shown quantum circuits simulating $e^{it\theta\sum_{(i,j)\in E} Z_{(i)}Z_{(j)}}$ for some simple graphs $G=(V,E)$. The circuits have been designed to optimize the number of Toffoli gates, as well.

\paragraph{Reducing the number of Toffoli gates : } We discussed before that the T-count from the Toffolis may be a significant factor in high error regime as the logarithmic cost of rotation synthesis may not dominate the additive constant that arises from the Toffoli gates needed. In order to reduce the number of Toffolis we can do the following. We design circuits reducing Toffolis for Hamiltonians over smaller graphs, such as in Figure \ref{fig:09}a-\ref{fig:09}e. Then we decompose a Hamiltonian over a large graph into Hamiltonians over these smaller graphs. 
For example, consider a 1-D cycle on $N$ points and $H_z=\theta\sum_{(i,j)\in E}Z_{(i)}Z_{(j)}$. We break this cycle into smaller chains of length 3 i.e. $H_z=\theta (Z_{(1)}Z_{(2)}+Z_{(2)}Z_{(3)})+(Z_{(3)}Z_{(4)}+Z_{(4)}Z_{(5)})+\ldots=\sum_i H_{zi}$. We have a quantum circuit that synthesizes each $e^{iH_{zi}t}$ with only one $cR_z$ gate (Figure \ref{fig:09}b). So to synthesize $e^{iH_zt}$ we require approximately $ N/2$ $cR_z$. This is about twice the number of controlled rotations required, had we synthesized without decomposing. But it does not require any extra Toffoli-pairs. We manage to get approximately $50\%$ reduction, compared to synthesizing each summand i.e. $e^{itZ_{(i)}Z_{(j)}}$. 

Now consider a large $N\times N$ lattice which has $N^2$ vertices and $2N(N-1)$ edges and the Hamiltonian $H_z=\theta\sum_{(i,j)\in E}Z_{(i)}Z_{(j)}$. We can decompose this into $(N-1)^2$ smaller interior cycles of 4-points and a bigger outer circle with $2N+2(N-2)=4(N-1)$ points. From Figure \ref{fig:09}a, we know that we can design a circuit simulating the exponentiated Hamiltonian corresponding to each interior cycle with 1 $cR_z$ and 1 Toffoli pair. We can further decompose the outer circle (as explained in the previous paragraph) and have a circuit with approximately $2(N-1)$ $cR_z$ gates. Thus we require $\approx (N-1)^2+2(N-1)=(N-1)(N+1)$ $cR_z$ and $(N-1)^2$ Toffoli-pairs. We have discussed before that for general graphs, number of $cR_z$ required is $\approx |E|/2=N(N-1)$, so we use $\approx (N-1)$ more $cR_z$ by decomposing, but the Toffoli cost reduces a lot. Had we synthesized each $e^{itZ_{(i)}Z_{(j)}}$, we would have used $2N(N-1)$ $R_z$. Thus we manage to get a reduction of $\approx (N-1)^2$ in the number of $R_z/cR_z$.

In Figure \ref{fig:09}e we gave a circuit for simulating $e^{it\theta\sum_{(i,j)\in E} Z_{(i)}Z_{(j)}}$, when the underlying graph is a 6-point cycle. We reduced the Toffoli-pairs by decomposing the graph into smaller cycles.  

\subsection{Application : Simulating with qDRIFT}
\label{sec:qdrift}

In this section we consider one simulation algorithm - qDRIFT.  We focus on qDRIFT rather than Trotter for our experiments because qDRIFT is easier to analyze numerically.  This is because Trotter errors subtly depend on operator ordering.  Specifically we consider a Hamiltonian $H=\sum_{j=1}^L h_jP_j$ and sample Pauli operators to apply in each short time step, as described earlier in the paper and in~\cite{2019_C}. We can assume each $h_j>0$, since the negation affects the angles of the rotation gates. In qDRIFT, in each iteration one Pauli term is sampled up to a total of $N$ samples. The probability of sampling $P_j$ is $h_j/ \sum_{i} h_i$ and is then simulated for a short time period. We consider another procedure where we re-write the Hamiltonian as $H=\sum_{j=1}^{L'}h_j'H_j$, where each $H_j=\sum_iP_i$, is a sum of commuting Paulis. In each iteration one $H_j$ is sampled with probability $\frac{h_j'}{\sum_ih_i'}$ and simulated for a short time period. Then we compare the growth of error, number of $R_z/cR_z$, Toffoli gates used in these two procedures - (i) one Pauli sampled in each iteration, (ii) group of commuting Paulis sampled in each iteration.

For our first set of experiments we examine the case of simulation of 4 and 6-qubit Heisenberg models.  The coefficients $J_x,J_y,J_z,d_h$ have been sampled from a 0 mean normal distribution with variance 1.
In Figure \ref{fig:10} we show that we achieve better scaling of $R_z/cR_z$ when multiple commuting Pauli operators are sampled and evolved in each iteration. In fact, the error also scales well with the number of iterations, i.e. we can achieve the same error in less number of iterations, or in another way, it is possible to achieve much lower error in the same time (iterations) when multiple operators are sampled. We calculate error as:
$$\text{Error} = \mathbb{E}_{\rho}(\|\chan_2(\rho)-\chan_1(\rho)\|_{l_2})$$
Where $\chan_1 = e^{iHt}\rho e^{-iHt}$, $\tau= t\cdot \left(\sum_j{h_{j}}\right) / N$ and $\|\cdot\|_{l_2}$ is the induced Euclidean norm on matrices and $\mathbb{E}_{\rho}$ is the Haar average over input states. We obtain $\chan_2$ through averaging $M$ random qDRIFT protocols, where $M$ varies from 100 to 3000 for our purposes.  These values are chosen to ensure that the sampling error is small at the scale of the plots generated. 
$$V_k = \prod_{j_{i_k}}e^{iH_{j_{i_k}}\tau} $$
$$\chan_2 = \frac{1}{M}\sum_{k=1}^{M}V_k \rho V_k^{\dag}$$
In our experiments $\rho$ is randomly drawn rather than chosen to maximize the diamond distance.  As a result, this does not give a tight upper bound on the error quantified by any induced channel norm.  Further, all evolution is done using $t=1$ and the groupings are hand optimized using counts given in Supplementary Method 5 (Appendix \ref{app:ham}). The data, tabulated in Figure \ref{fig:10}, shows that the number of iterations of the qDRIFT channel needed to simulate the dynamics to bound the error below a particular value, is reduced by a factor of $2.34$ and $2.8$ through the use of grouping commuting terms for the randomly chosen 4 and 6 qubit Heisenberg Hamiltonian respectively. The number of rotations is found to be reduced by a factor of roughly $2.34$ for the $4$ qubit ensemble but $1.8$ for the $6$ qubit case. This suggests that the groupings that we consider, while highly successful at reducing the number of iterations of qDRIFT needed, the number of gates per iteration increases from the $4$ to $6$ qubit examples.  This suggests that further computer aided optimization may be needed in order to see the full benefit of such groupings as we increase the size of models. 

Similar observations can be made for our second set of experiments where we simulate the Hamiltonian of $H_2$ and $LiH$ (with freezing in the STO-3G basis). The plots in Figure \ref{fig:11} show that in case of $H_2$, the number of iterations of the qDRIFT channel needed to simulate the dynamics to bound the error below a particular value, is reduced by a factor of $4$ through the use of grouping commuting terms. For $LiH$ this factor is nearly $2.1$. The number of rotations is found to be reduced by a factor of roughly $3.2$ for $H_2$ and $2$ for $LiH$. 

For all the experiments that we consider the Toffoli-pair gate count is comparable with the $R_z/cR_z$ count, so the Toffoli pairs do not contribute significantly to the overall T-count, as compared to the rotation gates. The number of gates depend on the diagonalizing circuits and the grouping into commuting Paulis. In this paper we have shown the set of results for the eigenbasis or grouping that were better among the options considered by us. In Supplementary Method 5 (Appendix \ref{app:ham}) we have explicitly mentioned the Hamiltonians, the groupings and given a short description of how we obtained the rotation and Toffoli costs.

All plots, code, and data can be found online in our public repository \url{https://github.com/SNIPRS/hamiltonian}. All code was written in Python. Our results were obtained partly with computing resources in the Cedar cluster of Compute Canada. Specifically, our code was run on an Intel(R) Xeon(R) E5-2683 v4 CPU at 2.10 GHz, utilizing 48 cores, up to 12GBs of RAM, and running Gentoo Linux 2.6. For the Heisenberg Hamiltonians, our results were obtained using 12 cores of an Intel(R) Core(TM) i5-12600K CPU at 3.6 GHz running Ubuntu 20.04.4 and up to 32GBs of RAM.

\begin{figure}
 \centering
 \includegraphics[width=\textwidth]{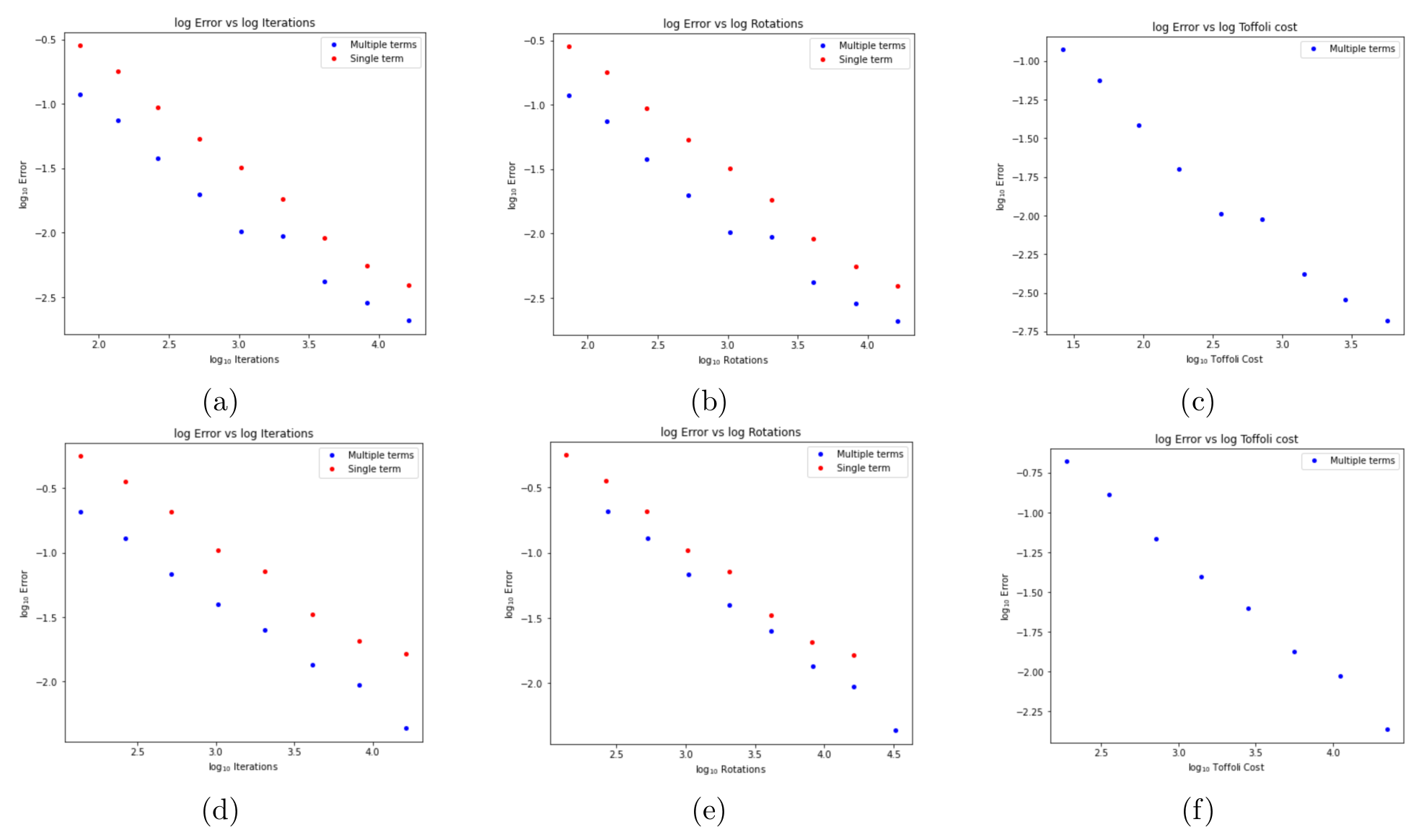}
 \caption{\textbf{Simulation of 4 and 6-qubit Heisenberg Hamiltonian : } Log-log plots showing number of iterations (a,d), $R_z/cR_z$ (b,e), Toffoli-pairs (c,f) as function of error, while simulating the 4 (a,b,c) and 6-qubit (d,e,f) quantum Heisenberg Hamiltonian ($H_H$) with qDRIFT. The red  and blue curve shows the variation when sampling single and multiple commuting Paulis per iteration, respectively. The Y-axis label in all plots is $\log_{10}\text{Error}$. The X-axis label of (a),(d) is $\log_{10}\text{Iterations}$, (b),(e) is $\log_{10}\text{Rotations}$ and (c),(f) is $\log_{10}\text{Toffoli Cost}$.}
\label{fig:10}
\end{figure}

\begin{figure}
 \centering
 \includegraphics[width=\textwidth]{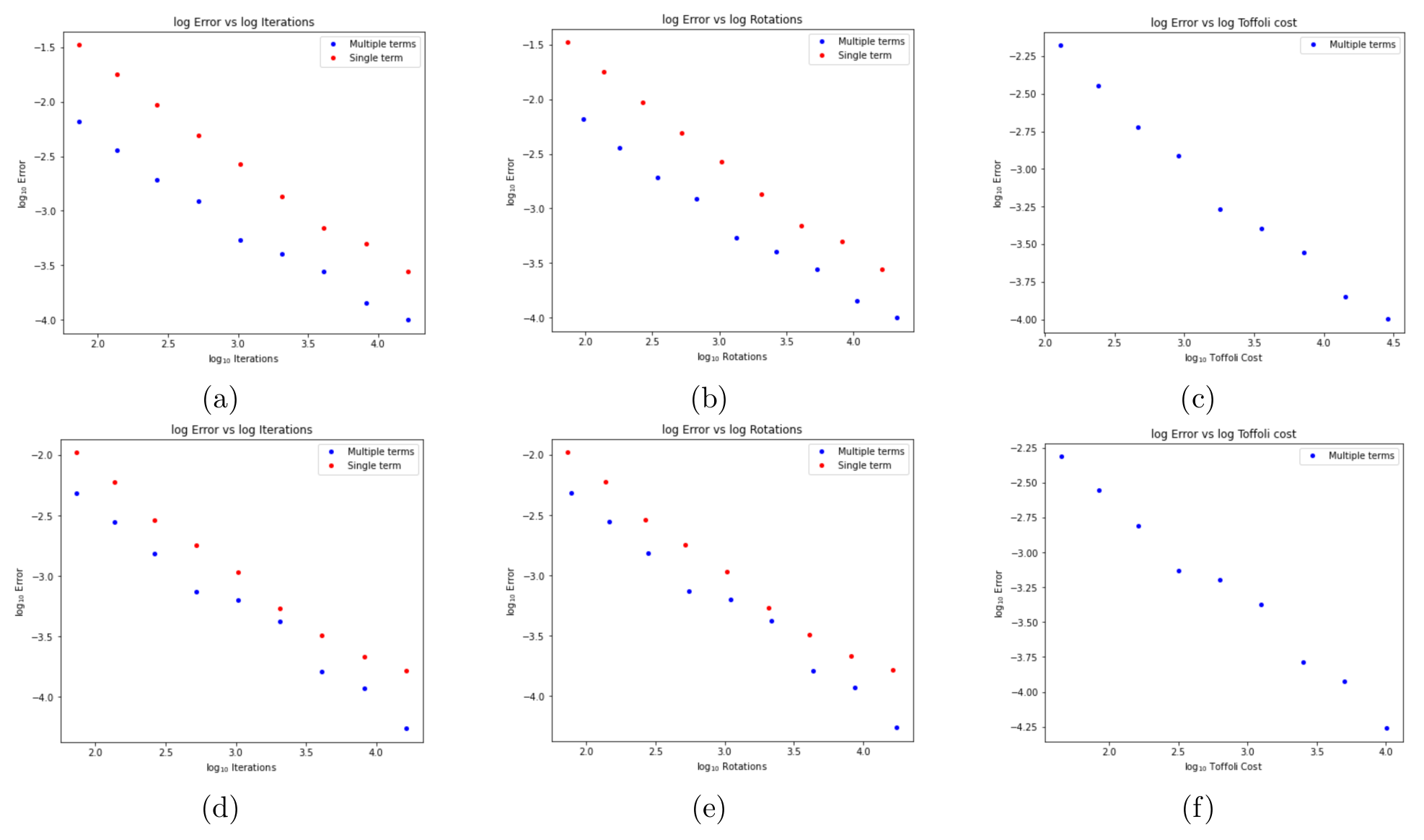}
 \caption{\textbf{Simulation of $H_2$ and LiH Hamiltonian : }
Log-log plots showing number of iterations (a,d), $R_z/cR_z$ (b,e), Toffoli-pairs (c,f) as function of error, while simulating the $H_2$ (a,b,c) and LiH (d,e,f) Hamiltonians with qDRIFT. The red  and blue dots show the variation when sampling single and multiple commuting Paulis per qDRIFT iteration, respectively. The Y-axis label in all plots is $\log_{10}\text{Error}$. The X-axis label of (a),(d) is $\log_{10}\text{Iterations}$, (b),(e) is $\log_{10}\text{Rotations}$ and (c),(f) is $\log_{10}\text{Toffoli Cost}$.}
\label{fig:11}
\end{figure}

\section{Discussion}
\label{sec:conclude}

In this paper, we have considered the problem of designing efficient quantum circuits for exponential of Hamiltonians that can be expressed as sum of Paulis. In contrast with most previous approaches, we synthesize circuit for a sum of exponentiated commuting Paulis, rather than concatenate circuits for each exponentiated Pauli. 
These resulting circuits are observed, for some parameter combinations, to require far fewer non-Clifford operations than the standard circuits.  We therefore propose an algorithm for greedily compiling a Trotter or qDRIFT simulation into a sequence of such simulations and observe that when multiple rotations are grouped we see at fixed error that a factor of roughly $1.8-3.2$ fewer rotations are needed to simulate 6 and 4-qubit Heisenberg models, $LiH$, $H_2$. Also, for simulation protocols like qDRIFT, it is possible to achieve a better performance, in the sense that the error accumulated per iteration is less if we sample multiple commuting Paulis. The overall non-Clifford gate cost of the entire protocol is also less.

There are a number of interesting avenues that are revealed by this work.  The first is that a more complete set of rules for compiling Hamiltonian terms into sets that can be easily exponentiated reveals the potential for more efficient simulation compilation of Hamiltonians.  These replacement rules, once identified, can be used inside a more systematic Hamiltonian compiler package that would allow more substantial optimizations of the Hamiltonian for the given simulation method.  This raises a second issue, while in this work we focus on the case of optimizing Trotter and related simulation methods, similar considerations could be performed for optimizing the prepare and select circuits used in LCU / qubitization simulation algorithms.  Such procedures are harder to optimize as the simulation algorithm does not factorize as nicely into independent simulations; however, the importance of these simulation methods makes the development of compilation strategies essential.

Finally, an important avenue hinted at by this work is the possibility that approximate unitary synthesis methods can be combined with quantum simulation routines to further reduce the cost.  If fermionic swaps are used, for example, simulation reduces to implementing a series of $4$-local Hamiltonians and optimal circuits can be in principle constructed for such Hamiltonians using existing approaches.  The computational overheads required for optimal (approximate) synthesis of these unitaries makes this a daunting task; however, if a sufficient lexicon of cheap unitaries are found for such simulation then it will not only lead to lower costs for quantum simulation using Trotter / qDRIFT: it will also unify Hamiltonian compilation with circuit synthesis into a single conceptual framework.  

\section*{Data Availability statement}
All plots and data can be found online in our public repository \url{https://github.com/SNIPRS/hamiltonian}. 

\section*{Code Availability statement}
All code can be found online in our public repository \url{https://github.com/SNIPRS/hamiltonian}. 

\section*{Acknowledgements}  
P.Mukhopadhyay wishes to thank NTT Research for their financial and technical support. This work was supported in part by Canada's NSERC. Research at IQC is supported in part by the Government of Canada through Innovation, Science and Economic Development Canada. This research was enabled in part by support provided by WestGrid (www.westgrid.ca) and Compute Canada Calcul Canada (www.computecanada.ca).
N.Wiebe acknowledges funding from the Google Quantum Research Program, the Natural Sciences and Engineering Research Council of Canada and this work on this project was supported by the U.S. Department of Energy, Office of Science, National Quantum Information Science Research Centers, Co-Design Center for Quantum Advantage under contract number DE-SC0012704.

\section*{Author contributions}

The ideas were given by P.Mukhopadhyay and N.Wiebe. The implementations were done by H.T.Zhang. All the authors contributed in the preparation of the manuscript.

\section*{Competing interests}

There are no competing interests.


\appendix

\section{Supplementary Note 1}
\label{app:prelim}

In this section we give some preliminary definitions and results used in the paper.
\paragraph{Cliffords and Paulis : }
The \emph{single qubit Pauli matrices} are as follows:
\begin{eqnarray}
 \X=\begin{bmatrix}
     0 & 1 \\
    1 & 0
    \end{bmatrix} \qquad  
 \Y=\begin{bmatrix}
     0 & -i \\
     i & 0
    \end{bmatrix} \qquad 
 \Z=\begin{bmatrix}
     1 & 0 \\
     0 & -1
    \end{bmatrix}\nonumber
\label{eqn:Pauli1}
\end{eqnarray}
Parenthesized subscripts are used to indicate qubits on which an operator acts For example, $\X_{(1)}=\X\otimes\id^{\otimes (n-1)}$ implies that Pauli $\X$ matrix acts on the first qubit and the remaining qubits are unchanged.

The \emph{$n$-qubit Pauli operators} are :
$
 \pauli_n=\{Q_1\otimes Q_2\otimes\ldots\otimes Q_n:Q_i\in\{\id,\X,\Y,\Z\} \}.
$

The \emph{single-qubit Clifford group} $\cliff_1$ is generated by the Hadamard and phase gates :
$
 \cliff_1=\braket{\had,\phase} 
 $
where
\begin{eqnarray}
 \had=\frac{1}{\sqrt{2}}\begin{bmatrix}
       1 & 1 \\
       1 & -1
      \end{bmatrix}\qquad 
 \phase=\begin{bmatrix}
       1 & 0 \\
       0 & i
      \end{bmatrix}\nonumber
\end{eqnarray}
When $n>1$ the \emph{$n$-qubit Clifford group} $\cliff_n$ is generated by these two gates (acting on any of the $n$ qubits) along with the two-qubit $\CNOT=\ket{0}\bra{0}\otimes\id+\ket{1}\bra{1}\otimes\X$ gate (acting on any pair of qubits). 

The Clifford group is special because of its relationship to the set of $n$-qubit Pauli operators. Cliffords map Paulis to Paulis, up to a possible phase of $-1$, i.e. for any $P\in\pauli_n$ and any $C\in\cliff_n$ we have
$
    CPC^{\dagger}=(-1)^bP'
$
for some $b\in\{0,1\}$ and $P'\in\pauli_n$. 

\paragraph{The group generated by Clifford and $\T$ gates : }

The group $\clifft_n$ is generated by the $n$-qubit Clifford group along with the $\T$ gate, where
\begin{eqnarray}
 \T=\begin{bmatrix}
     1 & 0 \\
     0 & e^{i\frac{\pi}{4}}
    \end{bmatrix}
\end{eqnarray}

Thus for a single qubit
$
 \clifft_1 = \braket{\had,\T}  
$
and for $n>1$ qubits
\begin{eqnarray}
 \clifft_n=\braket{\had_{(i)},\T_{(i)},\CNOT_{(i,j)}:i,j\in [n]}.
 \nonumber
\end{eqnarray}
It can be easily verified that $\clifft_n$ is a group, since the $\had$ and $\CNOT$ gates are their own inverses and $\T^{-1}=\T^7$. Here we note $\phase=\T^2$.

\paragraph{Diamond distance : }
We use $\|\cdot\|$ to denote the operator norm or Schatten-$\infty$ norm, which is equal to the largest singular value of an operator. We use $\|\cdot\|_1$ for the trace norm or Schatten 1-norm, defined as $\|Y\|_1=\tr\left[\sqrt{Y^{\dagger}Y}\right]$, which is equal to the sum of the singular values of an operator. Throughout, we use the diamond norm distance as a measure of error between two channels. The diamond distance is denoted
\begin{eqnarray}
 d_{\diamond}\left(\chan,\mathcal{N}\right)=\frac{1}{2}\|\chan-\mathcal{N}\|_{\diamond}
\end{eqnarray}
where $\|\cdot\|_{\diamond}$ is the diamond norm
\begin{eqnarray}
 \|\mathcal{P}\|_{\diamond}=\sup_{\rho:\|\rho\|_1=1}\|(\mathcal{P}\otimes\id)(\rho)\|_1,
\end{eqnarray}
where $\id$ acts on the same size Hilbert space as $\mathcal{P}$. The key properties of diamond norm, used in this paper are:
\begin{enumerate}
 \item \emph{Triangle inequality : } $\|\mathcal{A}\pm\mathcal{B}\|_{\diamond}\leq\|\mathcal{A}\|_{\diamond}+\|\mathcal{B}\|_{\diamond}$,
 
 \item \emph{Sub-multiplicativity : } $\|\mathcal{A}\mathcal{B}\|_{\diamond}\leq\|\mathcal{A}\|_{\diamond}\|\mathcal{B}\|_{\diamond}$ and consequently $\|\mathcal{A}^n\|_{\diamond}\leq\|\mathcal{A}\|_{\diamond}^n$.
\end{enumerate}
From the definition of diamond norm, it follows that if we apply the channels $\chan$ and $\mathcal{N}$ to the quantum state $\sigma$, we ahve the following.
\begin{eqnarray}
 d_{tr}\left(\chan(\sigma),\mathcal{N}(\sigma)\right)=\frac{1}{2}\|\chan(\sigma)-\mathcal{N}(\sigma)\|_1\leq d_{\diamond}(\chan,\mathcal{N})
\end{eqnarray}
The trace norm is an important quantity because it bounds the error in expectation values. If $M$ is an operator, then we have the following.
\begin{eqnarray}
 \|\tr[M\chan(\sigma)]-\tr[M\mathcal{N}(\sigma)]\|\leq 2\|M\|d_{tr}\left(\chan(\sigma),\mathcal{N}(\sigma)\right)\leq 2\|M\|d_{\diamond}(\chan,\mathcal{N})
\end{eqnarray}
If $M$ is a projection so that this represents a probability, then $\|M\|=1$. We see $\epsilon$ error in diamond distance ensures that the measurements statistics are correct up to additive error $2\epsilon$.

We have used the following inequalities in the paper.

\paragraph{Chebyshev's inequality : } Let $X$ be a random variable with finite expected value $\mu$ and finite non-zero variance $\sigma^2$. Then for any real number $k>0$,
\begin{eqnarray}
 \Pr\left[|X-\mu|\geq k\sigma\right]\leq \frac{1}{k^2}. \label{eqn:chebyshev}
\end{eqnarray}

\paragraph{Markov's inequality : } If $X$ is a non-negative random variable and $a>0$, then the probability that $X$ is at least $a$ is as follows.
\begin{eqnarray}
  \Pr\left[X\geq a\right]\leq \frac{\E[X]}{a}    \label{eqn:markov}
\end{eqnarray}

\paragraph{Hoeffding's inequality : } Let $X_1,\ldots, X_n$ be independent random variables such that $a_a\leq X_i\leq b_i$ almost surely. Let $S_n=\sum_{i=1}^nX_i$ be the sum of these variables. Then for all $t>0$ we have the following.
\begin{eqnarray}
\Pr\left[\left|S_n-\E[S_n]\right|\geq t\right]&\leq&2\exp\left(-\frac{2t^2}{\sum_{i=1}^n(b_i-a_i)^2}\right)  \label{eqn:hoeffding1}   \\
 \Pr\left[S_n-\E[S_n]\geq t\right]&\leq&\exp\left(-\frac{2t^2}{\sum_{i=1}^n(b_i-a_i)^2}\right)  \label{eqn:hoeffding2}
\end{eqnarray}

\section{Supplementary Method 1}
\label{app:1overlap}

In this section we show how to derive a diagonalizing circuit when overlap is on 1 qubit. We first show an eigenbasis for the Paulis in $G_{1y}$ and then we prove that a unitary $W_{1y}$ diagonalizes these Paulis.

\begin{lemma}[\textbf{Eigenbasis for $G_{1y}$}]
For the Paulis in $G_{1y}$ the eigenvectors are of the following form.
\begin{eqnarray}
 \ket{v_{y}}&=&\left(-\ket{Q_10Q_2}+\ket{\conj{Q_1}0Q_2}+\ket{Q_10\conj{Q_2}}+\ket{\conj{Q_1}0\conj{Q_2}}\right) \nonumber\\
&&\pm\left(-\ket{\conj{Q_1}1Q_2}-\ket{Q_11Q_2}+\ket{\conj{Q_1}1\conj{Q_2}}-\ket{Q_11\conj{Q_2}}\right) \nonumber 
\end{eqnarray}
Specifically we have the following. Let $\alpha=\left(\sqrt{-1}\right)^{i+j+k+1}$ and $\alpha'=\left(\sqrt{-1}\right)^{a+b+c+1}$.
\begin{eqnarray}
 P_iP_jP_k Y\id\id\id\ket{v_{y}}&=& \pm\alpha (-1)^{iq_1+jq_2+kq_3}\ket{v_{y}}\quad\text{ and }\quad \id\id\id Y P_aP_bP_c\ket{v_{y}}=\pm\alpha' (-1)^{aq_5+bq_6+cq_7}\ket{v_{y}} \nonumber
\end{eqnarray}
 \label{app:lem:ebasisGy}
\end{lemma}

\begin{proof}
 We have the following.
 \begin{eqnarray}
  P_iP_jP_k Y\id\id\id\ket{v_{y}}&=&\alpha (-1)^{iq_1+jq_2+kq_3} [\left(-\ket{\conj{Q_1}1Q_2}+(-1)^{i+j+k}\ket{Q_11Q_2}+\ket{\conj{Q_1}1\conj{Q_2}}+(-1)^{i+j+k}\ket{Q_11\conj{Q_2}} \right)   \nonumber   \\
  &&\pm \left(-(-1)^{i+j+k+1}\ket{Q_10Q_2}+\ket{\conj{Q_1}0Q_2}+(-1)^{i+j+k+1}\ket{Q_10\conj{Q_2}}+\ket{\conj{Q_1}0\conj{Q_2}}  \right)   ]  \nonumber   \\
  &=&\pm\alpha (-1)^{iq_1+jq_2+kq_3}\ket{v_{y}}  \quad [\because i+j+k\equiv 1\mod 2] \nonumber
\end{eqnarray}
Similarly we can prove the following.
 \begin{eqnarray}
 \id\id\id Y P_aP_bP_c\ket{v_{y}}&=&\alpha' (-1)^{aq_5+bq_6+cq_7} [\left(-\ket{Q_11\conj{Q_2}}+\ket{\conj{Q_1}1\conj{Q_2}}+(-1)^{a+b+c}\ket{Q_11Q_2}+(-1)^{a+b+c}\ket{\conj{Q_1}1Q_2} \right)   \nonumber   \\
  &&\pm \left(+\ket{\conj{Q_1}0\conj{Q_2}}+\ket{Q_10\conj{Q_2}}+(-1)^{a+b+c+1}\ket{\conj{Q_1}0Q_2}-(-1)^{a+b+c+1}\ket{Q_10Q_2}  \right)   ]  \nonumber   \\
  &=&\pm\alpha' (-1)^{aq_1+bq_2+cq_3}\ket{v_{y}}  \quad [\because a+b+c \equiv 1\mod 2] \nonumber
\end{eqnarray}
It is easy to check the orthogonality and the number of eigenvectors as $2^3\cdot 2\cdot 2^3 = 2^7$.
\end{proof}

Let $W_{1y}$ be the unitary consisting of the following sequence of gates. The rightmost one is the first to be applied. With a slight abuse of notation we denote $CNOT_{(4,1)}CNOT_{(4,2)}CNOT_{(4,3)}$ by $CNOT_{(4,I)}$ and $CNOT_{(4,5)}CNOT_{(4,6)}CNOT_{(4,7)}$ by $CNOT_{(4,II)}$.
\begin{eqnarray}
 W_{1y}=CNOT_{(4,I)}H_{(4)}Z_{(4)}CNOT_{(4,I)}CNOT_{(4,II)}H_{(4)}CNOT_{(4,I)}    \nonumber
\end{eqnarray}

\begin{theorem}
For each $i,j,k,l,a,b,c\in\intg_2$, such that $P_iP_jP_kY\id\id\id, \id\id\id YP_aP_bP_c\in G_{1y}$ we have the following.
 \begin{eqnarray}
  \sqrt{-1}^{i+j+k+1}W_{1y}\left(Z_{(1)}^iZ_{(2)}^jZ_{(3)}^kZ_{(4)}\id\id\id\right)W_{1y}^{\dagger}=P_iP_jP_k Y\id\id\id \nonumber \\
  \text{ and } \sqrt{-1}^{a+b+c+1}W_{1y}\left(\id\id\id Z_{(4)}Z_{(5)}^iZ_{(6)}^jZ_{(7)}^k\right)W_{1y}^{\dagger}=\id\id\id Y P_aP_bP_c \nonumber 
 \end{eqnarray}
 \label{app:thm:diagGy}
\end{theorem}

\begin{proof}
We prove the theorem by showing that the two operators have equivalent actions on the eigenstates of the Paulis in $G_{1y}$. Let
 \begin{eqnarray}
  \ket{v_{y+}}&=&\left(-\ket{Q_10Q_2}+\ket{\conj{Q_1}0Q_2}+\ket{Q_10\conj{Q_2}}+\ket{\conj{Q_1}0\conj{Q_2}}\right)    \nonumber \\
 &&+\left(-\ket{\conj{Q_1}1Q_2}-\ket{Q_11Q_2}+\ket{\conj{Q_1}1\conj{Q_2}}-\ket{Q_11\conj{Q_2}}\right)  \nonumber \\
 \text{ and } \ket{v_{y-}}&=&\left(-\ket{Q_10Q_2}+\ket{\conj{Q_1}0Q_2}+\ket{Q_10\conj{Q_2}}+\ket{\conj{Q_1}0\conj{Q_2}}\right)    \nonumber \\
 &&-\left(-\ket{\conj{Q_1}1Q_2}-\ket{Q_11Q_2}+\ket{\conj{Q_1}1\conj{Q_2}}-\ket{Q_11\conj{Q_2}}\right)  \nonumber
 \end{eqnarray}
  \textbf{Case - A :} First we consider $\ket{v_{y+}}$. We apply $W_{1y}^{\dagger}$ first. The state vector evolves as follows.
\begin{eqnarray}
 \ket{v_{y}+} &&\stackrel{CNOT_{(4,I)}}{\longrightarrow} \left(-\ket{Q_10Q_2}+\ket{\conj{Q_1}0Q_2}+\ket{Q_10\conj{Q_2}}+\ket{\conj{Q_1}0\conj{Q_2}}\right)    \nonumber \\
 &&+\left(-\ket{Q_11Q_2}-\ket{\conj{Q_1}1Q_2}+\ket{Q_11\conj{Q_2}}-\ket{\conj{Q_1}1\conj{Q_2}}\right)  \nonumber    \\
 &&\stackrel{H_{(4)}}{\rightarrow}\sqrt{2}\left(-\ket{Q_10Q_2}+\ket{\conj{Q_1}1Q_2}+\ket{Q_10\conj{Q_2}}+\ket{\conj{Q_1}1\conj{Q_2}}\right)    \nonumber    \\
 &&\stackrel{Z_{(4)}}{\rightarrow}\sqrt{2}\left(-\ket{Q_10Q_2}-\ket{\conj{Q_1}1Q_2}+\ket{Q_10\conj{Q_2}}-\ket{\conj{Q_1}1\conj{Q_2}}\right)    \nonumber    \\
 &&\stackrel{CNOT_{(4,I)}CNOT_{(4,II)}}{\longrightarrow}\sqrt{2}\left(-\ket{Q_10Q_2}-\ket{Q_11\conj{Q_2}}+\ket{Q_10\conj{Q_2}}-\ket{Q_11Q_2}\right)    \nonumber \\
 &&\stackrel{H_{(4)}}{\rightarrow}2\left(-\ket{Q_10Q_2}+\ket{Q_11\conj{Q_2}}\right)  \stackrel{CNOT_{(4,I)}}{\longrightarrow} 2\left(-\ket{Q_10Q_2}+\ket{\conj{Q_1}1\conj{Q_2}}\right) =\ket{v_{y}+}_1  \nonumber
\end{eqnarray}
We have the following after applying the tensor of Z operators.
\begin{eqnarray}
 \left(Z_{(1)}^i Z_{(2)}^jZ_{(3)}^kZ_{(4)}\id\id\id\right)\ket{v_{y+}}_1&=&(-1)^{iq_1+jq_2+kq_3}\ket{v_{y+}}_1 \nonumber \\
 \left(\id\id\id Z_{(4)} Z_{(5)}^aZ_{(6)}^bZ_{(7)}^c\right)\ket{v_{y+}}_1&=&(-1)^{aq_5+bq_6+cq_7}\ket{v_{y+}}_1 \nonumber 
\end{eqnarray}
Since we only have accumulation of different phase, so for the evolution of the state after applying $W_{1y}^{\dagger}$, in both cases it is enough to check the evolution of $\ket{v_{y+}}_1$. Since it is the same operators applied in reverse order we are not writing the states explicitly. 
\begin{eqnarray}
 \left(W_{1y}\left(Z_{(1)}^i Z_{(2)}^jZ_{(3)}^kZ_{(4)}\id\id\id\right)W_{1y}^{\dagger}\right)\ket{v_{y+}}=(-1)^{iq_1+jq_2+kq_3}\ket{v_{y+}}\nonumber \\
 \left(W_{1y}\left(\id\id\id Z_{(4)} Z_{(5)}^aZ_{(6)}^bZ_{(7)}^c\right)W_{1y}^{\dagger}\right)\ket{v_{y+}}=(-1)^{aq_5+bq_6+cq_7}\ket{v_{y+}}\nonumber 
\end{eqnarray}
Thus in this case the operators on the LHS and RHS have the same eigenvalues for the eigenvector $\ket{v_{y+}}$.

\textbf{Case - B :} By similar analysis we can prove that the operators on both the LHS and RHS have the same eigenvalues for the eigenvector $\ket{v_{y-}}$. 

This proves the theorem.
\end{proof}

\begin{lemma}[\textbf{Eigenbasis for $G_{1x}$}]
For the Paulis in $G_{1x}$ the eigenvectors are of the following form.
\begin{eqnarray}
 \ket{v_{\vect{x}}}&=&\left(\ket{Q_10Q_2}+(-1)^{x_1}\ket{Q_10\conj{Q_2}}+(-1)^{x_2}\ket{\conj{Q_1}0Q_2}+(-1)^{x_3}\ket{\conj{Q_1}0\conj{Q_2}}\right) \nonumber\\
&&\pm\left(\ket{Q_11Q_2}+(-1)^{x_1}\ket{Q_11\conj{Q_2}}+(-1)^{x_2}\ket{\conj{Q_1}1Q_2}+(-1)^{x_3}\ket{\conj{Q_1}1\conj{Q_2}}\right) \nonumber 
\end{eqnarray}
where $\vect{x}=(x_1,x_2,x_3)\subset\{0,1\}^3$ such that either $\|\vect{x}\|_1=0$ or $2$. Specifically we have the following.
\begin{eqnarray}
 P_iP_jP_kX\id\id\id\ket{v_{\vect{0}}}= \pm\alpha (-1)^{iq_1+jq_2+kq_3}\ket{v_{\vect{0}}} &\text{ and }& \id\id\id XP_aP_bP_c\ket{v_{\vect{0}}}=\pm\alpha' (-1)^{aq_5+bq_6+cq_7}\ket{v_{\vect{0}}} \nonumber\\
 P_iP_jP_k X\id\id\id\ket{v_{\vect{x}}}= \pm\alpha (-1)^{iq_1+jq_2+kq_3}\ket{v_{\vect{x}}} &\text{ and }& \id\id\id XP_aP_bP_c\ket{v_{\vect{x}}}=\mp\alpha' (-1)^{aq_5+bq_6+cq_7}\ket{v_{\vect{x}}} [\vect{x}=(1,0,1)]\nonumber\\
P_iP_jP_kX\id\id\id\ket{v_{\vect{x}}}= \mp\alpha (-1)^{iq_1+jq_2+kq_3}\ket{v_{\vect{x}}} &\text{ and }& \id\id\id XP_aP_bP_c\ket{v_{\vect{x}}}=\pm\alpha' (-1)^{aq_5+bq_6+cq_7}\ket{v_{\vect{x}}} [\vect{x}=(0,1,1)]\nonumber  \\
P_iP_jP_kX\id\id\id\ket{v_{\vect{x}}}= \mp\alpha (-1)^{iq_1+jq_2+kq_3}\ket{v_{\vect{x}}} &\text{ and }& \id\id\id XP_aP_bP_c\ket{v_{\vect{x}}}=\mp\alpha' (-1)^{aq_5+bq_6+cq_7}\ket{v_{\vect{x}}} [\vect{x}=(1,1,0)] \nonumber 
\end{eqnarray}
where $\alpha=\left(\sqrt{-1}\right)^{i+j+k}$ and $\alpha'=\left(\sqrt{-1}\right)^{a+b+c}$.
 \label{app:lem:ebasisGx}
\end{lemma}

The proof is similar to Lemma \ref{app:lem:ebasisGy}. There are $2^{3-1}$ choices for each of $Q_1$ and $Q_2$ that lead to independent vectors. There are $4$ choices for $\vect{x}$. Thus total number of eigenvectors as $2^2\cdot 2^2\cdot 2^2\cdot 2 = 2^7$.

\section{Supplementary Method 2}
\label{app:2overlap}

In this section we show how to derive a diagonalizing circuit when the overlap is on 2 qubits. The following lemma gives an eigenbasis for the Paulis in $G_{21}$ and can be proved in a manner similar to Lemma \ref{app:lem:ebasisGy}.
\begin{lemma}[\textbf{Eigenbasis for $G_{21}$}]
For the Paulis in $G_{21}$ the eigenvectors are of the following form.
\begin{eqnarray}
 \ket{v_{10}}&=&\left(-\ket{Q_100Q_3}+\ket{\conj{Q_1}00Q_3}+\ket{Q_100\conj{Q_3}}+\ket{\conj{Q_1}00\conj{Q_3}}\right)    \nonumber \\
 &&\pm\left(-\ket{\conj{Q_1}11Q_3}-\ket{Q_111Q_3}+\ket{\conj{Q_1}11\conj{Q_3}}-\ket{Q_111\conj{Q_3}}\right)  \label{eqn:v10} \\
 \ket{v_{11}}&=&\left(-\ket{Q_101Q_3}+\ket{\conj{Q_1}01Q_3}+\ket{Q_101\conj{Q_3}}+\ket{\conj{Q_1}01\conj{Q_3}}\right)    \nonumber \\
 &&\pm\left(\ket{\conj{Q_1}10Q_3}+\ket{Q_110Q_3}-\ket{\conj{Q_1}10\conj{Q_3}}+\ket{Q_110\conj{Q_3}}\right)  \label{eqn:v11} 
\end{eqnarray}
Specifically we have the following.
\begin{eqnarray}
 P_kP_lP_iP_j\id\id\ket{v_{10}}&=& \pm \alpha(-1)^{kq_1+lq_2}\ket{v_{10}}\quad\text{ and }\quad \id\id P_iP_jP_kP_l\ket{v_{10}}=\pm \alpha(-1)^{kq_5+lq_6}\ket{v_{10}} \nonumber   \\
 P_kP_lP_iP_j\id\id\ket{v_{11}}&=& \pm \alpha(-1)^{kq_1+lq_2+i}\ket{v_{11}}\quad\text{ and }\quad \id\id P_iP_jP_kP_l\ket{v_{11}}= \pm \alpha(-1)^{kq_5+lq_6+i}\ket{v_{11}} \nonumber
\end{eqnarray}
where $\alpha=\sqrt{-1}^{i+j+k+l}$.
 \label{app:lem:ebasisG1}
\end{lemma}

\begin{theorem}
For each $i,j,k,l\in\intg_2$, such that $P_kP_lP_iP_j\id\id, \id\id P_iP_jP_kP_l\in G_{21}$ we have the following.
 \begin{eqnarray}
  \sqrt{-1}^{i+j+k+l}W_1\left(Z_{(1)}^kZ_{(2)}^lZ_{(3)}Z_{(4)}^j\id\id\right)W_1^{\dagger}=P_kP_lP_iP_j\id\id \nonumber \\
  \text{ and } \sqrt{-1}^{i+j+k+l}W_1\left(\id\id Z_{(3)}Z_{(4)}^jZ_{(5)}^kZ_{(6)}^l\right)W_1^{\dagger}=\id\id P_iP_jP_kP_l \nonumber 
 \end{eqnarray}
 \label{app:thm:diag61}
\end{theorem}

\begin{proof}
  We prove this theorem by showing that the operators on LHS and RHS have the same eigenvalue-eigenvectors. The eigenvectors for the operators on RHS have been derived in Lemma \ref{app:lem:ebasisG1}. 
 
 \textbf{Case - I} Consider the eigenvector $\ket{v_{10}}$. Let
 \begin{eqnarray}
  \ket{v_{10+}}&=&\left(-\ket{Q_100Q_3}+\ket{\conj{Q_1}00Q_3}+\ket{Q_100\conj{Q_3}}+\ket{\conj{Q_1}00\conj{Q_3}}\right)    \nonumber \\
 &&+\left(-\ket{\conj{Q_1}11Q_3}-\ket{Q_111Q_3}+\ket{\conj{Q_1}11\conj{Q_3}}-\ket{Q_111\conj{Q_3}}\right)  \label{eqn:v10+} \\
 \text{ and } \ket{v_{10-}}&=&\left(-\ket{Q_100Q_3}+\ket{\conj{Q_1}00Q_3}+\ket{Q_100\conj{Q_3}}+\ket{\conj{Q_1}00\conj{Q_3}}\right)    \nonumber \\
 &&-\left(-\ket{\conj{Q_1}11Q_3}-\ket{Q_111Q_3}+\ket{\conj{Q_1}11\conj{Q_3}}-\ket{Q_111\conj{Q_3}}\right)  \label{eqn:v10-}
 \end{eqnarray}
\textbf{Case - IA} First we consider $\ket{v_{10+}}$. We apply $W_1^{\dagger}$ first. The state vector evolves as follows.
\begin{eqnarray}
 \ket{v_{10}+} &&\stackrel{CNOT_{(3,1)}CNOT_{(3,4)}}{\longrightarrow} \left(-\ket{Q_100Q_3}+\ket{\conj{Q_1}00Q_3}+\ket{Q_100\conj{Q_3}}+\ket{\conj{Q_1}00\conj{Q_3}}\right)    \nonumber \\
 &&+\left(-\ket{Q_110Q_3}-\ket{\conj{Q_1}10Q_3}+\ket{Q_110\conj{Q_3}}-\ket{\conj{Q_1}10\conj{Q_3}}\right)  \nonumber    \\
 &&\stackrel{H_{(3)}}{\rightarrow}\sqrt{2}\left(-\ket{Q_100Q_3}+\ket{\conj{Q_1}10Q_3}+\ket{Q_100\conj{Q_3}}+\ket{\conj{Q_1}10\conj{Q_3}}\right)    \nonumber    \\
 &&\stackrel{Z_{(3)}}{\rightarrow}\sqrt{2}\left(-\ket{Q_100Q_3}-\ket{\conj{Q_1}10Q_3}+\ket{Q_100\conj{Q_3}}-\ket{\conj{Q_1}10\conj{Q_3}}\right)    \nonumber    \\
 &&\stackrel{CNOT_{(3,1)}CNOT_{(3,5)}}{\longrightarrow}\sqrt{2}\left(-\ket{Q_100Q_3}-\ket{Q_110\conj{Q_3}}+\ket{Q_100\conj{Q_3}}-\ket{Q_110Q_3}\right)    \nonumber \\
 &&\stackrel{H_{(3)}}{\rightarrow}2\left(-\ket{Q_100Q_3}+\ket{Q_110\conj{Q_3}}\right)  \stackrel{CNOT_{(3,1)}}{\longrightarrow} 2\left(-\ket{Q_100Q_3}+\ket{\conj{Q_1}10\conj{Q_3}}\right) =\ket{v_{10}+}_1  \nonumber
\end{eqnarray}
We have the following after applying the tensor of Z operators. We note that $k+l=i+j=1$.
\begin{eqnarray}
 \left(Z_{(1)}^kZ_{(2)}^lZ_{(3)}Z_{(4)}^j\id\id\right)\ket{v_{10+}}_1&=&(-1)^{kq_1+lq_2}\ket{v_{10+}}_1    \nonumber \\
 \left(\id\id Z_{(3)}Z_{(4)}^jZ_{(5)}^kZ_{(6)}^l\right)\ket{v_{10+}}_1&=&(-1)^{kq_5+lq_6}\ket{v_{10+}}_1   \nonumber
\end{eqnarray}
Since we only have accumulation of different phase, so for the evolution of the state after applying $W_1$, in both cases it is enough to check the evolution of $\ket{v_{10+}}_1$.
\begin{eqnarray}
 \ket{v_{10+}}_1&&\stackrel{CNOT_{(3,1)}}{\longrightarrow} 2\left(-\ket{Q_100Q_3}+\ket{Q_110\conj{Q_3}}\right)  \stackrel{H_{(3)}}{\rightarrow}\sqrt{2}\left(-\ket{Q_100Q_3}-\ket{Q_110Q_3}+\ket{Q_100\conj{Q_3}}-\ket{Q_110\conj{Q_3}}\right)    \nonumber   \\
 &&\stackrel{CNOT_{(3,1)}CNOT_{(3,5)}}{\longrightarrow}\sqrt{2}\left(-\ket{Q_100Q_3}-\ket{\conj{Q_1}10\conj{Q_3}}+\ket{Q_100\conj{Q_3}}-\ket{\conj{Q_1}10Q_3}\right)    \nonumber   \\
 &&\stackrel{Z_{(3)}}{\rightarrow}\sqrt{2}\left(-\ket{Q_100Q_3}+\ket{\conj{Q_1}10\conj{Q_3}}+\ket{Q_100\conj{Q_3}}+\ket{\conj{Q_1}10Q_3}\right)    \nonumber\\
 &&\stackrel{H_{(3)}}{\rightarrow}-\ket{Q_100Q_3}-\ket{Q_110Q_3}+\ket{\conj{Q_1}00\conj{Q_3}}-\ket{\conj{Q_1}10\conj{Q_3}}+\ket{Q_100\conj{Q_3}}+\ket{Q_110\conj{Q_3}}+\ket{\conj{Q_1}00Q_3}-\ket{\conj{Q_1}10Q_3}  \nonumber\\
 &&\stackrel{CNOT_{(3,1)}CNOT_{(3,4)}}{\longrightarrow}\left(-\ket{Q_100Q_3}+\ket{\conj{Q_1}00\conj{Q_3}}+\ket{Q_100\conj{Q_3}}+\ket{\conj{Q_1}00Q_3} \right)   \nonumber \\
 &&+\left(-\ket{\conj{Q_1}11Q_3}-\ket{Q_111\conj{Q_3}}+\ket{\conj{Q_1}11\conj{Q_3}}-\ket{Q_111Q_3}   \right) = \ket{v_{10+}}    \nonumber
\end{eqnarray}
Thus we have the following.
\begin{eqnarray}
 \left(W_1\left(Z_{(1)}^kZ_{(2)}^lZ_{(3)}Z_{(4)}^j\id\id\right)W_1^{\dagger}\right)\ket{v_{10+}}=(-1)^{kq_1+lq_2}\ket{v_{10+}} \nonumber \\
  \text{ and } \left(W_1\left(\id\id Z_{(3)}Z_{(4)}^jZ_{(5)}^kZ_{(6)}^l\right)W_1^{\dagger}\right)\ket{v_{10+}}=(-1)^{kq_5+lq_6}\ket{v_{10+}} \nonumber 
\end{eqnarray}
Thus in this case the operators on the LHS and RHS have the same eigenvalues for the eigenvector $\ket{v_{10+}}$. By similar arguments we can prove that the operators on the LHS and RHS have the same eigenvalues for the eigenvector $\ket{v_{10-}}$ and $\ket{v_{11-}}$, hence proving the theorem. 

\end{proof}

\begin{lemma}[\textbf{Eigenbasis for $G_{20}$}]
For the Paulis in $G_{20}$ the eigenvectors are of the following form.
\begin{eqnarray}
 \ket{v_{A\vect{x}}}&=&\left(\ket{Q_1AQ_3}+(-1)^{x_1}\ket{Q_1A\conj{Q_3}}+(-1)^{x_2}\ket{\conj{Q_1}AQ_3}+(-1)^{x_3}\ket{\conj{Q_1}A\conj{Q_3}}\right)    \nonumber \\
 &&\pm\left(\ket{Q_1\conj{A}Q_3}+(-1)^{x_1}\ket{Q_1\conj{A}\conj{Q_3}}+(-1)^{x_2}\ket{\conj{Q_1}\conj{A}Q_3}+(-1)^{x_3}\ket{\conj{Q_1}\conj{A}\conj{Q_3}}\right)  \nonumber 
\end{eqnarray}
where $A\in\{00,01\}$ and $\vect{x}=(x_1,x_2,x_3)\subset\{0,1\}^3$ such that either $\|\vect{x}\|_1=0$ or $2$. 
Specifically we have the following. ($\alpha=\left(\sqrt{-1}\right)^{i+j+k+l}$)
\begin{eqnarray}
 P_kP_lP_iP_j\id\id\ket{v_{A\vect{0}}}= \pm\alpha (-1)^{kq_1+lq_2+iA_1}\ket{v_{A\vect{0}}}&\text{ and }& \id\id P_iP_jP_kP_l\ket{v_{A\vect{0}}}=\pm\alpha (-1)^{kq_5+lq_6+iA_1}\ket{v_{A\vect{0}}} \nonumber\\
 P_kP_lP_iP_j\id\id\ket{v_{A\vect{x}}}= \pm\alpha (-1)^{kq_1+lq_2+iA_1}\ket{v_{A\vect{x}}}&\text{ and }& \id\id P_iP_jP_kP_l\ket{v_{A\vect{x}}}=\mp\alpha (-1)^{kq_5+lq_6+iA_1}\ket{v_{A\vect{x}}} [\vect{x}=(1,0,1)]\nonumber\\
P_kP_lP_iP_j\id\id\ket{v_{A\vect{x}}}= \mp\alpha (-1)^{kq_1+lq_2+iA_1}\ket{v_{A\vect{x}}}&\text{ and }& \id\id P_iP_jP_kP_l\ket{v_{A\vect{x}}}=\pm\alpha (-1)^{kq_5+lq_6+iA_1}\ket{v_{A\vect{x}}} [\vect{x}=(0,1,1)]\nonumber  \\
P_kP_lP_iP_j\id\id\ket{v_{A\vect{x}}}= \mp\alpha (-1)^{kq_1+lq_2+iA_1}\ket{v_{A\vect{x}}}&\text{ and }& \id\id P_iP_jP_kP_l\ket{v_{A\vect{x}}}=\mp\alpha (-1)^{kq_5+lq_6+iA_1}\ket{v_{A\vect{x}}} [\vect{x}=(1,1,0)] \nonumber 
\end{eqnarray}
 \label{app:lem:ebasisG0}
\end{lemma}
The proof is similar to Lemma \ref{app:lem:ebasisGy}. 

\section{Supplementary Method 3}
\label{app:3overlap}

In this section we show how to derive a diagonalizing circuit when the overlap is on 3 qubits. The following lemma gives an eigenbasis for the Paulis in $G_{3y}$ and can be proved in a manner similar to Lemma \ref{app:lem:ebasisGy}.

\begin{lemma}[\textbf{Eigenbasis for $G_{3y}$}]
For the Paulis in $G_{3y}$ the eigenvectors are of the following form.
\begin{eqnarray}
 \ket{v_{0y}}&=&\left(-\ket{Q_1000Q_2}+\ket{\conj{Q_1}000Q_2}+\ket{Q_1000\conj{Q_2}}+\ket{\conj{Q_1}000\conj{Q_2}}\right)    \nonumber \\
 &&\pm\left(\ket{\conj{Q_1}111\conj{Q_2}}-\ket{Q_1111Q_2}-\ket{\conj{Q_1}111Q_2}-\ket{Q_1111\conj{Q_2}}\right)  \label{eqn:v0y} \\
 \ket{v_{1y}}&=&\left(-\ket{Q_1001Q_2}+\ket{\conj{Q_1}001Q_2}+\ket{Q_1001\conj{Q_2}}+\ket{\conj{Q_1}001\conj{Q_2}}\right)    \nonumber \\
 &&\pm\left(\ket{\conj{Q_1}110\conj{Q_2}}-\ket{Q_1110Q_2}-\ket{\conj{Q_1}110Q_2}-\ket{Q_1110\conj{Q_2}}\right)  \label{eqn:v1y} \\
 \ket{v_{2y}}&=&\left(-\ket{Q_1010Q_2}+\ket{\conj{Q_1}010Q_2}+\ket{Q_1010\conj{Q_2}}+\ket{\conj{Q_1}010\conj{Q_2}}\right)    \nonumber \\
 &&\pm\left(\ket{\conj{Q_1}101\conj{Q_2}}-\ket{Q_1101Q_2}-\ket{\conj{Q_1}101Q_2}-\ket{Q_1101\conj{Q_2}}\right)  \label{eqn:v2y} \\
 \ket{v_{3y}}&=&\left(-\ket{Q_1011Q_2}+\ket{\conj{Q_1}011Q_2}+\ket{Q_1011\conj{Q_2}}+\ket{\conj{Q_1}011\conj{Q_2}}\right)    \nonumber \\
 &&\pm\left(\ket{\conj{Q_1}100\conj{Q_2}}-\ket{Q_1100Q_2}-\ket{\conj{Q_1}100Q_2}-\ket{Q_1100\conj{Q_2}}\right)  \label{eqn:v3y} 
 \end{eqnarray}
 Specifically we have the following. Let $\alpha=\left(\sqrt{-1}\right)^{i+j+k+1}$.
\begin{eqnarray}
 YP_iP_jP_k\id\ket{v_{0y}}= \pm (-1)^{q_1}\alpha \ket{v_{0y}} &\text{ and }& \id P_iP_jP_kY\ket{v_{0y}}=\pm (-1)^{q_5}\alpha \ket{v_{0y}} \nonumber\\
 YP_iP_jP_k\id\ket{v_{1y}}= \pm (-1)^{q_1+k}\alpha \ket{v_{1y}} &\text{ and }& \id P_iP_jP_kY\ket{v_{1y}}=\pm (-1)^{q_5+k} \alpha \ket{v_{1y}} \nonumber   \\
 YP_iP_jP_k\id\ket{v_{2y}}= \pm (-1)^{q_1+j}\alpha \ket{v_{2y}} &\text{ and }& \id P_iP_jP_kY\ket{v_{2y}}=\pm (-1)^{q_5+j} \alpha \ket{v_{2y}} \nonumber   \\
  YP_iP_jP_k\id\ket{v_{3y}}= \mp (-1)^{q_1+i}\alpha \ket{v_{3y}} &\text{ and }& \id P_iP_jP_kY\ket{v_{3y}}=\mp (-1)^{q_5+i} \alpha \ket{v_{3y}} \nonumber
\end{eqnarray}
 \label{app:lem:ebasisG5y}
\end{lemma}

Let $W_{3y}$ be the unitary consisting of the following sequence of gates. The rightmost one is the first to be applied. With a slight abuse of notation we denote $CNOT_{(c,t_1)}CNOT_{(c,t_2)}CNOT_{(c,t_3)}\ldots$ by $CNOT_{(c;t_1,t_2,t_3,\ldots)}$ (multi-target CNOT).
\begin{eqnarray}
 W_{3y}=CNOT_{(2;1,3,4)}H_{(2)}Z_{(2)}CNOT_{(2;1,5)}H_{(2)}CNOT_{(2,1)}    \nonumber
\end{eqnarray}

\begin{theorem}
For each $i,j,k\in\intg_2$, such that $YP_iP_jP_k\id, \id P_iP_jP_kY  \in G_{3y}$ we have the following.
 \begin{eqnarray}
  \sqrt{-1}^{i+j+k+1}W_{3y}\left(Z_{(1)} Z_{(2)}Z_{(3)}^jZ_{(4)}^k\id\right)W_{3y}^{\dagger}=YP_iP_jP_k\id \nonumber \\
  \text{ and } \sqrt{-1}^{i+j+k+1}W_{3y}\left(\id Z_{(2)}Z_{(3)}^jZ_{(4)}^k Z_{(5)}\right)W_{3y}^{\dagger}=\id P_iP_jP_kY \nonumber 
 \end{eqnarray}
 \label{app:thm:diagG5y}
\end{theorem}

\begin{proof}
 We prove this theorem by showing that the operators on LHS and RHS have the same eigenvalue-eigenvectors. We prove for eigenvector $\ket{v_{3y}}$ and the remaining two can be proved similarly. Let
 \begin{eqnarray}
  \ket{v_{3y+}}&=&\left(-\ket{Q_1011Q_2}+\ket{\conj{Q_1}011Q_2}+\ket{Q_1011\conj{Q_2}}+\ket{\conj{Q_1}011\conj{Q_2}}\right)    \nonumber \\
 &&+\left(-\ket{\conj{Q_1}100Q_2}-\ket{Q_1100Q_2}+\ket{\conj{Q_1}100\conj{Q_2}}-\ket{Q_1100\conj{Q_2}}\right)  \label{eqn:v3y+} \\
 \text{ and } \ket{v_{3y-}}&=&\left(-\ket{Q_1011Q_2}+\ket{\conj{Q_1}011Q_2}+\ket{Q_1011\conj{Q_2}}+\ket{\conj{Q_1}011\conj{Q_2}}\right)    \nonumber \\
 &&-\left(-\ket{\conj{Q_1}100Q_2}-\ket{Q_1100Q_2}+\ket{\conj{Q_1}100\conj{Q_2}}-\ket{Q_1100\conj{Q_2}}\right)  \label{eqn:v3y-}
 \end{eqnarray}
 
 \textbf{Case - A} First we consider $\ket{v_{3y+}}$. We apply $W_{3y}^{\dagger}$ first. The state vector evolves as follows.
\begin{eqnarray}
 \ket{v_{3y}+} &&\stackrel{CNOT_{(2;1,3,4)}}{\longrightarrow} \left(-\ket{Q_1011Q_2}+\ket{\conj{Q_1}011Q_2}+\ket{Q_1011\conj{Q_2}}+\ket{\conj{Q_1}011\conj{Q_2}}\right)    \nonumber \\
 &&+\left(-\ket{Q_1111Q_2}-\ket{\conj{Q_1}111Q_2}+\ket{Q_1111\conj{Q_2}}-\ket{\conj{Q_1}111\conj{Q_2}}\right)  \nonumber    \\
 &&\stackrel{H_{(2)}}{\rightarrow}\sqrt{2}\left(-\ket{Q_1011Q_2}+\ket{\conj{Q_1}111Q_2}+\ket{Q_1011\conj{Q_2}}+\ket{\conj{Q_1}111\conj{Q_2}}\right)    \nonumber    \\
 &&\stackrel{Z_{(2)}}{\rightarrow}\sqrt{2}\left(-\ket{Q_1011Q_2}-\ket{\conj{Q_1}111Q_2}+\ket{Q_1011\conj{Q_2}}-\ket{\conj{Q_1}111\conj{Q_2}}\right)    \nonumber    \\
 &&\stackrel{CNOT_{(2;1,5)}}{\longrightarrow}\sqrt{2}\left(-\ket{Q_1011Q_2}-\ket{Q_1111\conj{Q_2}}+\ket{Q_1011\conj{Q_2}}-\ket{Q_1111Q_2}\right)    \nonumber \\
 &&\stackrel{H_{(2)}}{\rightarrow}2\left(-\ket{Q_1011Q_2}+\ket{Q_1111\conj{Q_2}}\right)  \stackrel{CNOT_{(2,1)}}{\longrightarrow} 2\left(-\ket{Q_1011Q_2}+\ket{\conj{Q_1}111\conj{Q_2}}\right) =\ket{v_{3y}+}_1  \nonumber
\end{eqnarray}
We have the following after applying the tensor of Z operators.
\begin{eqnarray}
 \left(Z_{(1)}Z_{(2)}Z_{(3)}^jZ_{(4)}^k\id\right)\ket{v_{3y+}}_1&=&(-1)^{q_1+j+k}\ket{v_{3y+}}_1    \nonumber \\
 \left(\id Z_{(2)}Z_{(3)}^jZ_{(4)}^kZ_{(5)}\right)\ket{v_{3y+}}_1&=&(-1)^{q_5+j+k}\ket{v_{3y+}}_1   \nonumber
\end{eqnarray}
Since we only have accumulation of different phase, so for the evolution of the state after applying $W_{3y}^{\dagger}$, in both cases it is enough to check the evolution of $\ket{v_{3y+}}_1$. Since it is the same operators applied in reverse order we are not writing the states explicitly. Thus since $i+j+k\equiv 1\mod 2$, we have the following.
\begin{eqnarray}
 \left(W_{3y}\left(Z_{(1)}Z_{(2)}Z_{(3)}^jZ_{(4)}^k\id\right)W_{3y}^{\dagger}\right)\ket{v_{3y+}}=(-1)^{q_1+j+k+i+i}\ket{v_{3y+}}=-(-1)^{q_1+i}\ket{v_{3y+}} \nonumber \\
  \text{ and } \left(W_{3y}\left(\id Z_{(2)}Z_{(3)}^jZ_{(4)}^kZ_{(5)}^l\right)W_{3y}^{\dagger}\right)\ket{v_{3y+}}=(-1)^{q_5+j+k+i+i}\ket{v_{3y+}}=-(-1)^{q_5+i}\ket{v_{3y+}} \nonumber 
\end{eqnarray}
So in this case the operators on the LHS and RHS have the same eigenvalues for the eigenvector $\ket{v_{3y+}}$. By similar arguments, we can come to the same conclusion for eigenvector $\ket{v_{3y-}}$, and the remaining eigenvectors as well. This proves the theorem. 
  
\end{proof}

\begin{lemma}[\textbf{Eigenbasis for $G_{3x}$}]
For the Paulis in $G_{3x}$ the eigenvectors are of the following form.
\begin{eqnarray}
 \ket{v_{A\vect{x}}}&=&\left(\ket{Q_1AQ_3}+(-1)^{x_1}\ket{Q_1A\conj{Q_3}}+(-1)^{x_2}\ket{\conj{Q_1}AQ_3}+(-1)^{x_3}\ket{\conj{Q_1}A\conj{Q_3}}\right)    \nonumber \\
 &&\pm\left(\ket{Q_1\conj{A}Q_3}+(-1)^{x_1}\ket{Q_1\conj{A}\conj{Q_3}}+(-1)^{x_2}\ket{\conj{Q_1}\conj{A}Q_3}+(-1)^{x_3}\ket{\conj{Q_1}\conj{A}\conj{Q_3}}\right)  \nonumber 
 \end{eqnarray}
 where $A\in\{000,001,010,011\}$ and $\vect{x}=(x_1,x_2,x_3)\subset\{0,1\}^3$ such that either $\|\vect{x}\|_1=0$ or $2$. 
 Specifically we have the following. Let $\alpha=\left(\sqrt{-1}\right)^{i+j+k}$.

\begin{eqnarray}
  XP_iP_jP_k\id\ket{v_{A\vect{0}}}= \pm\alpha (-1)^{jA_1+kA_0}\ket{v_{A\vect{0}}} &\text{ and }& \id P_iP_jP_kX\ket{v_{A\vect{0}}}=\pm\alpha (-1)^{jA_1+kA_0}\ket{v_{A\vect{0}}} \nonumber\\
 XP_iP_jP_k\id\ket{v_{A\vect{x}}}= \pm\alpha (-1)^{jA_1+kA_0}\ket{v_{A\vect{x}}} &\text{ and }& \id P_iP_jP_kX\ket{v_{A\vect{x}}}=\mp\alpha (-1)^{jA_1+kA_0}\ket{v_{A\vect{x}}} [\vect{x}=(1,0,1)]\nonumber\\
XP_iP_jP_k\id\ket{v_{A\vect{x}}}= \mp\alpha (-1)^{jA_1+kA_0}\ket{v_{A\vect{x}}} &\text{ and }& \id P_iP_jP_kX\ket{v_{A\vect{x}}}=\pm\alpha (-1)^{jA_1+kA_0}\ket{v_{A\vect{x}}} [\vect{x}=(0,1,1)]\nonumber  \\
XP_iP_jP_k\id\ket{v_{A\vect{x}}}= \mp\alpha (-1)^{jA_1+kA_0}\ket{v_{A\vect{x}}} &\text{ and }& \id P_iP_jP_kX\ket{v_{A\vect{x}}}=\mp\alpha (-1)^{jA_1+kA_0}\ket{v_{A\vect{x}}} [\vect{x}=(1,1,0)] \nonumber 
\end{eqnarray}

 \label{app:lem:ebasisG5x}
\end{lemma}

The proof is similar to Lemma \ref{app:lem:ebasisGy}. 

\section{Supplementary Method 4}
\label{app:error}

\paragraph{Single Vs Multiple terms : } Here we prove Lemma 2.2 (\textbf{CHECK}). We recall that we have the Hamiltonian $H=\sum_{j=1}^Lh_jH_j=\sum_{j=1}^L\sum_{i_j=1}^{L_j}h_jP_{i_j}$, where $H_j=\sum_{i_j=1}^{L_j} P_{i_j}$ - sum over commuting Paulis and $M$ is the total number of Pauli operators.

In the first procedure, we sample $H_j$ independently with probability $q_j=\frac{h_j}{\sum_j h_j}$. In the second procedure, in each iteration we select one single Pauli operator $P_k$ independently with probability $p_k'=\frac{\sum_{j'}h_{j'}}{\sum_i h_i L_i}$, where in the numerator the sum is over all the commuting Pauli groups in which $P_k$ appears. Let $\lambda=\sum_jh_j$ and $\lambda'=\sum_jh_jL_j$. The following Liouvillian generating unitaries under Hamiltonian $H_j$ and $P_{i_j}$ have been defined.
\begin{eqnarray}
 \liou_j&=& i(H_j\rho-\rho H_j)\quad\text{and}\quad\liou_{i_j}=i(P_{i_j}\rho-\rho P_{i_j}).
\end{eqnarray}
Thus if $\liou=i(H\rho-\rho H)$, then
$
 \liou=\sum_{j=1}^Lh_j\liou_j=\sum_{j=1}^Lh_j\sum_{i_j=1}^{L_j}\liou_{i_j}.  
$
We define two channels $\chan_1=\sum_{j=1}^Lq_je^{\tau\liou_j}$ and $\chan_2=\sum_{j=1}^Lp_j\sum_{i_j=1}^{L_j}e^{\tau'\liou_{i_j}}$, where $p_j=\frac{h_j}{\lambda'}$, that evolves the superoperators $\liou_j$ and $\liou_{ij}$ for time interval $\tau=\frac{\lambda t}{N}$ and $\tau'=\frac{\lambda' t}{N}$ respectively. Here we note that for the second channel, for each Pauli $P_k$, we have expanded the sum $p_{k'}=\sum_{j'}\frac{h_{j'}}{\lambda'}$ to reflect the commuting groups in which it belongs. Thus $\sum_{k=1}^Mp_{k'}=\sum_{j=1}^L\sum_{i_j=1}^{L_j}p_j$.
\begin{lemma}
$$
    \|\chan_2-\chan_1\|_{\diamond}\leq\frac{4t^2\lambda'^2}{N^2}
$$
 \label{app:lem:multErr}
\end{lemma}

\begin{proof}

\begin{eqnarray}
 \chan_1&=&\sum_{j=1}^Lq_je^{\tau\liou_j}=\sum_{j=1}^L\frac{h_j}{\lambda}e^{\tau\liou_j} \nonumber \\
 &=&\id+\sum_{j=1}^L\frac{h_j}{\lambda}\tau\liou_j+\sum_{j=1}^L\frac{h_j}{\lambda}\left(\sum_{n=2}^{\infty}\frac{\tau^n\liou_j^n}{n!}  \right) \nonumber \\
 &=&\id+\frac{\tau\liou}{\lambda}+\sum_{j=1}^L\frac{h_j}{\lambda}\left(\sum_{n=2}^{\infty}\frac{\tau^n(\sum_{i_j=1}^{L_j}\liou_{i_j})^n}{n!}  \right) 
 \end{eqnarray}
 \begin{eqnarray}
 \chan_2&=&\sum_{j=1}^Lp_j\sum_{i_j=1}^{L_j}e^{\tau'\liou_{i_j}}=\sum_{j=1}^L\frac{h_j}{\lambda'}\sum_{i_j=1}^{L_j}e^{\tau'\liou_{i_j}} \nonumber \\
 &=&\id+\frac{\tau'\liou}{\lambda'}+\sum_{j=1}^L\frac{h_j}{\lambda'}\sum_{i_j=1}^{L_j}\left(\sum_{n=2}^{\infty}\frac{\tau'^n\liou_{i_j}^n}{n!} \right)
\end{eqnarray}
We have $\frac{\tau}{\lambda}=\frac{\tau'}{\lambda'}=\frac{t}{N}$ and $\frac{h_j}{\lambda'}=\frac{h_j}{\lambda}\frac{\lambda}{\lambda'}=\frac{\sum_j h_j}{\sum_j h_jL_j}\frac{h_j}{\lambda}$. The error can be expressed as follows.
\begin{eqnarray}
 &&\|(\chanu-\chan_1)-(\chanu-\chan_2)\|_{\diamond} \nonumber   \\
 &=&\|\chan_2-\chan_1\|_{\diamond}  \nonumber \\
 &=&\Bigg\|\sum_{j=1}^L\frac{h_j}{\lambda}\sum_{n=2}^{\infty}\left(\frac{\tau^n(\sum_{i_j=1}^{L_j}\liou_{i_j})^n}{n!}-\frac{\lambda}{\lambda'}\sum_{i_j=1}^{L_j}\frac{\tau'^n\liou_{i_j}^n}{n!}  \right)\Bigg\|_{\diamond}    \nonumber \\
 &\leq& \sum_{j=1}^L\frac{h_j}{\lambda}\sum_{n=2}^{\infty}\Bigg\|\frac{\tau^n(\sum_{i_j}\liou_{i_j})^n}{n!}-\frac{\lambda}{\lambda'}\sum_{i_j}\left(\frac{\lambda'}{\lambda}\right)^n\frac{\tau^n\liou_{i_j}^n}{n!}\Bigg\|_{\diamond}    \nonumber \\
 &=&\sum_j\frac{h_j}{\lambda}\sum_{n=2}^{\infty}\frac{\tau^n}{n!}\Bigg\|\left(\sum_{i_j}\liou_{i_j}\right)^n-\sum_{i_j}\left(\frac{\lambda'}{\lambda}\right)^{n-1}\liou_{i_j}^n\Bigg\|_{\diamond} \nonumber \\
 &\leq&\sum_{n=2}^{\infty}\frac{\tau^n}{n!}\left(\sum_j\frac{h_j}{\lambda}\left(\sum_{i_j}\|\liou_{i_j}\|_{\diamond}\right)^n + \sum_j\frac{h_j}{\lambda}\sum_{i_j}\left(\frac{\lambda'}{\lambda}\right)^{n-1}\|\liou_{i_j}\|_{\diamond}^n \right)  \nonumber \\
 &\qquad& [\text{Triangle inequality and Sub-multiplicativity}]  \nonumber
\end{eqnarray}
Since $\liou_{i_j}$ are Liouvillian derived from Pauli operators, which are trace-preserving, so $\|\liou_{i_j}\|_{\diamond}\leq 2$. Thus we have the following.
\begin{eqnarray}
 &&\|\chan_2-\chan_1\|_{\diamond}  \nonumber \\
 &\leq&\sum_{n=2}^{\infty}\frac{2^n\tau^n}{n!}\left(\sum_j\frac{h_j}{\lambda}L_j^n + \sum_j\frac{h_j}{\lambda}L_j\left(\frac{\lambda'}{\lambda}\right)^{n-1} \right)  \nonumber \\
 &=&\sum_{n=2}^{\infty}\frac{2^n\tau^n}{n!}\sum_j\frac{h_jL_j}{\lambda}\left(L_j^{n-1} + \left(\frac{\lambda'}{\lambda}\right)^{n-1} \right)  \\
 &\leq& \sum_{n=2}^{\infty}\frac{2^n\tau^n}{n!}\sum_j\frac{h_jL_j}{\lambda}2\left(\frac{\lambda'}{\lambda}\right)^{n-1}  \nonumber \\
 &=&2\sum_{n=2}^{\infty}\frac{2^n\tau^n}{n!}\left(\frac{\lambda'}{\lambda}\right)^n \nonumber\\
 &\leq&2\frac{4t^2\lambda'^2}{2N^2}e^{2\lambda' t/N} \quad [\because \sum_{n=2}^{\infty}\frac{x^n}{n!}\leq \frac{x^2}{2}e^x \cite{2018_CMNetal}] \nonumber  \\
 &\approx&\frac{4t^2\lambda'^2}{N^2}
\end{eqnarray}

\end{proof}

\paragraph{Truncating Hamiltonian : } Here we prove Lemma 2.1 \textbf{CHECK}. We recall that we have the following Hamiltonian.
\begin{eqnarray}
 H&=&\sum_{j=1}^Mw_jH_j+\delta H=\tilde{H}+\delta H    \qquad [\|H\|\leq 1] \nonumber
\end{eqnarray}
$\chanu_N(\rho)$ and $\tilde{\chan}(\rho)$ are the channels corresponding to the Hamiltonian $H$ and the one implemented by the protocol, respectively. $\epsilon_N$ is the error per iteration of qDRIFT.
\begin{lemma}
$$
\epsilon_N\leq\|\tilde{\chan}(\rho)-\chanu_N(\rho)\|_{\diamond}\leq \epsilon_{qDRIFT}+2\delta\sqrt{\epsilon_{qDRIFT}}
$$
where $\epsilon_{qDRIFT}\lessapprox \frac{2\lambda^2t^2}{N^2}$ and $\lambda = \sum_i |w_i|$.
 \label{app:lem:truncErr}
\end{lemma}

\begin{proof}
Let $\tilde{U}=e^{it\tilde{H}}$ and $U=e^{itH}$. The channel corresponding to the Hamiltonian $H$ and $\tilde{H}$ are
\begin{eqnarray}
 \chanu_N(\rho)&=&e^{itH/N}\rho e^{-itH/N} \qquad\text{and}\qquad \tilde{\chanu}_N(\rho)=e^{it\tilde{H}/N}\rho e^{-it\tilde{H}/N}\quad\text{respectively.} \nonumber
\end{eqnarray}
\begin{eqnarray}
 \epsilon_N\leq \|\tilde{\chan}(\rho)-\tilde{\chanu}_N(\rho)\|_{\diamond}+\|\chanu(\rho)-\tilde{\chanu}_N(\rho)\|_{\diamond}\leq \epsilon_{qDRIFT}+\epsilon'\qquad [\text{Triangle inequality}] 
\end{eqnarray}
Here $\epsilon_{qDRIFT}$ is the error inherent in the qDRIFT protocol which we can bound as follows \cite{2019_C}.
\begin{eqnarray}
 \epsilon_{qDRIFT}\lessapprox \frac{2\lambda^2t^2}{N^2} \qquad [\lambda = \sum_i |w_i|]
\end{eqnarray}
To bound $\epsilon'$ we make use of the Liouvillian representation of a unitary channel as follows.
\begin{eqnarray}
 e^{iHt}\rho e^{-iHt}&=&e^{t\liou}(\rho)=\sum_{n=0}^{\infty}\frac{t^n\liou^n(\rho)}{n!} \quad \text{where}\quad \liou(\rho)=i(H\rho-\rho H) \\
 e^{i\tilde{H}t}\rho e^{-i\tilde{H}t}&=&e^{t\tilde{\liou}}(\rho)=\sum_{n=0}^{\infty}\frac{t^n\tilde{\liou}^n(\rho)}{n!} \quad \text{where}\quad \tilde{\liou}(\rho)=i(\tilde{H}\rho-\rho \tilde{H})
\end{eqnarray}
Since $H=\tilde{H}+\delta H$, implying $\tilde{H}=(1-\delta)H$, so we have the following.
\begin{eqnarray}
 \tilde{\liou}(\rho)=i\left(\tilde{H}\rho-\rho\tilde{H}\right)=(1-\delta)\liou  \nonumber
\end{eqnarray}
We have $\|\tilde{\liou}(\rho)\|_{\diamond}\leq 2\|\tilde{H}\|\leq 2\lambda$ and $\|\liou(\rho)\|_{\diamond}=\|\frac{1}{1-\delta}\tilde{\liou}(\rho)\|_{\diamond}\leq\frac{2\lambda}{1-\delta}$. 
\begin{eqnarray}
 \|\chanu_N(\rho)-\tilde{\chanu}_N(\rho)\|_{\diamond}&=&\left\|\sum_{n=0}^{\infty}\frac{t^n\liou^n(\rho)}{n!N^n}-\sum_{n=0}^{\infty}\frac{t^n\tilde{\liou}^n(\rho)}{n!N^n}\right\|_{\diamond}\leq \sum_{n=1}^{\infty}\frac{t^n}{n!N^n}\|\liou^n(\rho)-\tilde{\liou}^n(\rho)\|_{\diamond}   \quad [\text{Triangle inequality}] \nonumber\\
 &=& \sum_{n=1}^{\infty}\frac{t^n}{n!N^n}\|\liou^n(\rho)-(1-\delta)^n\liou^n(\rho)\|_{\diamond}    \quad[\text{Variables are positive real}]    \nonumber \\
 &\leq&\sum_{n=1}^{\infty}\frac{t^n}{n!N^n}[1-(1-\delta)^n] \|\liou(\rho)\|_{\diamond}^n \quad [\text{Sub-multiplicativity}]   \nonumber \\
 &\leq&\frac{2\delta t\lambda}{(1-\delta)N}+\sum_{n=2}^{\infty}\frac{t^n}{n!N^n}\left(\delta-n\delta^2+\frac{n(n-1)}{2}\delta^3-\ldots\right)\left(\frac{2\lambda}{1-\delta}\right)^n  \nonumber   \\
 &\lessapprox&\frac{2\delta t\lambda}{(1-\delta)N}+\delta\sum_{n=2}^{\infty}\frac{nt^n}{n!N^n}\left(\frac{2\lambda}{1-\delta}\right)^n \quad[\text{For small enough }\delta] \nonumber   \\
 &=&\frac{2\delta t\lambda}{(1-\delta)N}+\frac{\delta 2t\lambda}{(1-\delta)N}\sum_{n=2}^{\infty}\frac{t^{n-1}}{(n-1)!N^{n-1}}\left(\frac{2\lambda}{1-\delta}\right)^{n-1}   \nonumber \\
 &=&\frac{2\delta t\lambda}{(1-\delta)N}+\frac{\delta 2t\lambda}{(1-\delta)N}\sum_{m=1}^{\infty}\frac{t^m}{m!N^m}\left(\frac{2\lambda}{1-\delta}\right)^m \nonumber \\
 &=&\frac{2\delta t\lambda}{(1-\delta)N}+\frac{\delta 2t\lambda}{(1-\delta)N}\left[\frac{2\lambda t}{(1-\delta)N}+\sum_{m=2}^{\infty}\frac{t^m}{m!N^m}\left(\frac{2\lambda}{1-\delta}\right)^m\right] \nonumber \\
 &\leq&\frac{2\delta t\lambda}{(1-\delta)N}+\frac{2\delta\cdot 2\lambda^2t^2}{(1-\delta)^2N^2}+\frac{2\delta t\lambda}{(1-\delta)N}\frac{2\lambda^2t^2}{N^2(1-\delta)^2} e^{\frac{2\lambda t}{N(1-\delta)}}  \quad [\because \sum_{n=2}^{\infty}\frac{x^n}{n!}\leq \frac{x^2}{2}e^x \cite{2018_CMNetal}] \nonumber \\
 &\lessapprox&\delta\left(\frac{2t\lambda}{N}+\frac{2t^2\lambda^2}{N^2}\right)\lessapprox 2\delta \sqrt{\epsilon_{qDRIFT}}
\end{eqnarray}
Thus in our case the error per repetition or segment is 
\begin{eqnarray}
 \epsilon_N\lessapprox \epsilon_{qDRIFT}+2\delta\sqrt{\epsilon_{qDRIFT}}
\end{eqnarray}
\end{proof}

\begin{table}[!ht]
    \centering
    \begin{tabular}{|c|c|c|c|c|}
    \hline
    Hamiltonian & Group & Coeff. & $\#cR_z/R_z$ & $\#$Toff. pair \\ \hline
    \multirow{7}{*}{$H_2$} &   ZIII, IZII & 0.1371 - 0.1303 & 1 & 0 \\ \cline{2-5}
        & -IIZI, -IIIZ & 0.1303 & 1 & 0 \\ \cline{2-5}
        & XYYX, YXXY, -YYXX, -XXYY & 0.0491 & 1 & 3 \\ \cline{2-5}
        & IIZZ & 0.1632 - 0.1554 & 1 & 0 \\ \cline{2-5}
        & ZZII & 0.1566 - 0.1554 & 1 & 0 \\ \cline{2-5}
        & ZIIZ, IZZI & 0.1554 & 1 & 1 \\ \cline{2-5}
        & IZIZ, ZIZI & 0.1062 & 1 & 2 \\ \hline\hline
    \multirow{10}{*}{$LiH$}  &  XZXZ, IXIX, IYIY, YZYZ & -0.0014 & 2 & 3 \\ \cline{2-5}
      & XYYX, -YYXX, -XXYY, YXXY & 0.0029 & 1 & 3 \\ \cline{2-5}
      & ZYZY, XIXI, ZXZX, YIYI & 0.0115 & 2 & 3 \\ \cline{2-5}
    &    IXZX, IYZY, XZXI, YZYI & 0.0129 & 2 & 3 \\ \cline{2-5}
     &   IIZZ, ZZII, ZIIZ, IZZI & 0.0571 & 1 & 1 \\ \cline{2-5}
      &  IIZZ & 0.0848 - 0.0571 & 1 & 0 \\ \cline{2-5}
       & ZZII & 0.1244 - 0.0571 & 1 & 0 \\ \cline{2-5}
    &    IZIZ, ZIZI & 0.0541 & 1 & 2 \\ \cline{2-5}
     &   -IIZI, -IIIZ & 0.0132 & 1 & 0 \\ \cline{2-5}
     &   IZII, ZIII & 0.1620 & 1 & 0 \\ \hline\hline
    \multirow{4}{*}{4-qubit Heisenberg} &    $\sum_{i=1}^4X_{(i)}X_{(i+1\mod 4)}$ & $J_x$ & 1 & 1 \\ \cline{2-5}
    &    $\sum_{i=1}^4Y_{(i)}Y_{(i+1\mod 4)}$ & $J_y$ & 1 & 1 \\ \cline{2-5}
    &    $\sum_{i=1}^4Z_{(i)}Z_{(i+1\mod 4)}$ & $J_z$ & 1 & 1 \\ \cline{2-5}
    &    $Z_{(i)}$ & $d_h$ & 1 & 0 \\ \hline\hline
    \multirow{4}{*}{6-qubit Heisenberg} &    $\sum_{i=1}^6X_{(i)}X_{(i+1\mod 6)}$ & $J_x$ & 2 & 3 \\ \cline{2-5}
    &    $\sum_{i=1}^6Y_{(i)}Y_{(i+1\mod 6)}$ & $J_y$ & 2 & 3 \\ \cline{2-5}
    &    $\sum_{i=1}^6Z_{(i)}Z_{(i+1\mod 6)}$ & $J_z$ & 2 & 3 \\ \cline{2-5}
    &    $Z_{(i)}$ & $d_h$ & 1 & 0 \\ \hline\hline
    \end{tabular}
    \caption{Supplementary Table 1 : The groups of commuting Paulis (Group), coefficients (Coeff.), rotation gate cost ($\#cR_zz/R_z$), Toffoli cost ($\#$Toff. pair) for the different Hamiltonians considered for implementation in this paper.}
    \label{tab:cost}
\end{table}

\section{Supplementary Method 5}
\label{app:ham}

In this section we give the Hamiltonians considered by us in Section 2. We also give the grouping into commuting Paulis, that gave us the plots in Figures \ref{fig:10} and \ref{fig:11}. The groups have been shown in brackets. Each such group itself forms a Hamiltonian, which we refer to as component Hamiltonian. For simulation we synthesize circuits for each such exponentiated component Hamiltonian. We have truncated the coefficients to 4th decimal point. First we give the Hamiltonian for $H_2$.

\begin{eqnarray}
H_{H_2}&=&-0.3276081896748093 IIII +0.13716572937099508 ZIII +0.13716572937099503 IZII \nonumber \\
&&-0.13036292057109106 IIZI -0.13036292057109106 IIIZ +0.04919764587136755 XYYX \nonumber \\
&&+0.04919764587136755 YXXY -0.04919764587136755 YYXX -0.04919764587136755 XXYY \nonumber \\
&&+0.16326768673564346 IIZZ +0.15660062488237947 ZZII +0.15542669077992832 ZIIZ \nonumber \\
&&+0.15542669077992832 IZZI +0.10622904490856075 IZIZ
+0.10622904490856075 ZIZI   \nonumber \\
&=&0.0492\Big(XYYX+YXXY-YYXX-XXYY\Big)+0.1554\Big(IIZZ+ZZII+ZIIZ+IZZI\Big)  \nonumber \\
&&+0.1062\Big(IZIZ+ZIZI\Big)    
+0.1372\Big(ZIII+IZII\Big)-0.1304\Big(IIZI+IIIZ\Big)+0.0022\Big(ZZII\Big)   \nonumber \\
&&+0.0079\Big(IIZZ\Big) -0.3276081896748093 IIII    \nonumber
\end{eqnarray}
The identity shift can be ignored since it can be trivially added or removed as a phase. The component Hamiltonian built by the first group of commuting Paulis is like Case III of the Hamiltonian in Section 2.3 (Equation 19). Solving for $h_1, h_2, h_3$ from the coefficients of the Paulis, we get $h_1=-0.0246$, $h_2=0.0246$ and $h_3=0$. Then the circuit is the one given in Figure 1d. The second component Hamiltonian is the 4-qubit translationally invariant Ising Hamiltonian (4-point circle), without the single-qubit Z terms. So
the circuit is obtained from Section 2.4 (Figure 6a). The circuits for the other exponentiated component Hamiltonians can be obtained from the algorithm or procedure discussed in Section 2.4, that works for arbitrary Hamiltonians. Some optimizations can be done to reduce the number of Toffolis. The rotation and Toffoli cost for each component Hamiltonian is equal to the number of $R_z/cR_z$ gates and Toffoli pairs respectively.

In a similar way we can synthesize circuit for the 4-qubit Hamiltonian (with freezing in the STO-3G basis) for LiH, given below.
\begin{eqnarray}
 H_{LiH} &=&-0.0013743761078958677 XZXZ -0.0013743761078958677 IXIX -0.0013743761078958677 IYIY \nonumber \\
&&-0.0013743761078958677 YZYZ +0.002932996440950227 XYYX -0.002932996440950227 YYXX \nonumber \\
&&-0.002932996440950227 XXYY +0.002932996440950227 YXXY +0.011536413200774975 ZYZY  \nonumber \\
&&+0.011536413200774975 XIXI +0.011536413200774975 ZXZX +0.011536413200774975 YIYI  \nonumber \\
&&+0.012910780273117489 IXZX +0.012910780273117489 IYZY
+0.012910780273117489 XZXI  \nonumber \\
&&+0.012910780273117489 YZYI +0.08479609543670981 IIZZ
+0.12444770133137588 ZZII   \nonumber \\
&&+0.05706344223424907 ZIIZ +0.05706344223424907 IZZI
+0.054130445793298836 IZIZ  \nonumber \\
&&+0.054130445793298836 ZIZI -0.013243698330265952 IIZI
-0.013243698330265966 IIIZ  \nonumber \\
&&+0.1619947538800418 IZII +0.1619947538800418 ZIII
-7.498946902010707 IIII \nonumber   \\
&=&0.0029\Big(XYYX-YYXX-XXYY+YXXY\Big)    -0.0014\Big(XZXZ+IXIX+IYIY+YZYZ\Big) \nonumber \\
&&+0.0115(ZYZY+XIXI+ZXZX+YIYI)+0.0129\Big(IXZX+IYZY+XZXI+YZYI\Big)  \nonumber \\
&&+0.0571\Big(IIZZ+ZZII+ZIIZ+IZZI\Big)+0.0277\Big(IIZZ\Big)+0.0673\Big(ZZII\Big)    \nonumber \\
&&+0.0541\Big(IZIZ+ZIZI\Big)-0.0132\Big(IIZI+IIIZ\Big)+0.1620\Big(IZII+ZIII\Big)    \nonumber \\
&&-7.498946902010707 IIII \nonumber   
\end{eqnarray}
In Section 2.4 we have illustrated how to synthesize a circuit for the Heisenberg Hamiltonian, given below, for a graph $G = (V,E)$
\begin{eqnarray}
 H_H&=&\sum_{(i,j)\in E}\left(J_xX_{(i)}X_{(j)}+J_yY_{(i)}Y_{(j)}+J_zZ_{(i)}Z_{(j)}\right)+\sum_{i\in V}d_hZ_{(i)}  \nonumber \\
 &=&\Big(\sum_{(i,j)\in E}J_xX_{(i)}X_{(j)}\Big)+\Big(\sum_{(i,j)\in E}J_yY_{(i)}Y_{(j)}\Big)+\Big(\sum_{(i,j)\in E}J_zZ_{(i)}Z_{(j)}\Big)+\sum_{i\in V}\Big(d_hZ_{(i)}  \Big)  \nonumber
 \end{eqnarray}
For our experiments we sampled the coefficients $J_x, J_y, J_z, d_h$ from a 0-mean normal distribution with variance 1. We used 4 and 6-point cyclic graphs $V = \{1,...,n\}$, $E = \{(1,2), ... , (n-1,n),(n,1)\}$. In Table \ref{tab:cost} we have summarized the groups of commuting Paulis and their corresponding coefficients, rotation and Toffoli gate cost for $H_2, LiH$, 4 and 6-qubit Heisenberg Hamiltonian.

\end{document}